%% file: general-derand-main-arxiv.tex
\documentclass[11pt]{article}

\title{\textbf{Component Stability in Low-Space Massively Parallel Computation}\thanks{This work is partially supported by a Weizmann-UK Making Connections Grant,  the Centre for Discrete Mathematics and its Applications (DIMAP), IBM Faculty Award, EPSRC award EP/V01305X/1, European Research Council (ERC) Grant No. 949083, and the European Union's Horizon 2020 programme under the Marie Sk{\l}odowska-Curie grant agreement No 754411.}}
	
\author{\textbf{Artur Czumaj}
\thanks{Department of Computer Science and Centre for Discrete Mathematics and its Applications (DIMAP), University of Warwick. Email: A.Czumaj@warwick.ac.uk.
}
	\and
\textbf{Peter Davies}
\thanks{Institute of Science and Technology Austria (IST Austria). Email: Peter.Davies@ist.ac.at.
}
	\and
\textbf{Merav Parter}
\thanks{Weizmann Institute of Science, Rehovot, Israel. Email: Merav.Parter@weizmann.ac.il.
}
}
	
\author{\textbf{Artur Czumaj} \\
    University of Warwick \\
    A.Czumaj@warwick.ac.uk
	\and
\textbf{Peter Davies} \\
    IST Austria \\
    Peter.Davies@ist.ac.at
	\and
\textbf{Merav Parter} \\
    Weizmann Institute of Science \\
    Merav.Parter@weizmann.ac.il
}

\usepackage{amsmath,amsthm}
\usepackage{amssymb}
\usepackage[normalem]{ulem}

\newtheorem{theorem}{Theorem}
\newtheorem{corollary}[theorem]{Corollary}
\newtheorem{lemma}[theorem]{Lemma}
\newtheorem{proposition}[theorem]{Proposition}

\newtheorem{definition}[theorem]{Definition}

\newtheorem{question}[theorem]{Question}
\newtheorem{claim}[theorem]{Claim}
\newtheorem{fact}[theorem]{Fact}
\usepackage{algorithm}
\usepackage{algpseudocode}

\newenvironment{proofs}{%
	\proof}{\endproof}

\usepackage{bm}
\usepackage{xcolor}
\usepackage{fullpage}
\usepackage{xspace}
\usepackage{paralist}

\usepackage{ifpdf}
\newcommand{\mydriver}{hypertex}
\ifpdf
\renewcommand{\mydriver}{pdftex}
\fi
\usepackage[breaklinks,\mydriver]{hyperref}
\usepackage{cleveref}

\usepackage{framed}
\usepackage[framemethod=tikz]{mdframed}

\newcommand{\congest}{\textsf{CONGEST}\xspace}
\newcommand{\congc}{\textsf{CONGESTED CLIQUE}\xspace}
\newcommand{\local}{\textsf{LOCAL}\xspace}
\newcommand{\LOCAL}{\local}
\newcommand{\MPC}{\textsf{MPC}\xspace}
\newcommand{\LCL}{\textsf{LCL}\xspace}

\newcommand{\Prob}[1]{\mathbf{Pr}\left[#1\right]}	
\newcommand{\nat}{\ensuremath{\mathbb{N}}}

\def\epsilon{\ensuremath{\varepsilon}}
\newcommand{\eps}{\ensuremath{\epsilon}}

\newcommand{\Vars}{\ensuremath{\mathcal V }}

\newcommand{\ents}{\ensuremath{\mathcal X }}

\newcommand{\MPCDetStable}{$\mathsf{S}$-$\mathsf{DetMPC}$}
\renewcommand{\MPCDetStable}{\textsf{S-DetMPC}\xspace}
\newcommand{\MPCDetNonStable}{$\mathsf{DetMPC}$}
\renewcommand{\MPCDetNonStable}{\textsf{DetMPC}\xspace}
\newcommand{\MPCRandStable}{$\mathsf{S}$-$\mathsf{RandMPC}$}
\renewcommand{\MPCRandStable}{\textsf{S-RandMPC}\xspace}
\newcommand{\MPCRandNonStable}{$\mathsf{RandMPC}$}
\renewcommand{\MPCRandNonStable}{\textsf{RandMPC}\xspace}
\newcommand{\poly}{\operatorname{\textrm{poly}}}
\newcommand{\polylog}{\operatorname{\textrm{polylog}}}
\newcommand{\polyloglog}{\operatorname{\textrm{polyloglog}}}
\newcommand{\cA}{\mathcal{A}}
\newcommand{\cP}{\mathcal{P}}

\newcommand{\sparam}{\ensuremath{\phi}\xspace}

\newcommand{\NO}{\textbf{NO}\xspace}
\newcommand{\YES}{\textbf{YES}\xspace}

\newcommand{\Artur}[1]{\footnote{{\sc\small \textcolor[rgb]{1.00,0.00,0.00}{\textbf{Artur:}}} \textcolor[rgb]{0.00,0.50,0.00}{\textsf{#1}}}}

\newcommand{\Peter}[1]{\footnote{{\sc\small \textcolor[rgb]{1.00,0.00,0.00}{\textbf{Peter:}}} \textcolor[rgb]{0.00,0.00,1.00}{#1}}}

\newcommand{\COMMENTED}[1]{{}}
\newcommand{\junk}[1]{\COMMENTED{#1}}
\newcommand{\hide}[1]{{}}

\newcount\shortyear\newcount\shorthour\newcount\shortminute
\shorthour=\time\divide\shorthour by 60\shortyear=\shorthour
\multiply\shortyear by 60\shortminute=\time\advance\shortminute by
-\shortyear \shortyear=\year\advance\shortyear by -1900

\def\zeit{\number\shorthour:\ifnum\shortminute<10 0\number\shortminute
	\else\number\shortminute\fi}

\usepackage{environ}
\NewEnviron{algo}{\medskip\InGrayMiddle{\mbox{}\\\small\BODY\mbox{}\\[-0.35in]}\medskip}
\NewEnviron{walgo}{\medskip\InGray{\mbox{}\\\small\BODY\mbox{}\\[-0.35in]}\medskip}

\begin{document}

\begin{titlepage}

\maketitle

\begin{abstract}
In this paper, we study the power and limitations of component-stable algorithms in the low-space model of \emph{Massively Parallel Computation (\MPC)}. Recently Ghaffari, Kuhn and Uitto (FOCS 2019) introduced the class of \emph{component-stable} low-space \MPC algorithms, which are, informally, defined as algorithms for which the outputs reported by the nodes in different connected components are required to be independent. This very natural notion was introduced to capture most (if not all) of the known efficient \MPC algorithms to date, and it was the first general class of \MPC algorithms for which one can show non-trivial conditional lower bounds.  In this paper we enhance the framework of component-stable algorithms and investigate its effect on the complexity of randomized and deterministic low-space \MPC. Our key contributions include:
	
\begin{itemize}
\item We revise and formalize the lifting approach of Ghaffari, Kuhn and Uitto. This requires a very delicate amendment of the notion of component stability, which allows us to fill in gaps in the earlier arguments.
\item We also extend the framework to obtain conditional lower bounds for deterministic algorithms and fine-grained lower bounds that depend on the maximum degree $\Delta$.
\item We demonstrate a collection of natural graph problems for which deterministic non-component-stable algorithms break the conditional lower bound obtained for component-stable algorithms. This implies that, in the context of deterministic algorithms, component-stable algorithms are conditionally weaker than the non-component-stable ones.
\item We also show that the restriction to component-stable algorithms has an impact in the randomized setting. We present a natural problem which can be solved in $O(1)$ rounds by a component-unstable \MPC algorithm, but requires $\Omega(\log\log^* n)$ rounds for any component-stable algorithm, conditioned on the connectivity conjecture.
\end{itemize}

Altogether our results imply that component-stability might limit the computational power of the low-space \MPC model, at least in certain contexts, paving the way for improved upper bounds that escape the conditional lower bound setting of Ghaffari, Kuhn, and Uitto.
\end{abstract}

\end{titlepage}




\section{Introduction}
\label{sec:intro}
\input{intro.tex}


\section{Revised Framework of Component Stability}
\label{sec:prelim}
\input{prelim.tex}



\section{Conditional \MPC Lower Bounds from \LOCAL}
\label{sec:lb-LOCAL->MPC}
\input{appendix-lift.tex}


\section{Separation between Stable and Unstable Deterministic \MPC}
\label{sec:stable-nonstable-det}
\input{stable-non-deter}


\section{Separation between Stable and Unstable Randomized \MPC}
\label{sec:stable-nonstable-rand}
\input{sepRandMPC.tex}
\section{General Non-Uniform and Non-Explicit Derandomization of \MPC Algorithms}
\label{sec:derand}
\input{general-derand-arg.tex}


\section{Conclusions}
\label{sec:conclusions}
\input{conclusions.tex}


	
	\newcommand{\Proc}{Proceedings of the~}
	\newcommand{\APPROX}{International Workshop on Approximation Algorithms for Combinatorial Optimization Problems}
	\newcommand{\CACM}{Commununication of the ACM}
	\newcommand{\CSR}{International Computer Science Symposium in Russia (CSR)}
	\newcommand{\DISC}{International Symposium on Distributed Computing (DISC)}
	\newcommand{\FOCS}{IEEE Symposium on Foundations of Computer Science (FOCS)}
	\newcommand{\ICALP}{Annual International Colloquium on Automata, Languages and Programming (ICALP)}
	\newcommand{\IPCO}{International Integer Programming and Combinatorial Optimization Conference (IPCO)}
	\newcommand{\IPL}{Information Processing Letters}
	\newcommand{\ISAAC}{International Symposium on Algorithms and Computation (ISAAC)}
	\newcommand{\JACM}{Journal of the ACM}
	\newcommand{\JCSS}{Journal of Computer and System Sciences}
	\newcommand{\NIPS}{Conference on Neural Information Processing Systems (NeurIPS)}
	\newcommand{\OPODIS}{International Conference on Principles of Distributed Systems (OPODIS)}
	\newcommand{\OSDI}{Conference on Symposium on Opearting Systems Design \& Implementation (OSDI)}
	\newcommand{\PODS}{ACM SIGMOD Symposium on Principles of Database Systems (PODS)}
	\newcommand{\PODC}{ACM Symposium on Principles of Distributed Computing (PODC)}
	\newcommand{\RSA}{Random Structures \& Algorithms}
	\newcommand{\SICOMP}{SIAM Journal on Computing}
	\newcommand{\SIJDM}{SIAM Journal on Discrete Mathematics}
	\newcommand{\SIROCCO}{International Colloquium on Structural Information and Communication Complexity (SIROCCO)}
	\newcommand{\SODA}{Annual ACM-SIAM Symposium on Discrete Algorithms (SODA)}
	\newcommand{\SPAA}{Annual ACM Symposium on Parallelism in Algorithms and Architectures (SPAA)}
	\newcommand{\STACS}{Annual Symposium on Theoretical Aspects of Computer Science (STACS)}
	\newcommand{\STOC}{Annual ACM Symposium on Theory of Computing (STOC)}
	\newcommand{\TALG}{ACM Transactions on Algorithms}
	\newcommand{\TCS}{Theoretical Computer Science}
	
\bibliographystyle{alpha}
\bibliography{references}





	

\end{document}

%% file: intro.tex


The central goal of this paper is to advance our understanding of the computational power of low-space algorithms in the \emph{Massively Parallel Computation (\MPC)} model. Our main focus is on the notion of \emph{component-stable} low-space \MPC algorithms introduced recently by Ghaffari, Kuhn and Uitto \cite{GKU19} as the first general class of \MPC algorithms for which non-trivial conditional lower bounds can be obtained. Roughly speaking, in this class of algorithms the output of nodes in different connected components are required to be independent. While this definition has been introduced to capture most (if not all) of the known \MPC algorithms to date, and the notion of component-stable algorithms seems quite natural and unlimited, we demonstrate its effect on the complexity of randomized and deterministic low-space \MPC. Our main finding is that the notion of component-stability as defined in \cite{GKU19} is rather fragile and needs to be studied with care, leading us to a revision of this framework to make it robust. Our amended framework of com\-ponent-stable algorithms allows us to fill in gaps in the earlier arguments and make it more applicable.
In particular, the revised setup enables us to extend the framework of (conditional) lower bounds from \cite{GKU19} for component-stable randomized algorithms relating \LOCAL algorithms and low-space \MPC algorithms: we demonstrate that it can be parameterized with respect to the maximum graph degree $\Delta$ and holds also for deterministic algorithms, thereby making the framework more broadly applicable and proving for the first time a host of conditional lower bounds for a number of deterministic low-space component-stable \MPC algorithms.


Next, we will show that for some natural problems there are low-space \emph{component-unstable} \MPC algorithms (both randomized and deterministic) that are significantly more powerful than their component-stable counterparts. So, rather than being a technical triviality, component-stability is in fact a significant restriction on the power of the low-space \MPC model.

\paragraph{Background.}
The rapid growth of massively parallel computation frameworks, such as MapReduce~\cite{DG04}, Hadoop~\cite{White12}, Dryad~\cite{IBYBF07}, or Spark~\cite{ZCFSS10} resulted in the need of active research for understanding the computational power of such systems. The \emph{Massively Parallel Computation (\MPC)} model, first introduced by Karloff et al.\ \cite{KSV10} (and later refined in \cite{ANOY14,BKS13,GSZ11}) has became the standard theoretical model of algorithmic study, as it provides a clean abstraction of these frameworks. Over the past years, this model has been receiving a major amount of interest by several independent communities in theory and beyond. In comparison to the classical PRAM model, the \MPC model allows for a lot of local computation (in principle, unbounded) and enabled it to capture a more ``coarse--grained'' and meaningful aspect of parallelism (see, e.g., \cite{ASSWZ18,BHH19,CLMMOS18,GGKMR18,GU19}).

In the \MPC model, there are $M$ machines and each of them has $S$ words of local space at its disposal. Initially, each machine receives its share of the input.
%
In the context of graph problems where the input is a collection $V$ of nodes and $E$ of edges, $|V| = n$, $|E| = m$, the input is arbitrarily distributed among the machines (and so $S \cdot M \ge n + m$). 
In this model, the computation proceeds in synchronous \emph{rounds} in which each machine processes its local data and performs an arbitrary local computation on its data. At the end of each round, machines exchange messages. Each message is sent only to a single machine specified by the machine that is sending the message. All messages sent and received by each machine in each round, as well as the output have to fit into the machine's local space~$S$.

Our focus in this paper is on the \emph{low-space} setting where the local space of each machine is \emph{strongly sublinear} in the number of nodes, i.e., $S = n^{\sparam}$ for some
$\sparam \in (0,1)$. Our lower bounds will be against algorithms with any polynomial number of machines (i.e., $M = \poly(n)$), while our upper bounds (with the exception of the non-uniform general derandomization of Lemma \ref{lem:det-alg-many-machines}) will use at most $O(m+n^{1+\sparam})$ global space (i.e., $M = O(\frac{m}{n} + n^\sparam)$).

The low-space regime is particularly challenging due to the fact that a node's edges cannot necessarily be stored on a single machine, but rather are scattered over several machines. Nevertheless, for many classical graph problems, $\poly(\log n)$-round algorithms can be obtained, and recently we have also seen even \emph{sublogarithmic} solutions. Ghaffari and Uitto \cite{GU19} (see also \cite{Onak18}) presented a randomized graph sparsification technique resulting in $\widetilde{O}(\sqrt{\log\Delta})$ round algorithms for maximal matching and MIS, where $\Delta$ is the maximum degree. This should be compared, for example, with maximal matching algorithms with significantly more local space: Lattanzi et al.\ \cite{LMSV11} presented an $O(1/\eps)$-round randomized algorithm using $O(n^{1+\eps})$ local space and Behnezhad et al.\ \cite{BHH19} gave an $O(\log\log n)$-round randomized algorithm using $O(n)$ local space (see also \cite{ABBMS17,CLMMOS18,GGKMR18}). For the problem of $(\Delta+1)$-vertex coloring Chang et al. \cite{CFGUZ19} showed a randomized low-space \MPC algorithm that works in $O(\log\log \log n)$ rounds, when combined with the network decomposition result of Rohzo{\v{n}} and Ghaffari \cite{RG20}.

While we have seen some major advances in the design of low-space \MPC algorithms, no (unconditional) hardness results are known for any of the above problems in the low-space \MPC setting. A seminal work by Roughgarden et al.\ \cite{RVW18} provides an explanation for the lack of such lower bound results. They showed that obtaining any unconditional lower bound in the low-space \MPC (for algorithms with an arbitrarily polynomial number of machines) setting ultimately leads to a breakthrough result in circuit complexity, namely that $\textsf{NC}^1 \subsetneqq \textsf{P}$. This work has opened up a new avenue towards proving \emph{conditional} hardness results that are based on the widely believed \emph{connectivity conjecture}. This conjecture (extensively used in our current paper) states that there is no $o(\log n)$-round (randomized) low-space \MPC algorithm (even using any polynomial global space) for distinguishing between the input graph $G$ being an $n$-length cycle and two $\frac{n}{2}$-length cycles. 

The first conditional lower bounds in the low-space \MPC setting were presented by a recent insightful paper of Ghaffari, Kuhn and Uitto \cite{GKU19}. This work provides a collection of conditional hardness results for classical local problems by drawing a new connection between the round complexity of a given problem in the \LOCAL model \cite{Linial92}, and its corresponding complexity in the low-space \MPC model. Unlike the low-space \MPC setting, for the \LOCAL model, arguably one of the most extensively studied model in distributed computing, there is a rich collection of (unconditional) lower bound results. To enjoy these \LOCAL lower bound results in our context, \cite{GKU19} presented a quite general technique that for many graph problems translates an $\Omega(r)$-round \LOCAL lower bound (with an additional requirement of using shared randomness) into an $\Omega(\log r)$-round lower bound in the low-space \MPC model \emph{conditioned on the connectivity conjecture}. This beautiful lifting argument is surprisingly quite general, capturing the classical lower bounds for problems like MIS, maximal matching \cite{KMW06}, LLL (Lov{\'a}sz Local Lemma), and sinkless orientation \cite{BFH+16}. For example, one very strong implication of their technique is that conditioned on the connectivity conjecture, it shows that there is no randomized low-space \MPC algorithm for (even approximate) maximal matching or MIS problems using $o(\log \log n)$ rounds.

The framework of Ghaffari, Kuhn and Uitto \cite{GKU19} has one main caveat, which at first glance appears to be quite negligible, a mere technicality. Their conditional lower bounds do not hold for \emph{any} algorithms but rather only for the special class of \emph{component-stable} \MPC algorithms. The key property of these algorithms is that the output of the nodes in one connected component is independent of other components. More formally, in component-stable algorithms, the output of each node $v$ is allowed to depend (deterministically) only on the node $v$ itself, the initial distribution, the ID assignment of the connected component of $v$, and on the shared randomness. The class of component-stable algorithms is indeed quite natural, and at least by the time of publication of \cite{GKU19}, it appeared to capture most, if not all, of the existing \MPC algorithms, as explicitly noted by the authors:
\begin{mdframed}[hidealllines=true,backgroundcolor=gray!15]
\begin{quote}
\cite{GKU19}
\emph{To the best of our knowledge, all known algorithms in the literature are component-stable or can easily be made component-stable with no asymptotic increase in the round complexity.}
\end{quote}
\end{mdframed}
In this view, it appeared that the restriction to component-stable algorithms is no more than a minor technicality rather than an actual limitation on the low-space \MPC model.

The first indication that component-stability might actually matter was provided by recent works \cite{CDP20a,CDP20b}, which present deterministic low-space \emph{component-unstable} \MPC algorithms for several classic graph problems, even though the validity of solutions to these problems depends only on local information. Specifically, by derandomizing a basic graph sparsification technique, one can obtain $O(\log \Delta+\log\log n)$-round deterministic low-space component-unstable \MPC algorithms for MIS, maximal matching, and $(\Delta+1)$-coloring. A key ingredient of these algorithms is a global agreement on a logarithmic length seed, to be used by all nodes in order to simulate their randomized decisions. This global seed selection involves coordination between all the nodes, regardless of their components, thus yielding component-unstable algorithms. The component-instability here seems to be quite inherent to the derandomization technique, and it is unclear whether any component-stable method could perform equally well.

\subsection{Our aims}
\label{subsec:aims}

In this paper \emph{we thoroughly investigate the concept of component-stability and its impact on randomized and deterministic low-space \MPC algorithms}.
%

Upon examining the notion of component-stability in detail and after attempts to broaden its applications, it becomes apparent that the concept is highly sensitive to the exact definition used, and that one must be very careful in specifying what information the outputs of component-stable algorithms may depend on. For example, we must precisely specify whether we allow component-stable algorithms' outputs to depend on the input size $n$, and we find that either choice here holds problematic implications for the current lower bounds and for the analysis due to Ghaffari et al. \cite{GKU19}.

This raises the first main question of our work:
%

\begin{mdframed}[hidealllines=true,backgroundcolor=gray!25]\vspace{-8pt}
\begin{question}
\label{q:definition}
Can we revise the lifting framework and amend the definition of component-stability which both captures a wide array of algorithms, and also allows us to prove robust lower bounds?
\end{question}
\end{mdframed}


Having fixed such a definition, we ask to what extent component-stability restricts \MPC algorithms, and whether the concept is indeed a technicality or actually a significant limitation. The lifting arguments of \cite{GKU19} are designed for randomized algorithms, which raises the following question:

\begin{mdframed}[hidealllines=true,backgroundcolor=gray!25]\vspace{-8pt}
\begin{question}
\label{q:stable-rand}
Does component-instability help for obtaining improved \textsf{randomized} low-space \MPC algorithms for graph problems? Is there any \emph{separation between randomized compo\-nent-stable} and \emph{component-unstable} \MPC algorithms?
\end{question}
\end{mdframed}

\hide{We note that the lifting framework in \cite{GKU19} is not restricted to exact graph problems (e.g., \LCL problems, see \Cref{subsec:problem-def} and \cite{NS95,CKP19}), but it captures also lower bounds for the approximate variants of these problems (e.g., approximate maximum matching). However, one must restrict the problems considered in this setting to those whose solution can be defined in a component-stable manner. For example problems that involve global computations (e.g., aggregate function of all nodes) trivially cannot be addressed in this framework. Following \cite{GKU19}, we restrict attention to local graph problems (i.e., with polylogarithmic or sub-logarithmic \LOCAL complexity).
	
\Peter{ took out this paragraph since it doesn't seem too related here, and I'm not sure I fully agree with all of it (we don't explicitly restrict attention to local graph problems, we instead define our strange class of replicable problems etc..., but we don't really want to go into that here)}

}

We then turn to consider the impact of component-stability on \emph{deterministic} low-space \MPC algorithms. Since the recent derandomization technique of \cite{CDP20a,CDP20b} leads to inherently component-unstable algorithms, we ask:

\begin{mdframed}[hidealllines=true,backgroundcolor=gray!25]\vspace{-8pt}
\begin{question}
\label{q:stable-det}
Does component-instability help for obtaining improved \textsf{deterministic} low-space \MPC algorithms for graph problems? Is there any \emph{separation between deterministic compo\-nent-stable} and \emph{component-unstable} \MPC algorithms?
\end{question}
\end{mdframed}

\hide{To address this question, we first extend the lifting arguments of \cite{GKU19} to the deterministic setting. This provides conditional lower bounds for component-stable deterministic \MPC algorithms. The benefit in such a lifting result is that for several local problems, we will get exponentially stronger (conditional) lower bound results for deterministic algorithms than for randomized.
	
\Peter{I've postponed all 'answering' of questions (which this seems to be) until the end}}

Understanding the gap between randomized and deterministic solutions is one of the most fundamental and long-standing questions in graph algorithms. In a very related context, the last few years provided a sequences of major breakthrough results which almost tightly characterize the gap between the deterministic and randomized complexities for \emph{distributed computation} (in the \LOCAL model). Rohzo{\v{n}} and Ghaffari \cite{RG20} settled a several-decades-old open problems by presenting a deterministic polylogarithmic algorithm for network decomposition. Their result implies that any polylogarithmic-time randomized algorithm for \LCL problems \cite{CKP19,NS95} can be derandomized to a polylogarithmic-time deterministic algorithm. In other words, \emph{randomization does not help in the polylogarithmic-time regime}. On the other hand, Chang, Kopelowitz and Pettie \cite{CKP19} showed that \emph{in the sub-logarithmic-time regime, randomization might provide an exponential benefit}. For example, for 
$\Delta$-coloring of trees of maximum degree $\Delta$ there are $\Theta(\log_{\Delta}\log n)$ randomized-round algorithms and a deterministic lower bound of $\Omega(\log_{\Delta} n)$ rounds. Balliu et al. \cite{BBHORS19} recently showed that maximal matching and maximal independent sets cannot be found by $o(\Delta + \log\log n/ \log \log \log n)$-rounds randomized algorithms, and deterministically in $o(\Delta + \log n)$ rounds, thus provide an exponential gap in the lower bound for bounded-degree graphs. Finally, \cite{CQW+20} presented additional separation results for edge-coloring: $(2\Delta-2)$-edge-coloring requires $\Omega(\log_{\Delta} \log n)$ randomized rounds, and $\Omega(\log_\Delta n)$ deterministic rounds. In view of these separation results, we therefore ask in the context of related \MPC computation:

\begin{mdframed}[hidealllines=true,backgroundcolor=gray!25]\vspace{-8pt}
\begin{question}
\label{q:rand-det-stable-non-stable}
Is there a gap between component-stable \textsf{randomized} algorithms vs. component-stable \textsf{deterministic} algorithms? \hide{If so, does this gap exist for component-unstable algorithms? \Peter{Removed this part since the answer is no, and we say next that the answers are all yes.}}
\end{question}
\end{mdframed}

\hide{
	

Peter - I also removed this paragraph for now, since it refers to the removed Lemma, and I'm not sure it makes sense to have this here without it.


For the class of local graph problems (e.g., \LCL{}s, which have been extensively studied in the \LOCAL model) until now our arguments above demonstrate that
component-instability only help in reducing the deterministic complexity of the problem.

The latter is still bounded from below by the logarithm of the randomized complexity of these problems. We provide a partial explanation for this behavior. One of the greatest advantage of component-unstable \MPC algorithms is that they allow one to boost the success guarantee by a polynomial factor, without paying any slow-down in the round complexity. This allows one to consider, in principle, faster randomized \LOCAL algorithms which succeed just with a \emph{polynomially small probability} of $\frac{1}{\poly(n)}$ (rather than the standard requirement of at least $1-\frac{1}{n}$). The existence of an exponentially faster randomized algorithm would immediately provide an exponentially faster component-unstable \MPC algorithms, thus breaking the lifting argument for \LCL problems. We show (cf. Lemma \ref{lem:rand-barrier}) that no such algorithm exists but proving that for \LCL problems, the relaxation of the success guarantee of the \LOCAL randomized algorithm from the standard $1-\frac{1}{\poly(n)}$ into $\frac{1}{n}$ barely affects its round complexity. This stands in a stark contrast to the approximation variants of the \LCL problems. For the approximation of the maximum independent set problem, there is $O(1)$ round randomized algorithm that succeeds with constant probability, but any algorithm that succeeds with probability of at least $1-\frac{1}{\log n}$ requires $\Omega(\log^*n)$ rounds \cite{KKSS19}.

}

In this paper, we answer all four questions in the affirmative.

\subsection{Our contributions}
\label{subsec:contributions}

\hide{
To address Question \ref{q:definition}, we revise the framework of low-space component-stable \MPC algorithms due to Ghaffari et al.\ \cite{GKU19}.
We study in details the framework and demonstrate that in order to be broadly applicable, many aspects of the original setting are highly sensitive to the exact definition used, and require amendments to carefully specify what information the outputs of component-stable algorithms may depend on. We present a modified framework to reach a revised definition of component-stability which both encompasses many existing \MPC algorithms, and for which robust conditional lower bounds can be shown. Our revised framework does not only recover all main results from the framework of Ghaffari et al.\ \cite{GKU19}, but also extends the arguments to include
conditional lower bounds for \emph{deterministic algorithms} and fine-grained lower bounds that depend on 
$\Delta$.


To address Question \ref{q:stable-rand}, we present a local problem that admits an $\Omega(\log^* n)$ round lower bound in the \LOCAL model, and thus (by our revised framework of low-space component-stable \MPC algorithms extending \cite{GKU19}) a conditional lower bound of $\Omega(\log\log^* n)$ rounds for low-space component-stable \MPC algorithms. This problem admits a very simple $O(1)$-round randomized low-space \MPC algorithm which is component-unstable (in fact, this algorithm can also be derandomized in $O(1)$ rounds as well). The problem that separates component-stable and component-unstable randomized low-space \MPC algorithms is an approximation of maximum independent sets, recently studied in the \LOCAL model by \cite{KKSS19}.


To address Question~\ref{q:stable-det}, we rely on our generalization of the lifting framework revising \cite{GKU19} to the deterministic setting. This allows us to provide even stronger conditional lower bounds for deterministic low-space component-stable \MPC algorithms that those that could be obtained by comparison to randomized bounds via \cite{GKU19}. We then break many of these stronger lower bounds by exploiting component-\emph{instability}. Towards that goal, we provide a sequence of derandomization results using a collection of tools such as bounded-independence hash functions and pseudorandom generators (PRG) \cite{Vadhan12}. Using PRGs, we are able to provide a more general derandomization technique, that allows us to simulate a class of $r$-round randomized \LOCAL algorithms in $O(\log r+\log^* n)$ \emph{deterministic} \MPC rounds.


To address Question \ref{q:rand-det-stable-non-stable}, we show that for any $r$-round randomized low-space \MPC algorithm there exists a non-uniform and non-explicit $O(r)$-round deterministic low-space \MPC algorithm. This deterministic algorithm is inherently \emph{component-unstable}. On the other hand, for component-stable \MPC algorithms we extend the randomized lifting arguments from \cite{GKU19} to the deterministic setting and provide a \emph{conditional exponential separation} for several \LCL problems (e.g., those that admit an exponential separation in the \LOCAL model).

}
In this paper, we study the power and limitations of component-stable algorithms in the low-space \MPC model. This class of algorithms is to this date the only class for which conditional lower bound for local graph problems can be obtained.  Our main contribution is in demonstrating the impact of this property on the complexity of randomized and deterministic local graph problems. Towards that goal, we informally define the following four complexity classes:
\begin{itemize}
\item \textbf{\MPCDetNonStable:} deterministic low-space component-unstable \MPC algorithms,
\item \textbf{\MPCDetStable:} deterministic low-space component-stable \MPC algorithms,
\item \textbf{\MPCRandNonStable:} randomized low-space component-unstable \MPC algorithms,
\item \textbf{\MPCRandStable:} randomized low-space component-stable \MPC algorithms.
\end{itemize}

We refer to \Cref{subsec:complexity-classes} for the precise definitions.


\paragraph{A robust lifting framework.}
%

We rectify the framework of low-space component-stable \MPC algorithms due to Ghaffari et al.\ \cite{GKU19}.
We study the framework in detail and demonstrate that in order to be broadly applicable, many aspects of the original setting are highly sensitive to the exact definition used, and require amendments to carefully specify what information the outputs of component-stable algorithms may depend on. We present a modified framework, reaching a revised definition of component-stability (see Definition \ref{def:component-stable-algorithms}) which both encompasses many existing \MPC algorithms, and for which robust conditional lower bounds can be shown. This answers Question~\ref{q:definition}.

\paragraph{Extensions to deterministic and degree-dependent lower bounds.}

 Our revised framework not only recovers all main results from the framework of Ghaffari et al.\ \cite{GKU19}, but also extends the arguments to include
conditional lower bounds for \emph{deterministic algorithms} and fine-grained lower bounds that depend on 
$\Delta$. While our main theorem (\Cref{thm:lb}) lifting \LOCAL lower bounds to component-stable \MPC algorithms has several subtle assumptions, the main, informal claim is that for many graph problems $\mathcal{P}$, if $\mathcal{P}$ has a $T(n,\Delta)$-round (randomized or deterministic) lower bound in the \LOCAL model, then assuming the connectivity conjecture, any low-space component-stable (respectively, randomized or deterministic) \MPC algorithm solving $\mathcal{P}$ requires $\Omega(\log T(n,\Delta))$ rounds.


\paragraph{Instability helps \emph{randomized} \MPC algorithms.}

To address Question \ref{q:stable-rand}, we consider the problem of finding a large (specifically, of size $\Omega(n/\Delta$) independent set. This problem has been recently studied by 
Kawarabayashi et al.\ \cite{KKSS19}, who provided a randomized lower bound of $\Omega(\log^* n)$ rounds (for a specific range of $\Delta$)\hide{under the success guarantee of $1-\frac{1}{10\log n}$ - Peter: don't think this is important, since as we discussed, our lifting Thm. isn't set up to carry it through to MPC}. We show that their lower bound can be adapted to our revised lower-bound lifting framework of \Cref{thm:lb}, obtaining a conditional lower bound of $\Omega(\log\log^* n)$ rounds for component-stable \MPC algorithms.\footnote{The original lifting arguments of \cite{GKU19} only hold for \LOCAL lower bounds that hold under exact knowledge of $n$;  the lower bound of \cite{KKSS19} does not, but holds only under knowing a polynomial estimate of $n$, which is allowed for in our framework.} In contrast, we present a very simple $O(1)$-round component-unstable randomized algorithm for the problem. In fact, this algorithm can further be \emph{derandomized} within $O(1)$ rounds (see \Cref{thm:stableMPC-LB}), demonstrating an instance in which deterministic component-unstable algorithms are more powerful even than randomized component-stable algorithms (i.e., $\MPCDetNonStable \nsubseteq \MPCRandStable$).

\begin{theorem}
\label{thm:approx-IS}
Conditioned on the connectivity conjecture, any com\-po\-nent-stable low-space \MPC algorithm for computing an independent set of size $\Omega(n/\Delta)$ on $n$-node graphs (for the full range of $\Delta\in [1,n)$) and succeeding with probability at least $1-\frac 1n$, requires $\Omega(\log\log^*n)$ rounds. This problem admits a simple $O(1)$-round randomized low-space \MPC algorithm which is component-unstable; additionally, the algorithm can be derandomized within $O(1)$ rounds.
\end{theorem}

Stated in complexity language, \Cref{thm:approx-IS} provides a \emph{separation between the class of \MPCRandNonStable and \MPCRandStable} (conditioned on the connectivity conjecture, see \Cref{thm:complexity-SRan-Ran}). The basic observation providing this separation is the fact that one can easily compute
in $O(1)$ rounds (even in the \LOCAL model) an independent set of $\Omega(n/\Delta)$ nodes \emph{in expectation}. In the \LOCAL model, we need provably longer to achieve a high-probability success guarantee of $1-\frac 1n$. In the low-space \MPC model, however, we can perform the process of \emph{success probabiliy amplification}: we run $\Theta(\log n)$ \emph{parallel} repetitions of the basic algorithm, and choose a successful one if such exists, amplifying the success probability to $1-\frac 1n$ without any slow-down. This powerful process, though, is inherently component-unstable, since it relies on globally agreeing on one of the repetitions to use\footnote{Indeed, this causes an issue with the proof of Lemma III.1 of \cite{GKU19}, where success probability amplification is used in an algorithm $A_{\MPC}$ that is later (in Lemma IV.2 in \cite{GKU19}) assumed to be component-stable.}.


\paragraph{Instability helps \emph{deterministic} \MPC algorithms.}

We then turn to consider the effect of component-stability on deterministic \MPC algorithms. While the original setup of Ghaffari et al.\ \cite{GKU19} had been designed only for randomized algorithms, the revised framework developed in our paper in \Cref{subsubsec:component-stable-algorithms} extends naturally to the deterministic setting, providing a robust deterministic lifting analog in \Cref{thm:lb}. \Cref{thm:lb} provides a general framework lifting unconditional deterministic lower bounds for the \LOCAL model for many natural graph problems to conditional lower bounds for low-space component-stable \MPC algorithms 
in the same way as the randomized framework in \cite{GKU19}.

We then turn to show that with component-instability one can in fact surpass these conditional lower bounds and present several results showing a \emph{separation between \MPCDetNonStable and \MPCDetStable} (conditioned on the connectivity conjecture) and positively answering Question \ref{q:stable-det}. In \Cref{subsec:LLL-related}, we show that for several problems closely related to LLL, including sinkless orientation and some variants of edge-coloring and vertex-coloring, component-instability helps for deterministic algorithms. Finally, in \Cref{subsec:examps-complexity-separation}, we demonstrate a similar result for the class of all \LOCAL \emph{extendable} algorithms by combining the lifting of deterministic \LOCAL lower bounds in \Cref{thm:lb} with a derandomization technique using pseudorandom generators. To demonstrate the applicability of this derandomization recipe, we show how it can be used to \emph{improve the deterministic running times} of two cornerstone problems in low-space MPC: \emph{maximal independent set} and \emph{maximal matching}. And so, on one hand we prove (\Cref{thm:MIS-MM-1}) that conditioned on the connectivity conjecture, there is no deterministic low-space component-stable \MPC algorithm that computes a maximal matching or maximal independent set, even in forests, in $o(\log\Delta +\log\log n)$ rounds, and on the other hand, we give a deterministic low-space component-unstable \MPC algorithm for these problem running in $O(\log\log\Delta + \log\log\log n)$ rounds (when $\Delta=2^{\log^{o(1)} n}$, Corollary \ref{cor:thm:MIS}). (The resulting \MPC algorithm must either perform heavy local computations, or alternatively, the underlying PRGs can be hard-coded in the \MPC machines for a non-uniform but computationally-efficient algorithm.)


%


\paragraph{Relations between randomized and deterministic \MPC algorithms.}

Finally, we consider the interesting gap between randomized and deterministic algorithms in the low-space \MPC setting. As observed by \cite{CKP19} and \cite{GK19} (and see Lemma \ref{lem:det-large-prob} for a sketch of the argument), randomized algorithms that succeed with probability of $1-1/2^{n^2}$ can be turned into non-uniform deterministic algorithms. This result can also be extended to the low-space \MPC setting, with some caveat. In contrast to the \LOCAL model where the space of the nodes is unlimited, in the low-space \MPC setting, the transformation implied by \cite{CKP19} and \cite{GK19} yields a non-explicit algorithm (see Lemma \ref{lem:det-large-prob}). By using success probability amplification with $\poly(n)$ machines, one can boost the success guarantee of any randomized  \MPC algorithm from $1-1/\poly(n)$ to $1-1/2^{n^2}$ without any slowdown in the round complexity. From the complexity perspective, restricting to non-uniform and non-explicit computation, one therefore finds that $\MPCDetNonStable = \MPCRandNonStable$, see \Cref{thm:complexity-Det-Ran}. For some specific problems, we can perform more careful derandomization methods that \emph{do not} cause the resulting deterministic algorithms to be non-uniform, non-explicit, or to use excessive global space, demonstrating that component-stability restricts power even without these allowances.

As we have noted, success probability amplification is inherently component-unstable. Turning our focus to low-space component-stable \MPC algorithms, here we can provide a conditional separation between randomized and deterministic algorithms (\Cref{thm:instability-provably-helps}), positively answering Question \ref{q:rand-det-stable-non-stable}. This separation follows by combining (i) the conditional lifting for randomized component-stable algorithms and deterministic component-stable algorithm with (ii) local problems for which there is provable gap in their randomized and deterministic \LOCAL complexity.


\hide{Peter - I hid this because I think most of it has already been said in the previous paragraphs and I think our introduction is too long in general, but feel free to change if you disagree
	
\paragraph{Hypothesis for the role of instability.}

The main property of component-unstable algorithms that we leverage in this work is their ability to boost the success guarantee without incurring any slow-down in the round complexity. In the context of approximation problems this is useful already for randomized algorithm. In the context of \LCL problems, we prove that this might not help to improve their randomized complexity, but considerably helps to improve their deterministic complexity, for the following reason. Many of the gap results between randomized and deterministic algorithms, are based on exploiting the following dichotomy. {\color{red} A randomized algorithm, in principle, is allowed to err with probability at most $1/n^c$; in contrast, a deterministic algorithm, corresponds to a randomized algorithm with an exponentially smaller error probability of $1-1/2^{n^2}$.} This exponential gap on the allowed error leads to an exponential gap in the round complexity of  these algorithms. The lifting arguments imply that this exponential gap is at least conditionally preserved for component-stable algorithms. Instability however can beat the conditional deterministic bounds by using the polynomially many machines to reduce the polynomial error into an exponential error without any loss. This ultimately leads to strictly improved component-unstable deterministic algorithms for \LCL problems. Since these boosting arguments lead to non-uniform and non-explicit deterministic algorithms, and also use larger global space than we would like, we provide a more delicate and algorithmic derandomization for breaking the conditional deterministic lower bounds for specific problems. Finally, for the class of \LCL problems, we speculate that the \emph{randomized} lifting results of \cite{GKU19} and from this paper should hold also for component-unstable algorithms.
}

\paragraph{Complexity summary:}

Let us summarize the complexity results, assuming the connectivity conjecture, and allowing non-uniform \MPC algorithms. Our study demonstrates that in low-space \MPC, component-unstable algorithms are provably stronger than their component-stable counterparts, both for deterministic and randomized algorithms
(\Cref{thm:complexity-SDet-Det,thm:complexity-SRan-Ran}). Further, for component-stable algorithms, randomized algorithms are provably stronger than their deterministic counterparts
(\Cref{thm:complexity-SDet-SRan}). However, for arbitrary (possibly component-unstable) algorithms this is not the case: any randomized algorithm can be efficiently simulated by a deterministic one (Lemma \ref{lem:det-alg-many-machines} and \Cref{thm:complexity-Det-Ran}).

%% file: prelim.tex



In this section we suggest an array of changes that may be made to the framework of component-stable algorithms due to Ghaffari et al.\ \cite{GKU19}, in order to reach a revised definition of component-stability (Definition~\ref{def:component-stable-algorithms}) which both encompasses many existing randomized \MPC algorithms, and for which robust conditional lower bounds can be shown. These changes also allow us to extend the original setting to both deterministic algorithms and those that have running-time dependency on maximum degree $\Delta$.

\subsection{Discussion of definitions of component stability}
\label{subsec:problems-with-def-comp-stab}

Let us first present the description of component stability from \cite[Section~II]{GKU19}:
\begin{quote}\small
\emph{
Formally, assume that for a graph $G$, $\mathcal{D}_G$ denotes the initial distribution of the edges of $G$ among the $M$ machines and the assignment of unique IDs to the nodes of $G$. For a subgraph $H$ of $G$ let $\mathcal{D}_H$ be defined as $\mathcal{D}_G$ restricted to the nodes and edges of $H$. Let $H_v$ be the connected component of node $v$. An \MPC algorithm $\cA$ is called \emph{component-stable} if for each node $v \in V$, the output of $v$ depends (deterministically) on the node $v$ itself, the initial distribution and ID assignment $\mathcal{D}_{H_v}$ of the connected component $H_v$ of $v$, and on the shared randomness $\mathcal{S}_M$.}
\end{quote}

We informally sketch the line of argument of \cite{GKU19} that leverages this definition to lift \LOCAL lower bounds to \MPC: first, it is shown that if there is a \LOCAL lower bound for a problem $P$, and an MPC algorithm $A_{MPC}$ is able to solve $P$ faster than the $\log$ of the \LOCAL lower bound, there there must exist two graphs $G$ and $G'$, which are locally indistinguishable but on which $A_{MPC}$'s output must differ (at least with some sufficiently large probability). In the terminology of \cite{GKU19}, $A_{MPC}$ must be `farsighted': able to somehow make use of information from far-away nodes.

These graphs $G$ and $G'$, and the assumed algorithm $A_{MPC}$, are then ingeniously used to construct an algorithm $B_{st-conn}$ that solves a connectivity problem conjectured to be hard. Specifically, $B_{st-conn}$ constructs a pair of simulation graphs based on its input to the connectivity problem. These simulation graphs consist of many disjoint copies of induced subgraphs of $G$ and $G'$ respectively. The construction is made in such a way that a \emph{full} copy of $G$ and $G'$ only appears if two particular nodes (designated $s$ and $t$) are connected in the input graph for the connectivity problem.

$B_{st-conn}$ simulates $A_{MPC}$ on this pair of simulation graphs. If $s$ and $t$ are connected, then full copies of $G$ and $G'$ are present as connected components in the simulation graphs, and $A_{MPC}$ should return \emph{different} outputs on them with sufficiently high probability. Otherwise, there are no full copies of $G$ and $G'$, and $A_{MPC}$ returns the same outputs on both simulation graphs. This difference in behavior is exploited to allow $B_{st-conn}$ to determine whether $s$ and $t$ are connected, and solve the connectivity problem faster than is conjectured to be possible.

The property of component-stability is crucial in this last step: we require that $A_{MPC}$ behaves the same on $G$ and $G'$ when they are connected components of the (much larger) simulation graphs as it does when they are the entire input (as was the case when showing that $A_{MPC}$ was farsighted). Otherwise, we could not say anything about $A_{MPC}$'s output on the simulation graphs. It transpires that this argument is quite fragile, and highly sensitive to the precise definition of component-stability used. We discuss some of the issues below.

\paragraph{Randomized component-stable algorithms must be allowed dependency on $n$.}

The first major point of consideration is that, as defined in \cite{GKU19}, the output of a node $v$ under a component-stable algorithm must depend only on shared randomness, the IDs of $v$ and its component, and the input distribution of edges to machines. In particular, no provision is made for dependency on the number of nodes $n$ in the input graph, and indeed, the arguments of \cite{GKU19} seem to forbid it.\footnote{Specifically in proof of \cite[Lemma IV.2]{GKU19}, where algorithm $A_{MPC}$ is simulated on large simulation graphs containing smaller components $G$ and $G'$, as discussed above, its behavior on $G$ and $G'$ as components is only identical to when run on them as sole input if no dependency on $n$ is permitted.} This is somewhat counter-intuitive for \MPC algorithms: while a \LOCAL algorithm can never determine $n$ unless it is given as input (and therefore it is commonly assumed that we provide the algorithm with at least a polynomial estimate of $n$), an \MPC algorithm can easily do so in $O(1)$ rounds, by simply summing counts of the number of nodes held on each machine. We can therefore assume any such algorithm \emph{has} knowledge of the exact value of $n$, and natural algorithmic approaches would generally make use of this knowledge.

Furthermore, the \emph{success} probability of correct randomized algorithms is defined to be at least $1-\frac 1n$, in accordance with the standard definition of \emph{with high probability correctness}. This causes a contradiction for algorithms with no output dependency on $n$:

Consider a correct component-stable \MPC algorithm $A$ for a problem in which the validity of a node's output can be verified by seeing its connected component (we will formalize this notion later), running on a $n$-node graph. This algorithm must produce a valid output for each node in the graph with probability at least $1-\frac 1n$.

We now add $\eta$ disconnected nodes to the input graph. If $A$'s output does not have any dependency on $n$, then it must be unchanged at each of the original nodes, since they are in an unchanged connected component. However, $A$ must succeed on the new graph with probability at least $1-\frac{1}{n+\eta}$. Since the problem is component-stable, the probability that $A$ succeeds on all of the nodes of the original graph is at least as high as the probability that it succeeds on the whole graph, i.e., $1-\frac{1}{n+\eta}$. So, $A$ must succeed on the original graph with probability at least $1-\frac{1}{n+\eta}$, and since we can set $\eta$ arbitrarily high, must therefore succeed with certainty; this requires $A$ to be deterministic!

Another problem with disallowing dependency on $n$ is that \emph{running times} generally depend on $n$. While, for some (deterministic) algorithms, this dependency could be replaced by other parameters (e.g., the maximum number of nodes in a connected component, or size of the ID space), this is not the case for randomized algorithms, where the running time usually directly affects the success probability, which itself must depend on $n$.

So, in short, a definition of component-stability which does not allow any dependency on $n$ includes essentially no non-trivial randomized algorithms.

\paragraph{If we allow dependency on $n$, we must restrict the class of problems in order to obtain \MPC lower bounds.}

We have seen that, to give results which apply to probabilistic algorithms, we must allow dependency on $n$. However, we cannot then hope to obtain a result akin to Theorem I.4 of \cite{GKU19} for \emph{all} graph problems.

As an example, consider the following problem: each node must output \YES if the entire graph is a simple path with consecutive node IDs, and \NO otherwise. Note that there is only one possible correct output for each node $v$, and that this output is a deterministic function of its component and the value of $n$ (since $v$'s output should be \YES iff its component is an $n$-node path with consecutive IDs). Furthermore, there is an $O(1)$-round \MPC algorithm for the problem: it is straightforward to check whether there are two nodes of degree $1$, $n-2$ nodes of degree $2$, and that each node's $1$-hop neighborhood is consistent with being in a path of consecutive IDs. So, if component-stability is defined to allow dependency on $n$, an $O(1)$-round deterministic component-stable algorithm for the problem exists.

However, the problem has a trivial $n-1$-round (randomized) \LOCAL lower bound, since a \YES instance can be changed to a \NO instance by only altering the ID of one endpoint of the path, and the other endpoint requires $n-1$ rounds to detect this change. Hence, we cannot hope for a universal method of lifting \LOCAL lower bounds to non-trivial component-stable \MPC lower bounds if component-stability allows dependency on $n$.

We will see, though, that such counterexamples are necessarily quite contrived, and that we \emph{can} prove such a result for a class that includes most problems of interest (such as, e.g., \emph{all locally-checkable (\LCL)} problems, see \Cref{subsec:problem-def}).

\paragraph{Uniqueness of identifiers.}

It is common in both \LOCAL and \MPC to assume that nodes of the input graph are equipped with identifiers (IDs) that are unique throughout the entire graph. This assumption, however, is somewhat at odds with the concept of component-stability: if, for example, a disconnected node is added to a valid graph, sharing an ID with an existing node, then the input becomes invalid. So, outputs for the original nodes are now allowed to change arbitrarily, even though their components have not altered.

We could instead require that IDs are only component-unique (i.e., they are allowed to be shared by disconnected nodes, but not connected ones). This is a weaker assumption which aligns well with component-stability, and is still sufficient for \LOCAL algorithms (in which nodes have no interaction with or dependency on disconnected nodes).

This approach, though, presents a problem in \MPC. Unlike in \LOCAL, where nodes are inherently separate computational entities which only need IDs for symmetry-breaking (particularly for deterministic or shared randomness algorithms), in \MPC an input graph node essentially \emph{is} its ID. The input is given only as a binary encoding of the IDs of nodes and edges, and so any two nodes with the same ID will be contracted to a single node when this input is interpreted as a graph. As a consequence, \MPC algorithms cannot work with graphs in which IDs are only component-unique.

Our solution to this problem is to separate the two functions of IDs. We will assume that IDs are only component-unique, and that component-stable \MPC algorithms can depend on these. However, we also provide \MPC algorithms with fully-unique \emph{names} for nodes, whose purpose is \emph{only} to allow the algorithm to distinguish the input graph's nodes as separate objects. Accordingly, we do not allow the output of component-stable algorithms to depend on the names.%
\footnote{Note that, unlike the changes regarding dependency on $n$ and problem class, this change is not necessary to show a general framework for conditional \MPC lower bounds --- the same results could be proven assuming fully-unique IDs (at least for randomized algorithms) using techniques from \cite{GKU19}. However, we feel that this definition better captures the `spirit' of component-stability.}

\paragraph{Initial distribution of input.}
The definition of \cite{GKU19} allows \MPC algorithms' outputs to depend on the initial distribution of the input to the machines. While this is natural to do, we observe that under our definition it is not necessary: given a component-stable algorithm $A_{MPC}$ whose output depends on this distribution, we can always create a new component-stable $B_{MPC}$ which does not.

Specifically, since the nodes have unique $poly(n)$ names (and we can also give edges unique $poly(n)$ names based on their endpoints), and we are allowed any $poly(n)$ number of machines, algorithm $B_{MPC}$ can first (in one round) redistribute each node and edge of the input to its own dedicated machine, with the same name as the corresponding node or edge. Then, it simulates $A_{MPC}$, and reaches a valid output, which is now independent of the initial distribution. Since $A_{MPC}$'s output is component-stable, $B_{MPC}$'s is also.\footnote{We will define component-stable outputs to not depend on the names of machines --- this is not a major restriction, since we are not aware of any \MPC algorithms which are not independent of renaming machines. However, it is an important point here, since $B_{MPC}$'s machine names now depend on node names, upon which $B_{MPC}$'s output must not depend.}

Therefore, a lower bound for component-stable algorithms that depend on input distribution implies a lower bound for those that do not. So, we can disallow this dependency from the definition without weakening our results.


\subsection{Graph families}
\label{subsec:graph-families}

In this section, we make some definitions concerning the input graphs on which \MPC algorithms run. Firstly, to address the problem concerning uniqueness of identifiers, we define \emph{legal graphs} to be those with separate unique node names and component-unique node IDs as discussed:

\begin{definition}
\label{def:legal-graphs}
A graph $G$ is called \textbf{legal} if it is equipped with functions $\textsf{ID}, \textsf{name} : V(G)\rightarrow [poly(n)]$ providing nodes with IDs and names, such that all names are fully unique and all IDs are unique in every connected component.
\end{definition}

Throughout the paper, we will always assume that input graphs for \MPC are legal (and we will \emph{ensure it when constructing inputs ourselves}). For component-stable algorithms, this is to allow a weaker dependency on the IDs and not the names, as discussed above. For non-component-stable algorithms, it is no different from the standard (fully-unique IDs) requirement, since their outputs are allowed to depend on the names, and so we can simply use the names as IDs.

Next, we make a definition which will allow us to show \MPC lower bounds on specific families of graphs. \LOCAL lower bounds are often proven using graphs from some specific family $\mathcal H$ as the \emph{``hard instances''}: in particular, many such bounds are proven on trees. Lower bounds on restricted families of graphs are stronger than those on the class of all graphs, and can also provide more meaningful hardness results for problems which are only \emph{possible} on restricted families (such as $\Delta$-vertex coloring, see Theorem \ref{thm:vcoloring-1}). When lifting \LOCAL lower bounds on restricted graph families to \MPC, we therefore wish to preserve the family on which the lower bound holds.

As discussed in Section \ref{subsec:problems-with-def-comp-stab}, the lines of argument made both by \cite{GKU19} and this work involve construction of a \emph{simulation graph} $H$ as the ``hard instance'' in \MPC. These simulation graphs (from the proof of Lemma IV.2 of \cite{GKU19} and Lemma \ref{lem:stconn} here) are constructed differently, but both consist of disjoint unions of induced subgraphs of some hard instance $G$ for the \LOCAL lower bound.

The simulation graph $H$ is not necessarily in $\mathcal H$, if $\mathcal H$ is an arbitrary graph family, and is not therefore a valid input for an \MPC algorithm defined to run on $\mathcal H$. This is not necessarily a problem for \cite{GKU19}, since one could require that their definition of component-stability should hold \emph{even if algorithms are given an invalid input}. Under this definition, the output of a component-stable \MPC algorithm $A_{MPC}$ on $H$ must be the disjoint union of $A_{MPC}$'s outputs on its connected components (which, in this case, are in $\mathcal H$) separately, and this circumvents the need for $H$ itself to be in $\mathcal H$.

However, as we notice in \Cref{subsec:problems-with-def-comp-stab}, to incorporate non-trivial randomized algorithms, any suitable definition of component-stable algorithms must allow dependency on $n$. Then, the output of $A_{MPC}$ of $H$ need not be the union of that on $H$'s connected components, since the inputs have differing values of $n$. This necessitates several changes from the proofs of \cite{GKU19}, one of which is that we \emph{do} require $H\in \mathcal H$.

To ensure that this is the case, we prove our lower-bound lifting argument only for the following families of graphs: 

\begin{definition}[Normal families of graphs]
\label{def:normal-graphs}
A class of graphs $\mathcal{H}$ is called \textbf{normal} if it is hereditary (i.e., closed under removal of nodes) and closed under disjoint union.
\end{definition}

The set of \emph{all graphs} is a normal family, and can always be used in the worst case. Further, observe that the \emph{class of all trees is not a normal family} of graphs; however, the family of \emph{all forests} is normal. Therefore, for example, Theorem \ref{thm:lb} implies that \LOCAL lower bounds on \emph{trees} can be lifted to conditional \MPC lower bounds on \emph{forests} (but not trees).

\subsection{Types of graph problems and replicability}
\label{subsec:problem-def}

We next define the types of \emph{problem} we will encompass with this work. We will focus on graph problems, and let $\mathbb{G}$ be the collection of all legal input graph instances.

We consider only graph problems where \emph{each node of the input graph must output some label} from a finite set $\Sigma$. For example, for the vertex coloring problem the label of a node corresponds to its color, and for the independent set problem, the label corresponds to the indicator variable whether the node is in the independent set returned.

As in \cite{GKU19}, we do not explicitly output labels for edges. However, to apply our results to problems in which only edges are given output labels (such as matching, edge coloring, or sinkless orientation), we can simply redefine the problem as vertex-labeling on the \emph{line graph} (where the vertices represent edges of the input graph, with IDs and names given by Cartesian products of the IDs and names of their endpoints). For any normal graph class $\mathcal H$, the family of $\mathcal {L_H}$ of line graphs of graphs from $\mathcal H$ is also normal. Working on the line graph increases number of nodes $n$ at most quadratically and maximum degree $\Delta$ at most by a factor of $2$, and in \LOCAL requires only $1$ extra round. We will see that performing this conversion to the line graph will allow us to obtain results for edge-labeling problems without any asymptotic change in bounds. We can then convert back to recover a solution to the problem on the original graph.

A graph problem is then defined by a collection of \emph{valid} outputs for each possible pair (topology, IDs) of a legal input graph. Importantly, we do \emph{not} allow validity to be dependent on the \emph{names} of graph nodes (though these names are part of any legal input). That is, given a particular input graph topology and set of IDs, the collection of valid outputs must be consistent regardless of node names. This is because component-stable outputs are not allowed to depend on names, so most problems which allowed solution-dependency on names would be trivially unsolvable by component-stable algorithms. In any case, names were introduced solely to allow \MPC algorithms to distinguish nodes as objects, and should not be considered part of problems.

The goal of any algorithm for the problem is then to provide a valid output for the specific legal input it was given. For many problems it is useful to have a concept of the output of a \emph{particular node} being valid. The overall output is then valid if all nodes' outputs are valid. To capture this concept, we define the following sub-class of problems:

\begin{definition}
For $r \in \nat$, an \textbf{$r$-radius checkable problem} is defined by a collection of \textsf{valid} outputs for each $r$-radius centered graph equipped with unique IDs.\footnote{$r$-radius graphs are, by definition, connected, so component-unique IDs are unique IDs.} The output of a node $v$ in input graph $G$ is deemed valid if the centered graph given by its $r$-radius ball, and the outputs thereof, is a member of this valid collection. An overall output on $G$ is valid if all nodes' outputs are valid.
\end{definition}

An \emph{$r$-radius centered graph} here is simply a connected graph with a designated center node, from which all other nodes are of distance at most $r$.

One can see that $r$-radius checkable problems are precisely those whose solutions can be verified in $r$ rounds of \LOCAL. Note that the majority of graph problems of interest are $r$-radius-checkable for some $r\le n$: for example, the vertex coloring problem requires that each node outputs a color distinct from the colors output by its neighboring nodes, and thus is easily 1-radius-checkable. Similarly, all \LCL (\emph{locally-checkable labeling}, see, e.g., \cite{CKP19,NS95}) problems, a commonly studied class particularly from a lower-bounds perspective, are $O(1)$-radius checkable. Still, some natural problems are not $n$-radius checkable problems: most notably, approximation problems are not, since there is no notion of a node's validity, and nor can nodes determine overall validity by seeing their $n$-radius ball (i.e., their entire connected component). So, while some of our results concern $r$-radius checkable problems (such as those in \Cref{subsec:examps-complexity-separation}), our main lower bounds results will use a more general class of problems, see below, in order to incorporate approximation~problems.


\subsubsection{Replicable graph problems}
\label{subsec:which-are-replicable-problem}

We have seen, from \Cref{subsec:problems-with-def-comp-stab}, that to transfer \LOCAL lower bounds to \MPC, under a definition of component-stability that includes randomized algorithms (and so allows dependency on $n$), one must restrict the class of problems, since some (contrived) problems have $\Omega(n)$-round \LOCAL lower bounds and $O(1)$-round \MPC algorithms. Our goal in this section is to make the minimal restriction needed to facilitate such lower-bound lifting arguments.

During proof of \Cref{thm:lb} (our main lower-bound lifting theorem), we will consider multiple disjoint copies of the input graph enhanced by isolated nodes. To facilitate this concept in our analysis, we introduce the notion of \emph{replicable graph problems}.

\begin{definition}
\label{def:replicable}
A graph problem is $R$-\textbf{replicable} if it satisfies the following property. For any
\begin{itemize}
\item graph $G \in \mathbb{G}$ with $|V(G)| \ge 2$,
\item output labeling $L: V(G) \rightarrow \Sigma$,
\item individual output label~$\ell \in \Sigma$, and
\item graph $\Gamma_G$ which is a disjoint union of at least $|V(G)|^R$ disjoint copies of $G$ (with the same IDs as $G$) and fewer than $|V(G)|$ isolated nodes (with the same ID as each other),
\end{itemize}
let output labeling $L'$ on $\Gamma_G$ be given by $L$ on each copy of $G$, and $\ell$ on each isolated node. Then, if $L'$ is valid on $\Gamma_G$, $L$ must be valid on $G$.
\end{definition}

Note that replicability is monotonic in $R$, i.e., if $P$ is $R$-replicable then $P$ is also $(R+1)$-replicable, since any $\Gamma_G$ satisfying the construction of $(R+1)$-replicability also satisfies the construction for $R$-replicability.

The definition of replicability may seem unnatural: it is designed to align with a specific construction needed in proof of \Cref{lem:sens} (and ultimately \Cref{thm:lb}). However, we argue that the vast majority of natural graph problems are replicable. We first show all that $r$-radius-checkable problems (and hence all \LCL problems \cite{CKP19,NS95}) are replicable;

\begin{lemma}
\label{lemma:LCL-replicable}
Any $r$-radius-checkable problem is $0$-replicable.
\end{lemma}

\begin{proof}
For any $r$-radius-checkable problem, the validity of the output of a connected component depends only on the IDs and topology of the connected component. Any $\Gamma_G$ satisfying the construction of $0$-replicability contains at least one copy of $G$ as a connected component. So, for the output on $L'$ on $\Gamma_G$ to be valid, the output $L$ must be valid on $G$.
\end{proof}

Further, a major strength of our framework is that \emph{most approximation problems are also replicable}. As an example, we show replicability for the problem of finding an independent set of size $\Omega(n/\Delta)$ (which is an $\Omega(1/\Delta)$-approximation of the maximum independent set), a problem for which we will later (in \Cref{sec:stable-nonstable-rand}) show a separation between component-stable and non-component-stable algorithms.

\begin{lemma}
\label{lem:IS-replicable}
The problem of finding an independent set of size $\Omega(n/\Delta)$ (on graphs with $\Delta\ge 1$) is $2$-replicable.
\end{lemma}

\begin{proof}
Let $c > 0$ be some fixed constant. Consider a graph $G$ on $n$ nodes, with maximum degree $\Delta$, and a graph $\Gamma_G$ consisting of $k \ge n^2$ copies of $G$ and fewer than $n$ isolated nodes. Note that $\Gamma_G$ has at least $kn$ nodes. Assume we have some output valid labeling $L'$ on $\Gamma_G$, which corresponds to an independent set of size at least $ckn/\Delta$, and in which each copy of $G$ is labeled identically, as is each isolated node. More than $ckn/\Delta-n$ of the nodes in the independent set must be in copies of $G$ (since there are fewer than $n$ nodes not in copies of $G$). Since each copy of $G$ is labeled identically, each must contain more than
\begin{align*}
    \frac{ckn/\Delta-n}{k} &\ge
    \frac{cn}{\Delta} - \frac{1}{n} \ge
    \frac{cn}{2\Delta}
\end{align*}
nodes in the independent set (for $n\ge 2/c$), and therefore the output on $G$ is a valid $\Omega(n/\Delta)$-independent set.
\end{proof}

Similarly, we have a related lemma for approximate matching, one of the central problems demonstrating the power of our framework summarized in \Cref{thm:lb} (which will yield also the conditional hardness of the approximate maximum matching problem on \MPC). The same arguments can also straightforwardly show that $\Omega(1)$-approximation of maximum matching and minimum vertex cover are $O(1)$-replicable.

\begin{lemma}
\label{lem:approx-matching-replicable}
The problem of finding an $\Omega(1)$-approximation of maximum matching is $2$-replicable.
\end{lemma}

\begin{proof}
To fit maximal matching into our vertex-labeling definition for problems, we characterize it as maximal independent set on the (normal) family of \emph{line graphs}, as discussed above. An $\Omega(1)$-approximation of maximum matching on an input graph corresponds to an $\Omega(1)$-approximate maximal independent set on its line graph. Let $G$ be such a line graph, on $n$ nodes, and let $c > 0$ be some fixed constant. We consider a graph $\Gamma_G$ consisting of $k \ge n^2$ copies of $G$ and fewer than $n$ isolated nodes. Again, $\Gamma_G$ has at least $kn$ nodes, and we note that (denoting $\Lambda(H)$ to be the size of the MIS of a graph $H$) $\Lambda(\Gamma_G) \ge k\Lambda(G)$. Assume we have some output valid labeling $L'$ on $\Gamma_G$, which corresponds to an independent set of size at least $c \Lambda(\Gamma_G)$, and in which each copy of $G$ is labeled identically, as is each isolated node. More than $c \Lambda(\Gamma_G)-n$ of the nodes in the independent set must be in copies of $G$ (since there are fewer than $n$ nodes not in copies of $G$). Since each of the $k$ copies of $G$ is labeled identically, each must contain more than
\begin{align*}
    \frac{c \Lambda(\Gamma_G)-n}{k} &\ge
	\frac{ck\Lambda(G)}{k} - \frac{1}{n} \ge
	\frac{c\Lambda(G)}{2}
\end{align*}
nodes in the independent set (for $n\ge 2/c$), and therefore the output on $G$ is a valid $\Omega(1)$-approximate MIS, corresponding to an $\Omega(1)$-approximate maximal matching in the input graph.
\end{proof}


\subsection{Algorithm definitions, and revised definition of component-stability}
\label{sec:component-stable-algorithms}

Once we have defined \emph{problems}, as considered in our paper, we may define \LOCAL and \MPC algorithms that solve them, and in particular, give a formal, amended definition of component-stable algorithms used in this paper, taking into account the discussion above.


\subsubsection{\LOCAL algorithms}
\label{subsubsec:LOCAL-algs}

Our formal definition of algorithms in the \LOCAL model used in this paper is as follows:

\paragraph{\textbf{Input.}}
\LOCAL algorithms receive as input an $n$-node graph $G$, with component-unique IDs for each node. Randomized algorithms also provide each node with access to a \emph{shared, unbounded, random seed $\mathcal{S}$}. Algorithms are provided with the exact value of the maximum degree $\Delta$, and an \emph{input size estimate} $N$ of $n$ such that $n\le N \le poly(n)$.\footnote{The reason we assume that only a polynomial estimate $N$ of $n$ is known here is that some \LOCAL lower bounds to which we wish to apply our lifting result only hold without exact knowledge of $n$ (e.g., that of \cite{KKSS19}). Most \LOCAL lower bounds, however, do hold under exact knowledge, and in these cases we can simply set $N=n$.}

The nodes of the input graph are the computational entities, and each initially has knowledge of its adjacent edges in $G$ (i.e., the IDs of their other endpoints). The computation proceeds in synchronous rounds, and in each round, a node may send an arbitrary message along each of its adjacent edges. At the termination of the algorithm, each node must give an output label from $\Sigma$.

\paragraph{\textbf{Output.}}
Correct deterministic algorithms must always provide a valid overall output labeling for the problem, on every output; randomized algorithms must give a valid labeling with probability at least $1-\frac 1N$, over the distribution of the random seed $\mathcal{S}$, for any input.

\paragraph{Shared randomness.}
Given that \MPC algorithms naturally allow shared randomness, it is important for our study of randomized \LOCAL algorithms to allow the nodes to have access to shared randomness too. The use of \emph{shared randomness} is non-standard in the \LOCAL model, where one typically assumes only private randomness. However, as shown by Ghaffari et al.\ \cite[Section~V]{GKU19}, many of the existing \LOCAL lower bounds can be extended (and without any asymptotic loss in their \LOCAL round complexity) also if the nodes have access to shared randomness. (Notice that the notion of shared randomness is only relevant to randomized algorithms, and hence, for deterministic complexity one can use the existing deterministic \LOCAL lower bounds without any constraints, as black box results.)


\subsubsection{\MPC algorithms}
\label{subsubsec:MPC-algs}

We use the standard definition of \MPC algorithms (see, e.g., \cite{ASSWZ18,BHH19,CLMMOS18,GGKMR18,GU19,KSV10}) amended to fit the framework of low-space \MPC{}s used in the paper.

\paragraph{\textbf{Input.}}
\MPC algorithms receive as input a legal $n$-node graph $G$, distributed arbitrarily over $poly(n)$ machines, each with local space $O(n^{\sparam})$ for some $\sparam<1$. Randomized algorithms also provide each node with access to a shared, random seed $\mathcal{S}$ of $poly(n)$ bits (again distributed arbitrarily among machines). We do not assume that maximum $\Delta$ or $n$ are given explicitly as input, but \MPC algorithms can determine them easily in $O(1)$ rounds, so we may assume knowledge thereof.

Computation proceeds in synchronous rounds, and in each round, a machine first perform an arbitrary local computations on its local data and then may send and receive a total of $O(n^{\sparam})$ information, divided between any other machines as desired. At the termination of the algorithm, each machine must give an output label from $\Sigma$ for each node it received in the initial distribution.

\paragraph{\textbf{Output.}}
Correct deterministic algorithms must always provide a valid overall output labeling for the problem, on every output; randomized algorithms must give a valid labeling with probability at least $1 - \frac{1}{n}$, over the distribution of the random seed $\mathcal{S}$, for any input.

\paragraph{Computation in \MPC algorithms.}
While we are mainly using the most standard setup of \MPC algorithms, closely following, e.g., \cite{ASSWZ18,BHH19,CLMMOS18,GGKMR18,GU19,KSV10}, occasionally we will use some features which (while often standard) are less commonly used.

The standard \MPC model assumes that in each synchronous rounds, each machine performs arbitrary local computations on its data (fitting its local memory of size $S = O(n^{\sparam})$) and then the machines simultaneously exchange messages, in a way that each machine is sender and receiver of up to $O(S)$ messages. While some papers also consider the sequential running time of any single \MPC machine in every round, the main focus of our study is primarily on the information theoretic aspects of understanding the round complexity in \MPC algorithms. (Notice that unbounded local computation assumption is standard in the classical distributed models as \LOCAL, \congest, and \congc.) As the result, while many of our algorithms perform only $\poly(n)$-time computations, occasionally we will allow \MPC machines to perform \emph{heavy local computations}, up to $2^{O(S)}$ local computations in a round; still, the space used on a single machine remains $S = O(n^{\sparam})$. Our results show that allowing such heavy computations might provide advantageous in the context of deterministic algorithms and derandomization, however they are not necessary to find examples of component-unstable deterministic algorithms which surpass component-stable conditional lower bounds. 

Furthermore, while typically one is concerned with the design of uniform \MPC algorithms, as it has been observed by Fish et al.\ \cite{FKLRT15}, the original setup of \MPC (e.g., \cite{KSV10}) leads naturally to the non-uniform model of computation. Most of the \MPC algorithms presented in our paper are uniform, but occasionally we use \emph{non-uniform algorithms}. In our setting, this means that the \MPC algorithm, on each single machine initially knows the number of nodes $n$ (or its estimation), and possibly different algorithms are used for different values of $n$. This can be also seen as having some non-uniform advice hardwired in the algorithms on individual \MPC machines (or as Boolean circuits; for more details, see, e.g., Section~7.1.1 in \cite{Vadhan12}).

Finally, some of the non-uniform \MPC algorithms we use are also \emph{non-explicit}. That is, we will be showing that there is a low-space \MPC algorithm for a specific task, but we will not be able to provide a procedure to explicitly design it (there is generally a brute-force procedure obvious from the proof, but one that requires at least exponential computation, and possibly also too much space to perform in low-space \MPC). In this paper, non-uniform and non-explicit \MPC algorithms will be occasionally used in the context of derandomization.


\subsubsection{Component-stable \MPC algorithms}
\label{subsubsec:component-stable-algorithms}

Now, after our discussion in \Cref{subsec:problems-with-def-comp-stab,subsec:graph-families,subsec:problem-def}, we are ready to provide a new definition of component-stable \MPC algorithms used in this paper.

\begin{definition}[\textbf{Component-stable \MPC algorithms}]
\label{def:component-stable-algorithms}
A randomized \MPC algorithm $A_{MPC}$ is \textbf{component-stable} if its output at any node $v$ is entirely, deterministically, dependent on the \emph{topology and IDs (but independent of names)} of $v$'s connected component (which we will denote $CC(v)$), $v$ itself, the exact number of nodes $n$ and maximum degree $\Delta$ in the \emph{entire} input graph, and the input random seed $\mathcal{S}$. That is, the output of $A_{MPC}$ at $v$ can be expressed as a deterministic function $A_{MPC}(CC(v),v,n,\Delta,\mathcal{S})$.

A deterministic \MPC algorithm $A_{MPC}$ is component-stable under the same definition, but omitting dependency on the random seed $\mathcal{S}$.
\end{definition}

Finally, let us state the main technical result demonstrating the power of our revised framework of component-stable \MPC algorithms, lifting unconditional lower bounds from the \LOCAL model to conditional lower bounds for low-space component-stable \MPC algorithms. The following theorem extends the main result in the component-stable algorithms framework due to Ghaffari et al. \cite[Theorem~I.4]{GKU19} to our framework and enhances it to include lower bounds against \emph{deterministic algorithms}, and lower bounds with \emph{dependency on maximum input degree $\Delta$}. Informally, similarly to \cite[Theorem~I.4]{GKU19}, \Cref{thm:lb} below states that, conditioned on the connectivity conjecture, for $O(1)$-replicable graph problems, any $T(N,\Delta)$-round lower bound in the \LOCAL model yields a $\Omega(\log T(N,\Delta))$-round lower bound for any low-space component-stable \MPC algorithm $\cA_{\MPC}$. Furthermore, the claim holds for both randomized and deterministic algorithms (deterministic algorithms were not studied in \cite{GKU19}).

The proof of \Cref{thm:lb} is given in \Cref{sec:lb-LOCAL->MPC} (the notion of constrained functions is defined in Definition \ref{def:constrained-functions}).

\begin{theorem}[\textbf{Lifting \LOCAL lower bounds to component-stable \MPC algorithms}]
\label{thm:lb}
Let $\mathcal{P}$ be a $O(1)$-replicable graph problem that has a $T(N,\Delta)$-round lower bound in the randomized \LOCAL model with shared randomness, for constrained function $T$, on graphs with input estimate $N$ and maximum degree $\Delta$, from some normal family $\mathcal{G}$. Suppose that there is a randomized $o(\log T(n,\Delta))$-round low-space component-stable \MPC algorithm $\cA_{\MPC}$ for solving $\mathcal{P}$ on legal $n$-node graphs with maximum degree $\Delta$ from $\mathcal{G}$, succeeding with probability at least $1-\frac 1n$. Then, there exists a low-space randomized \MPC algorithm $\cA^*$ that can distinguish one $n$-node cycle from two $\frac{n}{2}$-node cycles in $o(\log n)$ rounds, succeeding with probability at least $1-\frac 1n$.
	
The same holds if the \LOCAL lower bound and algorithm $\cA_{\MPC}$ are both deterministic (but the obtained algorithm $\cA^*$ remains randomized).
\end{theorem}


\subsection{Landscape of \MPC complexity classes and component-stability}
\label{subsec:complexity-classes}

In this section we define \MPC complexity classes considered in this paper. We study the \MPC model introduced in \Cref{sec:intro} and described in details \Cref{sec:component-stable-algorithms}, focusing on low-space \MPC{}s.

Let us begin with reminding the reader a fundamental obstacle to fully understand the computational complexity of problems in the low-space \MPC setting: a seminal work of Roughgarden et al.\ \cite{RVW18} showed that obtaining any unconditional lower bound in the low-space \MPC setting ultimately leads to breakthrough results in circuit complexity. In view of that, we will take a more modest approach and will rely on \emph{conditional} hardness results based on the widely believed \emph{connectivity conjecture}. The core hardness result considered in this paper is a revised framework lifting unconditional lower bounds from the \LOCAL model to conditional lower bounds for low-space component-stable \MPC algorithms (see \Cref{sec:prelim} and \Cref{sec:lb-LOCAL->MPC}, which amend the framework developed earlier in \cite{GKU19}). In particular, this framework can be used to obtain a number of deterministic and randomized lower bounds for low-space component-stable \MPC algorithms (these bounds are conditioned on the connectivity conjecture). Therefore, providing low-space non-component-stable \MPC algorithms that beat these bounds will demonstrate the \emph{conditional complexity gap} between low-space component-stable and non-component-stable \MPC algorithms --- which is the ultimate goal of this section.

\junk{
We will consider only graph problems. Let $\mathbb{G}$ be the collection of all input graph instances. We say that an \MPC algorithm $\cA$ has round complexity $T(n)$ with local space $S(n)$ on $M(n)$ machines, if for any $n \in \mathbb{N}$, for any input instance on $n$ nodes, it performs at most $T(n)$ rounds on \MPC with local space $S(n)$ and $M(n)$ machines. A deterministic algorithm $\cA$ solves a problem $\mathcal{P}$ if for any input instance $G \in \mathbb{G}$ it returns a correct solution for $\mathcal{P}$ (e.g., for decision problem, it is a correct \YES or \NO\Artur{In \Cref{subsec:problem-def}, we're assuming each nodes returns a label; hence one should rephrase the sentence.}). A randomized algorithm $\cA$ solves a problem $\mathcal{P}$ with probability $1-p(n)$, $0 \le p(n) \le 1$, if for any input instance $G \in \mathbb{G}$ with $n$ nodes, it returns a correct solution for $\mathcal{P}$ with probability at least $1-p(n)$, which by our convention in \Cref{subsubsec:MPC-algs} has $p(n) = \frac1n$. (For the sake of simplicity, we focus on Monte Carlo randomized algorithms, which have fixed time complexity but may have some probability $1-\frac1n$ to err and produce an inadmissible output.)
}

We focus on upper bound round complexity, local space, global space, and success probability which depend only on the number $n$ of the graph's nodes in the input instance from $\mathbb{G}$. Since our main focus is to study the low-space \MPC regime, we consider the following definitions for \MPC algorithms and \MPC component-stable algorithms (see \Cref{sec:component-stable-algorithms}).

\begin{definition}[\textbf{\MPCDetNonStable}]
Denote by $\MPCDetNonStable(T(n))$ the class of all graph problems for which there is a (possibly non-uniform and non-explicit) deterministic low-space \MPC algorithm $\cA$, such that for some positive constant $c$, algorithm $\cA$ has round complexity at most $c \cdot T(n)$.
\end{definition}

\begin{definition}[\textbf{\MPCDetStable}]
Denote by $\MPCDetStable(T(n))$ the subclass of $\MPCDetNonStable(T(n))$ restricted to component-stable algorithms. That is, $\MPCDetStable(T(n))$ is the class of all graph problems for which there is a (possibly non-uniform and non-explicit) deterministic low-space component-stable \MPC algorithm $\cA$, such that for some positive constant $c$, algorithm $\cA$ has round complexity at most $c \cdot T(n)$.
\end{definition}

\begin{definition}[\textbf{\MPCRandNonStable}]
Denote by $\MPCRandNonStable(T(n))$ the class of all graph problems $\mathcal{P}$ for which there is a randomized low-space \MPC algorithm $\cA$, such that for some positive constant $c$, algorithm $\cA$ solves $\mathcal{P}$ with probability $1 - \frac1n$ and has round complexity at most $c \cdot T(n)$.
%
\end{definition}

\begin{definition}[\textbf{\MPCRandStable}]
$\MPCRandStable(T(n))$ is the subclass of $\MPCRandNonStable(T(n))$ restricted to component-stable algorithms. That is, $\MPCRandStable(T(n))$ is the class of all graph problems $\mathcal{P}$ for which there is a randomized low-space component-stable \MPC algorithm $\cA$, such that for some positive constant $c$, algorithm $\cA$ solves $\mathcal{P}$ with probability $1 - \frac1n$ and has round complexity at most $c \cdot T(n)$.
%
\end{definition}

Clearly, observe that for any function $T(n)$, we have both $\MPCDetStable(T(n)) \subseteq \MPCDetNonStable(T(n))$ and $\MPCRandStable(T(n)) \subseteq \MPCRandNonStable(T(n))$. However, as mentioned in Introduction, it has been informally argued (for example, by Ghaffari et al.\ \cite{GKU19}) that most, if not all, \MPC algorithms are or can be easily made component-stable. Thus one could imagine that the pairs of sets $\MPCDetStable(T(n))$ and $\MPCDetNonStable(T(n))$, and $\MPCRandStable(T(n))$ and $\MPCRandNonStable(T(n))$ are in fact identical, especially in the regime of small $T(n)$. However, we will demonstrate that (assuming the connectivity conjecture) this is not the case, and that some low-space general \MPC algorithms can do better than their component-stable counterparts.

We begin with the study of \emph{deterministic} complexity. Our first lemma shows that, informally, $\MPCDetStable \ne \MPCDetNonStable$ (when omitting functions $T(n)$ during complexity class comparisons, we mean that there \emph{exists} some $T(n)$ for which the comparison holds, when both classes are parameterized by $T(n)$). In \Cref{subsec:sinkless-orientation} we consider a sinkless orientation problem for which, on one hand, (assuming the connectivity conjecture) there is no deterministic low-space component-stable \MPC algorithm running in $o(\log\log_\Delta n)$ rounds (\Cref{thm:sinkless-cs}) and, on the other hand, which can be solved by a low-space deterministic \MPC algorithm running in $poly(\Delta)+O(\log\log \log n)$ rounds (\Cref{thm:sinkless-ub}), surpassing the component-stable lower bound for $\Delta = \log^{o(1)}\log n$. In fact, these results hold even for forests. Similar results are also shown for some variants of edge-coloring (\Cref{thm:ecoloring-1,thm:ecoloring-2}) and vertex-coloring (\Cref{thm:vcoloring-1,thm:vcoloring-2}). Further, while the main deterministic upper bounds here use heavy local computation, for bounded degree graphs their local computation is $\poly(n)$, demonstrating that \emph{component-instability helps for deterministic algorithms even using polynomial computation}. (In \Cref{subsec:examps-complexity-separation}, we demonstrate a similar result for the class of \LOCAL \emph{extendable} algorithms, showing that instability also helps in deterministic algorithms for two cornerstone problems in low-space MPC: maximal independent set and maximal matching.)
This gives the following.

\begin{theorem}
\label{thm:complexity-SDet-Det}
There is some function $T(n)$ such that, conditioned on the connectivity conjecture,
\begin{align*}
    \MPCDetStable(T(n)) \subsetneqq \MPCDetNonStable(T(n))
	\enspace.
\end{align*}
\end{theorem}

Next, we move to the study of \emph{randomized} algorithms and show that, informally, $\MPCRandStable \ne \MPCRandNonStable$. This result follows directly from our \Cref{thm:approx-IS}, which shows that some variant of the independent set problem has a deterministic low-space non-stable constant-rounds \MPC algorithm and conditioned on the connectivity conjecture, there is no $o(\log\log^* n)$-rounds low-space component-stable \MPC algorithm that succeeds with probability at least~$1 - \frac1n$.

\begin{theorem}
\label{thm:complexity-SRan-Ran}
There is some function $T(n)$ such that, conditioned on the connectivity conjecture,
\begin{align*}
    \MPCRandStable(T(n)) \subsetneqq \MPCRandNonStable(T(n))
	\enspace.
\end{align*}
\end{theorem}

Next, we provide a conditional separation between randomized and deterministic algorithms (\Cref{thm:instability-provably-helps}). This separation follows by combining (i) the conditional lifting for randomized component-stable algorithms and deterministic component-stable algorithm using the framework in \Cref{thm:lb} with (ii) local problems for which there is provable gap in their randomized and deterministic \LOCAL complexity (e.g., \cite{BBHORS19,CQW+20,CKP19}). This yields the following.

\begin{theorem}
\label{thm:complexity-SDet-SRan}
There is some function $T(n)$ such that, conditioned on the connectivity conjecture,
\begin{align*}
    \MPCDetStable(T(n)) \subsetneqq \MPCRandStable(T(n))
	\enspace.
\end{align*}
\end{theorem}

Further, in \Cref{sec:derand}, we study the power of deterministic low-space component-unstable \MPC algorithms and their relationship to the randomized ones. In Lemma \ref{lem:det-alg-many-machines} we prove that if there is a randomized \MPC algorithm that solves a graph problem $\cP$ on $n$-node graphs with maximum degree $\Delta$ in $T(n,\Delta)$ rounds, then one can derandomized such algorithm to solve the same problem in $O(T(n,\Delta))$ rounds. The resulting deterministic \MPC algorithm is component-unstable, non-uniform, and non-explicit, and it has the same local space as that of the randomized \MPC algorithm, and uses an $O(n^2)$-factor more machines. This yields the following.

\begin{theorem}
\label{thm:complexity-Det-Ran}
$\MPCDetNonStable(T(n)) = \MPCRandNonStable(T(n))$.
\end{theorem}

Let us summarize the results in this section, assuming the connectivity conjecture, and allowing non-uniform \MPC algorithms. Our study demonstrates that for low-space \MPC{}s, component-unstable algorithms are provably stronger than their component-stable counterparts, both for deterministic and randomized algorithms (\Cref{thm:complexity-SDet-Det,thm:complexity-SRan-Ran}). Further, we see that for component-stable algorithms, randomized algorithms are provably stronger than their deterministic counterparts (\Cref{thm:complexity-SDet-SRan}), however, for arbitrary (possibly component-unstable) algorithms this is not the case: any randomized algorithm can be efficiently simulated by a deterministic one (Lemma \ref{lem:det-alg-many-machines} and \Cref{thm:complexity-Det-Ran}).


%% file: appendix-lift.tex


In this section, we present a framework lifting unconditional lower bounds from the \LOCAL model to conditional lower bounds for low-space component-stable \MPC algorithms, extending and revising the analysis of Ghaffari et al. \cite{GKU19} to prove \Cref{thm:lb}.

While on a high level, our analysis follows closely the approach from \cite{GKU19}, our arguments diverge in several \emph{subtle} but \emph{crucial} places. On a technical level, we rely on the central but also quite subtle notions of \emph{replicable graph problems}, \emph{normal graph families}, \emph{$(D, \eps, n, \Delta)$-sensitive \MPC algorithms}, and a revised definition of \emph{component stability} (see Definition \ref{def:component-stable-algorithms}). The two major reasons behind these changes are:
\begin{itemize}
\item to make the arguments robust against the issues we have identified concerning component-stability, and incorporate the definitional changes that these issues necessitated, and
\item to extend the arguments to include lower bounds against \emph{deterministic algorithms}, and lower bounds with \emph{dependency on maximum input degree $\Delta$}.
\end{itemize}

After introducing some useful notation in \Cref{subsec:basic-def}, we show in \Cref{subsec:sens} that in our setting, for $R$-replicable graph problems for normal graph families, lower bounds for \LOCAL algorithms imply the existence of some graphs which can be distinguished by component-stable \MPC algorithms only by relying on non-local information. Then, in \Cref{subsec:stconn}, we apply this non-locality to provide a conditional \MPC lower bound for component-stable \MPC algorithms for $O(1)$-replicable graph problems in our setting (for normal graph families). In particular, conditioned on the connectivity conjecture, our main result (\Cref{thm:lb}) states, informally, that for $O(1)$-replicable graph problems, any $T(N,\Delta)$-round lower bound in the \LOCAL model yields a $\Omega(\log T(N,\Delta))$-round lower bound for any low-space component-stable \MPC algorithm $\cA_{\MPC}$. Furthermore, the claim holds for both randomized and deterministic algorithms.
%


\subsection{Basic definitions: \emph{normal graph families} and \emph{sensitive \MPC algorithms}}
\label{subsec:basic-def}

\begin{definition}
Two connected graphs $G = (V, E)$ and $G' = (V', E')$, with \emph{center} nodes $v \in V$ and $v'\in V'$ respectively, are \textbf{$D$-radius-identical} if the topologies and node IDs (but not necessarily names) of the $D$-radius balls around $v$ and $v'$ are identical.
\end{definition}

Our next definition (related to \cite[Definition~III.1]{GKU19}, though set up in our framework), formalizes the notion of \MPC algorithms depending on non-local information in the graph.

\begin{definition}[\textbf{$(D, \eps, n, \Delta)$-sensitive \MPC algorithms}]
\label{def:sensitive-MPC-algorithms}
For integers $D,n,\Delta \ge 0$ and some $\eps \in [0,1]$, a component-stable \MPC algorithm $\cA$ for some graph problem is called \textbf{$(D, \eps, n,\Delta)$-sensitive} with respect to two $D$-radius-identical centered graphs $G = (V, E)$ and $G' = (V', E')$, with centers $v \in V$ and $v'\in V'$, if the probability (over choice of seed $\mathcal{S}$) that $\cA(G, v, n, \Delta, \mathcal{S}) \ne \cA(G', v', n, \Delta, \mathcal{S})$ is at least $\eps$.
\end{definition}

We can apply this definition also to deterministic component-stable algorithms. Since these do not depend on the random seed, any $(D, \eps, n, \Delta)$-sensitive deterministic  component-stable \MPC algorithm with $\eps > 0$ is $(D, 1, n, \Delta)$-sensitive.

As one final point of notation, in this section our arguments will often involve several different graphs. For clarity, for any graph $G$ we will denote by $n_G$ its number of nodes.


\subsection{\LOCAL hardness yields indistinguishability of graphs locally}
\label{subsec:sens}

Our next lemma (cf. \cite[Lemma~III.1]{GKU19}) shows that for any $R$-replicable graph problem, a lower bound for a normal graph family for any \LOCAL algorithm implies some useful property for any component-stable \MPC algorithm for that problem: to distinguish some graphs from that normal family one needs to rely on non-local information, in the sense of Definition~\ref{def:sensitive-MPC-algorithms}.

\begin{lemma}
\label{lem:sens}
For any $N, \Delta, R \in \nat$, let $\mathcal{P}$ be an $R$-replicable graph problem for which there is no $T(N,\Delta)$-round \LOCAL algorithm (with shared randomness) that solves $\mathcal{P}$ on graphs from some normal family $\mathcal{G}$ of maximum degree $\Delta$, with input size estimate $N$ (i.e., satisfying $n_H \le N \le poly(n_H)$), and with probability at least $1-\frac1N$.

Suppose there is a component-stable \MPC algorithm $\cA_{\MPC}$ that solves $\mathcal{P}$ on all graphs $G \in \mathcal{G}$, with probability at least $1 - \frac{1}{n_G}$. Then, there are two $T(N,\Delta)$-radius-identical centered graphs $G, G' \in \mathcal{G}$ with at most $N$ nodes and maximum degree (exactly) $\Delta$, such that $\cA_{\MPC}$ is $(T(N,\Delta), \frac{1}{4N^2},N^{R+2},\Delta)$-sensitive with respect to $G,G'$.

The same claim holds if the \LOCAL lower bound and algorithm $\cA_{\MPC}$ are both \emph{deterministic}.
%
%
%
\end{lemma}

\begin{proof}
The proof is by contradiction. Denote $D = T(N,\Delta)$ and $\eps = \frac{1}{4N^2}$. Let us assume that there are no two $D$-radius-identical centered graphs $G, G' \in \mathcal{G}$ with at most $N$ nodes and maximum degree $\Delta$ such that the given \MPC algorithm $\cA_{\MPC}$ is $(D, \eps, N^{R+2}, \Delta)$-sensitive with respect to $G$ and $G'$. We use this fact to construct a $D$-round randomized \LOCAL algorithm $\cA_{\LOCAL}$ to solve $\mathcal{P}$ on graphs with valid size estimate $N$ with probability at least $1 - \frac1N$, which is a contradiction to the assumption of the lemma.


\paragraph{Design of the \LOCAL algorithm.}

The input for $\cA_{\LOCAL}$ is an unbounded random seed $\mathcal{S}$, a graph $H \in \mathcal{G}$, integers $N$, $\Delta$, and the assignment of the IDs to the nodes of $H$, such that:
\begin{itemize}
\item $N$ is the input size estimate, that is, $n_H \le N \le poly(n_H)$,
\item $\Delta$ is the maximum degree of $H$,
\item all IDs are in $[\poly(N)]$ and are component-unique.
\end{itemize}

\junk{
The input for $\cA_{\LOCAL}$ is an unbounded random seed $\mathcal{S}$ and an arbitrary graph $H \in \mathcal{G}$ with:
\begin{itemize}
\item a valid input size estimate $N$ (i.e., $n_H \le N$),
\item component-unique IDs in $[\poly(N)]$,
\item maximum degree exactly $\Delta$.
\end{itemize}
}

For a fixed $H$, let us take an arbitrary node $v$ and describe how to determine $v$'s output under $\cA_{\LOCAL}$, on a particular input $\mathcal{I} = \langle H, \mathcal{S} \rangle$. Node $v$ first collects its $D$-radius ball $B_D(v)$ in $H$. It then considers the family $\mathcal{H}_v$ of all graphs with at most $ N$ nodes, centered at $v$, with component-unique IDs from $[\poly(N)]$, and maximum degree $\Delta$ that are $D$-radius-identical to $B_D(v)$, i.e., the set of possible inputs to $\cA_{\LOCAL}$ for which $v$ would see $B_D(v)$ as its $D$-radius ball. For each $G\in \mathcal{H}_v$, $v$ creates a \emph{simulation graph $\Gamma_G$} (see Definition \ref{def:replicable}) as follows: $\Gamma_G$ consists of $\lfloor \frac{N^{R+2}}{n_G} \rfloor \ge n_G^R$ disjoint copies of $G$. One of these copies is arbitrarily designated the \emph{``true'' copy}, and its nodes use the IDs of $G$ as their own IDs and names. All other copies use the same IDs, but as names they use arbitrary unique values in $[\poly(N)]$. We also add enough isolated nodes to raise the number of nodes to exactly $N^{R+2}$ (i.e., $N^{R+2} - n_G\lfloor \frac{N^{R+2}}{n_G} \rfloor < n_G$ isolated nodes), sharing the same arbitrary ID, and with arbitrary unique names, in $[\poly(N)]$. Note that this construction corresponds closely to Definition \ref{def:replicable} of problem replicability, and is the reason we chose such a definition.

In order to determine the output of $\cA_{\LOCAL}$ at $v$, we will be running $\cA_{\MPC}$ on the simulation graph $\Gamma_G$. For the supply of randomness, denote $\mathcal{S}'$ to be $\mathcal{S}$ curtailed to the length required by $\cA_{\MPC}$, i.e., the first $\poly(N)$ bits. Since $\cA_{\MPC}$ is component-stable, its output at $v$ on $\Gamma_G$, using seed $\mathcal{S}'$, is $\cA_{\MPC}(G,v,N^{R+2},\Delta,\mathcal{S}')$.

Then, we set
\begin{align*}
    \cA_{\LOCAL} (H,v,\mathcal{S}) &:=
    \arg \max_{x} |\{G\in \mathcal{H}_v: \cA_{\MPC}(G,v,N^{R+2},\Delta,\mathcal{S}')=x \}|
    \enspace,
\end{align*}
i.e., $v$'s output label output by $\cA_{\LOCAL}$ is that returned by the most $\cA_{\MPC}$ simulations.

\paragraph{Correctness of the \LOCAL algorithm.}

We use our assumption that there are no two $D$-radius-identical centered graphs $G, G' \in \mathcal{G}$ with at most $N$ nodes and maximum degree $\Delta$ for which $\cA_{\MPC}$ is $(D, \eps, N^{R+2}, \Delta)$-sensitive with respect to $G$ and $G'$. Hence, for any $G, G' \in \mathcal{H}_v$, with probability at least $1-\eps$ over choice of $\mathcal{S}$, $\cA_{\MPC}(G, v, N^{R+2}, \Delta, \mathcal{S}')= \cA_{\MPC}(G', v, N^{R+2}, \Delta, \mathcal{S}')$. In particular, since the actual input graph $H$ is certainly in $\mathcal{H}_v$,
\begin{align*}
    \mathbf E_{\mathcal{S}}\left[ \frac {|\{G\in \mathcal{H}_v:\cA_{\MPC}(G,v,N^{R+2},\Delta,\mathcal{S}')\ne \cA_{\MPC}(H,v,N^{R+2},\Delta,\mathcal{S}') \}|}{|\mathcal{H}_v|}\right]
        & \le
    \eps
    \enspace.
\end{align*}

Then, for the random seed $\mathcal{S}$, the probability that
\begin{align*}
    \frac {|\{G\in \mathcal{H}_v:\cA_{\MPC}(G,v,N^{R+2},\Delta,\mathcal{S}') \ne \cA_{\MPC}(H,v,N^{R+2},\Delta,\mathcal{S}') \}|}{|\mathcal{H}_v|}
        &\ge
    \frac12
    \enspace.
\end{align*}

is at most $2\eps$ by Markov's inequality. We therefore have
\begin{align*}
    \Prob{\cA_{\LOCAL} (H, v, \mathcal{S}) =
        \cA_{\MPC}(H, v, N^{R+2}, \Delta, \mathcal{S}')}
        &\ge
    1 - 2 \eps
    \enspace.
\end{align*}

Taking a union bound over all nodes $v \in V(H)$, we have that
\begin{align*}
    \Prob{\text{for all $v \in V(H)$: }
        \cA_{\LOCAL} (H,v,\mathcal{S}) =
        \cA_{\MPC}(H,v,N^{R+2},\Delta,\mathcal{S}')}
        &\ge
    1 - 2 N \eps
    \enspace.
\end{align*}

In this case, the overall output of $\cA_{\LOCAL} (H, \mathcal{S})$ is equal to the output of $\cA_{\MPC}$ on the true copy of $H$ when run on $\Gamma_H$ with seed $S'$. Since $\Gamma_H \in \mathcal{G}$ (as it is a disjoint union of members of $\mathcal{G}$, see Definition \ref{def:normal-graphs}), $\cA_{\MPC}$ returns a valid output on $\Gamma_H$ with probability $1-1/n_{\Gamma_H} = 1- N^{-(R+2)}$. Note that $\Gamma_H$ is a disjoint union of at least $n_H^{R}$ copies of $H$ and fewer than $n_H$ extra isolated nodes. This satisfies the construction in the definition of $R$-replicability.\footnote{When $n_H \ge 2$; for $n_H \le 1$ Lemma \ref{lem:sens} is trivial since there can be no \LOCAL lower bounds on such graphs.} Since $\cA_{\MPC}$ is component-stable, it must return the same output labeling in all copies of $H$ in $\Gamma_H$ (since they share the same topology and IDs), and likewise must return the same output label on all isolated nodes. So, the output labeling given on $\Gamma_H$ is of the form required in the replicability definition, and so $\cA_{\MPC}$ returns a valid output on the true copy (and, indeed, on all copies) of $H$ in $\Gamma_H$.

Then, by a union bound, $\cA_{\LOCAL}$'s output is correct with probability at least $1-2N\eps - N^{-(R+2)} \ge 1 - \frac1N$, reaching a contradiction to our assumption of a \LOCAL lower bound. This proves the randomized part of the lemma.

To prove the deterministic analogue, we perform exactly the same construction, except that $\cA_{\LOCAL}$ does not have the random seed $\mathcal{S}$, and does not pass it to $\cA_{\MPC}$. Since our algorithm $\cA_{\LOCAL}$ is now deterministic, we have $\cA_{\LOCAL}(H,v) = \cA_{\MPC}(H,v,N^{R+2},\Delta)$ with certainty for all $v$, i.e., $\cA_{\LOCAL}$'s output is identical to $\cA_{\MPC}$'s output on an $N^{R+2}$-node graph containing $H$ as a connected component. Since $\cA_{\MPC}$ is also deterministic, this is a valid output on $H$ with certainty. We have then constructed a deterministic \LOCAL algorithm solving $\mathcal{P}$, contradicting our initial assumption. So, we must instead have that $\cA_{\MPC}$ is $(T(N,\Delta), 1,N^{R+2},\Delta)$-sensitive, and therefore, since it is deterministic, $(T(N,\Delta), 1,N^{R+2},\Delta)$-sensitive.
\end{proof}

\junk{
\begin{proof}
Denote $D = T(N,\Delta)$ and $\eps = \frac{1}{4N^2}$. Let us assume that there are no two centered graphs $G$ and $G'$ with $N$ nodes and maximum degree $\Delta$ such that the given \MPC algorithm $\cA_{\MPC}$ is $(D, \eps, N^{R+2}, \Delta)$-sensitive with respect to $G$ and $G'$. We use this fact to construct a $D$-round randomized \LOCAL algorithm $\cA_{\LOCAL}$ to solve $\mathcal{P}$ on graphs with valid size estimate $N$ with probability at least $1 - 1/N$, which is a contradiction to the assumption of the lemma.

The input for $\cA_{\LOCAL}$ will be a graph $H$ with:
\begin{itemize}
	\item a valid input size estimate $N$ (i.e., $n_H\le N$),
	\item component-unique IDs in $[\poly(N)]$,
	\item maximum degree exactly $\Delta$.
\end{itemize}

$\cA_{\LOCAL}$ also receives an unbounded random seed $\mathcal{S}$.

We now fix a node $v$ to consider, and describe how we determine $v$'s output under $\cA_{\LOCAL}$, on a particular input $I=(\mathcal{S},H)$. Node $v$ first collects its $D$-radius ball $B_D(v)$ in $H$. It then considers the family $\mathcal{H}_v$ of all graphs with at most $ N$ nodes, centered at $v$, with component-unique IDs from $[\poly(N)]$, and maximum degree $\Delta$ that are $D$-radius-identical to $B_D(v)$, i.e., the set of possible inputs to $\cA_{\LOCAL}$ for which $v$ would see $B_D(v)$ as its $D$-radius ball. For each $G\in \mathcal{H}_v$, $v$ creates a simulation graph $\Gamma_G$ as follows: $\Gamma_G$ consists of $\lfloor \frac{N^{R+2}}{n_G} \rfloor \ge n_G^r$ disjoint copies of $G$. One of these copies is arbitrarily designated the `true' copy, and its nodes use the IDs of $G$ as their own IDs and names. All other copies use the same IDs, but as names they use arbitrary unique values in $[\poly(N)]$. We also add enough isolated nodes to raise the number of nodes to exactly $N^{R+2}$ (i.e., $N^{R+2}-n_G\lfloor \frac{N^{R+2}}{n_G} \rfloor < n_G$ isolated nodes), sharing the same arbitrary ID, and with arbitrary unique names, in $[\poly(N)]$.

We will be running $\cA_{\MPC}$ on this simulation graph $\Gamma_G$. For the supply of randomness, denote $\mathcal{S}'$ to be $\mathcal{S}$ curtailed to the length required by $\cA_{\MPC}$, i.e., the first $\poly(N)$ bits. Since $\cA_{\MPC}$ is component-stable, its output at $v$ on $\Gamma_G$, using seed $\mathcal{S}'$, is $\cA_{\MPC}(G,v,N^{R+2},\mathcal{S}')$.

Then, we set

\begin{align*}
    \cA_{\LOCAL} (H,v,\mathcal{S}) &:=
    \arg \max_{x} |\{G\in \mathcal{H}_v: \cA_{\MPC}(G,v,N^{R+2},\Delta,\mathcal{S}')=x \}|
    \enspace,
\end{align*}
i.e., $v$'s output label is that returned by the most $\cA_{\MPC}$ simulations.

\paragraph{Correctness of the \LOCAL algorithm.}

By our assumption that $\cA_{\MPC}$ is not sensitive, we now have that for any $G, G'\in \mathcal{H}_v$, with probability at least $1-\eps$ over choice of $\mathcal{S}$, $\cA_{\MPC}(G,v,N^{R+2},\mathcal{S}')= A_{\MPC}(G',v,N^{R+2},\Delta,\mathcal{S}')$. So, in particular, since the actual input graph $H$ is certainly in $\mathcal{H}_v$,

\begin{align*}
\mathbf E_{\mathcal{S}}\left[ \frac {|\{G\in \mathcal{H}_v:A_{\MPC}(G,v,N^{R+2},\Delta,\mathcal{S}')\ne A_{\MPC}(H,v,N^{R+2},\Delta,\mathcal{S}') \}|}{|\mathcal{H}_v|}\right]\le\eps\enspace.
\end{align*}

Then, for random $\mathcal{S}$, the probability that
\[ \frac {|\{G\in \mathcal{H}_v:A_{\MPC}(G,v,N^{R+2},\Delta,\mathcal{S}')\ne A_{\MPC}(H,v,N^{R+2},\Delta,\mathcal{S}') \}|}{|\mathcal{H}_v|}\ge \frac 12\] is at most $2\eps$ by Markov's inequality. We therefore have \[\Prob{A_{\LOCAL} (H,v,\mathcal{S}) = A_{\MPC}(H,v,N^{R+2},\Delta,\mathcal{S}')}\ge 1-2\eps\enspace.\]

Taking a union bound over all nodes $v \in V(H)$, we have that with probability at least $1-2N\eps$ over $\mathcal{S}$, $\cA_{\LOCAL} (H,v,\mathcal{S}) = A_{\MPC}(H,v,N^{R+2},\Delta,\mathcal{S}')$ for all $v$. In this case, the overall output of $\cA_{\LOCAL} (H,\mathcal{S})$ is equal to the output of $\cA_{\MPC}$ on the true copy of $H$ when run on $\Gamma_H$ with seed $S'$. Since $\Gamma_H\in\mathcal{G}$ (as it is a disjoint union of members of $\mathcal{G}$),   $\cA_{\MPC}$ returns a valid output on $\Gamma_H$ with probability $1-1/n_{\Gamma_H} = 1- N^{-(R+2)}$. Note that $\Gamma_H$ is a disjoint union of at least $n_H^{R}$ copies of $H$ and fewer than $n_H$ extra isolated nodes. This satisfies the construction in the definition of $R$-replicability.\footnote{when $n_H \ge 2$; for $n_H=1$ the lemma is trivial since there can be no \LOCAL lower bounds on such graphs.} Since $\cA_{\MPC}$ is component-stable, it must return the same output labeling in all copies of $H$ in $\Gamma_H$ (since they share the same topology and IDs), and likewise must return the same output label on all isolated nodes. So, the output labeling given on $\Gamma_H$ is of the form required in the replicability definition, and so $\cA_{\MPC}$ returns a valid output on the true copy (and, indeed, on all copies) of $H$ in $\Gamma_H$.

Then, by a union bound, $\cA_{\LOCAL}$'s output is correct with probability at least $1-2N\eps - N^{-(R+2)} \ge 1-1/N$, reaching a contradiction to our assumption of a \LOCAL lower bound. This proves the randomized part of the lemma.

To prove the deterministic analogue, we perform exactly the same construction, except that $\cA_{\LOCAL}$ does not have the random seed $\mathcal{S}$, and does not pass it to $\cA_{\MPC}$. Since our algorithm $\cA_{\LOCAL}$ is now deterministic, we have $\cA_{\LOCAL}(H,v) = A_{\MPC}(H,v,N^{R+2},\Delta)$ with certainty for all $v$, i.e., $\cA_{\LOCAL}$'s output is identical to $\cA_{\MPC}$'s output on an $N^{R+2}$-node graph containing $H$ as a connected component. Since $\cA_{\MPC}$ is also deterministic, this is a valid output on $H$ with certainty. We have then constructed a deterministic \LOCAL algorithm solving $\mathcal{P}$, contradicting our initial assumption. So, we must instead have that $\cA_{\MPC}$ is $(T(N,\Delta), 1,N^{R+2},\Delta)$-sensitive, and therefore, since it is deterministic, $(T(N,\Delta), 1,N^{R+2},\Delta)$-sensitive.
\end{proof}
}


\subsection{Conditional \MPC lower bounds from \LOCAL (proving Theorem \ref{thm:lb})}
\label{subsec:stconn}

In this section we apply Lemma \ref{lem:sens} to obtain conditional lower bounds for component-stable \MPC algorithms for $O(1)$-replicable graph problems through their hardness in the \LOCAL model. Our main result, \Cref{thm:lb}, can be seen as a revision and extension of the main arguments due to Ghaffari et al. \cite{GKU19} in their Theorem~I.4, applied to low-space component-stable \MPC algorithms for $O(1)$-replicable graph problems and for normal graph families.

Since the complexity of our algorithms depends on two parameters, input size estimate $N$ and maximum degree $\Delta$, we define the type of functions for which we will study lower bounds:

\begin{definition}
\label{def:constrained-functions}
We call a function $T: \nat \times \nat \rightarrow \nat$ \textbf{constrained} if $T(N,\Delta) = O(\log^\gamma N)$ for all $\Delta \le N$ and $\gamma \in (0,1)$, and $T(N^c,\Delta) \le c \cdot T(N,\Delta)$ for any $c\ge 1$.
\end{definition}

This is similar to the corresponding requirement $T(N) = O(\log^\gamma N)$ in \cite{GKU19}, but allowing for dependency on $\Delta$. The definition of constrained functions further incorporates a `smoothness' condition, which is also implicitly required by \cite{GKU19} (there, it is used that when $T(N) = O(\log^\gamma N)$, $T(\poly(N)) = O(T(N))$; while this is true for smooth functions, it is not technically true in general).\footnote{As a counterexample, let $T(N)$ be a tower of 2s of height $\log^*N - 3$. Then, $T(N) = O(\log\log N)$, but $\limsup\limits_{N\to \infty}\frac{T(N+1)}{T(N)} = \infty$.}
%

With the notion of constrained functions at hand, we can leverage Lemma \ref{lem:sens} to relate the complexity of \LOCAL algorithms for $R$-replicable graph problems in normal graph families to the connectivity conjecture. In particular, given the conjectured hardness of this problem in the \MPC model, this will imply (see \Cref{thm:lb} for a precise statement) a lower bound of $\Omega(\log T(n,\Delta))$-rounds for any low-space component-stable \MPC algorithm for any $O(1)$-replicable graph problem in normal graph families with no $T(N,\Delta)$-round \LOCAL algorithm.

We begin with an auxiliary lemma that relates Lemma \ref{lem:sens} to the complexity of low-space \MPC algorithms for the $D$-diameter $s$-$t$ connectivity problem.\footnote{The \emph{$D$-diameter $s$-$t$ connectivity problem} is formally defined in \cite[Definition~IV.1]{GKU19} as follows: Given two special nodes $s$ and $t$ in the graph, the goal is to determine whether $s$ and $t$ are connected or not, in a way that provides the following guarantee: If $s$ and $t$ are in the same connected component and that component is a path of length at most $D$ with $s$ and $t$ at its two endpoints, then the algorithm should output \YES, with probability $1 - \frac{1}{\poly(n)}$. If $s$ and $t$ are in different connected components, the algorithm should output \NO, with probability $1 - \frac{1}{\poly(n)}$. In other cases, the output can be arbitrary.}

\begin{lemma}
\label{lem:stconn}
Let $\mathcal{P}$ be a $O(1)$-replicable graph problem that has a $T(N,\Delta)$-round lower bound in the randomized \LOCAL model with shared randomness, for constrained function $T$, on graphs with input size estimate $N$ and maximum degree $\Delta$, from some normal family $\mathcal{G}$. Suppose that there is a randomized $o(\log T(n,\Delta))$-round low-space component-stable \MPC algorithm $\cA_{\MPC}$ for solving $\mathcal{P}$ on legal $n$-node graphs with maximum degree $\Delta$ from $\mathcal{G}$, succeeding with probability at least $1-\frac1n$. Then, there exists a low-space randomized \MPC algorithm $B_{\text{st-conn}}$ with round complexity $o(\log T(n,\Delta))$ for $T(n,\Delta)$-diameter $s$-$t$ connectivity on graphs with $n$ nodes, succeeding with probability at least $1-\frac1n$.
%
	
The same holds if the \LOCAL lower bound and algorithm $\cA_{\MPC}$ are both deterministic (but the obtained \MPC algorithm $B_{\text{st-conn}}$ is still randomized).
\end{lemma}

One point of note in the theorem statement is that we do not make any constraint on the maximum degree of the $s$-$t$ connectivity instance; indeed, for that problem any node of degree over $2$ can be immediately discarded, so maximum degree is largely irrelevant. $\Delta$ appears only in the expression for the diameter $D$.

\begin{proof}
Let $c$ be a sufficiently large constant. We will require that $\mathcal{P}$ is $\frac{c}{\sparam}$-replicable; since replicability is monotone, this will be the case when $c$ is sufficiently large. We then construct our randomized \MPC algorithm for $T(n,\Delta)$-diameter $s$-$t$ connectivity $B_{\text{st-conn}}$ as follows:

We are given an $n$-node input graph $H$, and $\poly(n)$-bit random seed $\mathcal{S}$, for $T(n,\Delta)$-diameter $s$-$t$ connectivity. We set $N=n^{\sparam/2}$, denote $D:= T(N,\Delta)$, and apply Lemma \ref{lem:sens} to obtain\footnote{Lemma \ref{lem:sens} shows the \emph{existence} of such graphs, but a uniform \MPC algorithm must also be able to \emph{find} them (a non-uniform algorithm can simply have them `hard-coded' in). To do so, we can run a brute-force search on each machine; $O(n^{\sparam})$ space is required to identify each of the $2^{O(n^{\sparam})}$ pairs of graphs and check the sensitivity property. However, we require $2^{O(n^{\sparam})}$ local computation (which is also implicitly the case in \cite{GKU19}). For specific lower bounds there may be efficient constructions, but this must be determined on a case-by-case basis.} two $D$-radius-identical centered graphs $G,G'\in \mathcal{G}$ with at most $ N$ nodes and maximum degree $\Delta$, such that $\cA_{\MPC}$ is $(D, \frac{1}{4N^2},n^{c/\sparam+2},\Delta)$-sensitive with respect to $G,G'$.
		
Our algorithm $B_{\text{st-conn}}$ will work by combining the results of several parallel simulations, run on disjoint groups of machines. We first divide the seed $\mathcal{S}$ equally among these simulations; by choosing our input $\poly(n)$ seed length sufficiently high, we can guarantee that each simulation is given sufficient randomness. In each simulation, each node in $v\in V(H)$ uses part of the random seed assigned to the simulation to independently, uniformly randomly, choose a value $h(v) \in [D]$.

For each simulation we now construct a pair of graphs $G_H$ and $G'_H$ on which to run $\cA_{\MPC}$. In each case we begin with the graph $H$, and discard all nodes of degree higher than $2$. Then, for each remaining node in $H$ we collect its radius-$1$ neighborhood onto a machine (this takes $O(1)$ rounds and does not increase asymptotic space usage). Each node of $H$ then drops out if its radius-$1$ neighborhood is not consistent with being part of an $s-t$ path with consecutively increasing values of $h$ along the path (other than the value of $h(t)$, which we make no requirements for). In particular:
		
\begin{itemize}
\item Nodes $s$ and $t$ in $H$ must have degree $1$ (and otherwise we terminate the simulation immediately and output \NO).
\item Other nodes must have degree $2$, and the values of $h$ in their radius-$1$ neighborhood must be a consecutive triplet (other than $h(t)$).
\end{itemize}
		
Now, every remaining node (i.e., satisfying these conditions) in $H$ is assigned (copies of) nodes in $G$ and $G'$ as follows:
		
\begin{itemize}
\item Node $s$ in $H$ is assigned all nodes of distance at most $h(s)$ from $v$ in $G$ and $G'$.
\item Node $t$ in $H$ is assigned all nodes of distance more than $D$ from $v$ in $G$ and $G'$.
\item Other nodes $u$ in $H$ are assigned the nodes of distance exactly $h(u)$ from $v$ in $G$ and $G'$.
\end{itemize}
		
We can now define the new graphs $G_H$ and $G'_H$ as follows:

The node sets of $G_H$ and $G'_H$ consist of every assigned copy of nodes in $G$ and $G'$ respectively. We also add to $G_H$ and $G'_H$ an extra, full, copy of $G$, disconnected from the rest of the graph, whose sole purpose is to ensure that the maximum degree of $G_H$ and $G'_H$ is exactly $\Delta$, and enough isolated nodes to raise the number of nodes (currently $O(nN) = O(n^{1+\sparam/2})$) to exactly $n^{c/\sparam +2}$ (which is larger for sufficiently large constant $c$). Let $w_{u}$ in $G_H$ denote a copy of $w$ in $G$, assigned to node $u$ in $H$. We place an edge between $w_{u}$ and $\hat w_{\hat u}$ in $G_H$ iff $w\sim_G \hat w$ and $u\sim_H \hat u$ (or $u=\hat u$). As IDs, for a simulated node $w_{u}$, we use the original ID of $w$. We will show that these IDs remain component-unique in $G_H$ and $G'_H$. As names, we use $(name(u),name(w))$; these names are guaranteed to be fully unique. We arbitrarily choose fully unique names and IDs for the extra isolated nodes and copy of $G$.

The graphs $G_H$ and $G'_H$ can easily be constructed in $O(1)$ rounds of \MPC, using $O(n^{\sparam})$ local space, since machines each need only store $G$, $G'$, and $O(n^{\sparam})$ edges of $H$. We construct $G'_H$ from $G'$ analogously.
		
We see that the connected components of $G_H$ and $G'_H$ are isomorphic to induced subgraphs of $G$ and $G'$ respectively. Furthermore, the IDs of nodes in these components are identical to the relevant subgraphs of $G$ and $G'$, and are therefore component-unique (as they were in $G$ and $G'$). In particular, we consider the topology and IDs of the connected components $CC({v_{s}})$ and $CC'(v_{s})$ containing $v_{s}$ in $G_H$ and $G'_H$ respectively:
		
\begin{enumerate}
\item If $s$ and $t$ are endpoints of a path of $p\le D+1$ \emph{nodes} (and therefore of \emph{length} at most $D$) in $H$, $h(s)=p-D$, and $h(u_d) = p-D+d$ for each $u_d\ne t$ which is the node at distance $d$ along the path from $s$, then $CC({v_{s}}) =  G$ and $CC'(v_{s}) = G'$;
\item otherwise, $CC(v_{s})=CC'(v_{s})$.
\end{enumerate}
		
In the first case, this is because the nodes on the $s-t$ path are together assigned precisely one copy of each on the nodes in $G$ and $G'$, in such away that every two adjacent nodes in $G$ and $G'$ are assigned to adjacent nodes on this path. In the second case, the only nodes that differ in $G$ and $G'$ are assigned only to $t$, and no outputs of $h$ can cause these to be connected to $v_{s}$. Therefore the graphs of nodes that \emph{are} connected to $v_{s}$ are identical.
	
We now simulate $\cA_{\MPC}$ on $G_H$ and $G'_H$ using a sufficiently large $\poly(n)$-length part of $\mathcal{S}$ allocated to this simulation as the random seed (the same seed for both $G_H$ and $G'_H$). Since $G_H$ and $G'_H$ are in the \emph{normal} family $\mathcal{G}$ (they are disjoint unions of induced subgraphs of $G, G'\in \mathcal{G}$), these are valid inputs for $\cA_{\MPC}$.
		
If $s$ and $t$ are endpoints of a path of length at most $D$ in $H$, then the probability that all nodes $v$ on this path choose the `correct' value $h(v)$ (i.e., as specified in point $1$ above) is at least $D^{-D} = N^{-o(1)}$. In this case, we have $CC({v_{s}}) =  G\in \mathcal{G}$ and $CC'(v_{s}) = G'\in \mathcal{G}$. Since $\cA_{\MPC}$ is component-stable, its output for $v_{s}$ in each case is equal to $\cA_{\MPC}(G,v,n^{c/\sparam+2},\Delta,\mathcal{S})$  and $\cA_{\MPC}(G',v,n^{c/\sparam+2},\Delta,\mathcal{S})$ respectively, and by sensitivity, with probability at least $\frac{1}{4N^2}$, these are different. If this occurs, the simulation output, and the overall output, is correctly~\YES.

If $s$ and $t$ are not endpoints of a path of length at most $D$ in $H$, then in all simulations we have $CC({v_{s}})=CC'({v_{s}})$. Since $\cA_{\MPC}$ is component-stable, its output depends (deterministically) on these connected components (as well as equal values of $\Delta$, random seed $\mathcal{S}$, and $n_{G_H}=n_{G'_H}=n^{c/\sparam+2}$), and is therefore the same in both cases. The output of all simulations is then correctly~\NO.
				
The final output for $B_{\text{st-conn}}$ is as follows: after running $\poly(n)$-many concurrent simulations, we return \YES if any simulation returned \YES, and \NO otherwise. We have already shown that if $s$ and $t$ are not endpoints of a path of length at most $D$ in $H$, all simulations will return \NO, so $B_{\text{st-conn}}$ will correctly output \NO. It remains to show that $B_{\text{st-conn}}$ is also correct with high probability on \YES instances.

Each simulation has a $D^{-D} = N^{-o(1)}$ probability of the nodes in the s-t path choosing the correct values of $h(v)$, and a $\frac{1}{4N^2}$ probability of returning different outputs on $G_H$ and $G'_H$. Therefore, there is at least a $\frac{1}{4N^3}$ probability of the simulation returning \YES. Since simulations use independent randomness, and we can run sufficiently large $\poly(n)$ many of them in parallel, we can amplify the total probability of success to $1-\frac {1}{n}$ as required.

The round complexity of $B_{\text{st-conn}}$ is dominated by simulating $\cA_{\MPC}$ on graphs with $O(n^{c/\sparam+2})$ nodes and maximum degree $\Delta$, which, using that $T$ is a constrained function, is $O(T(n^{c/\sparam+2},\Delta)) = O(T(n,\Delta))$.
\end{proof}


While Lemma \ref{lem:stconn} already provides a strong conditional lower bound for $O(1)$-replicable graph problems in our setting, the lower bound is conditioned on the complexity of solving $D$-diameter $s$-$t$ connectivity. However, we can apply the framework \cite{GKU19} to extend the analysis to condition on the more standard connectivity conjecture --- to conclude our analysis with the result stated in \Cref{thm:lb}.

\junk{
\begin{theorem}
\label{thm:lb}
Let $\mathcal{P}$ be a $O(1)$-replicable graph problem that has a $T(N,\Delta)$-round lower bound in the randomized \LOCAL model with shared randomness, for constrained function $T$, on graphs with input estimate $N$ and maximum degree $\Delta$, from some normal family $\mathcal{G}$. Suppose that there is a randomized $o(\log T(n,\Delta))$-round low-space component-stable \MPC algorithm $\cA_{\MPC}$ for solving $\mathcal{P}$ on legal $n$-node graphs with maximum degree $\Delta$ from $\mathcal{G}$, succeeding with probability at least $1-\frac 1n$. Then, there exists a low-space randomized \MPC algorithm $\cA_{cycle}$ that can distinguish one $n$-node cycle from two $\frac{n}{2}$-node cycles in $o(\log n)$ rounds, succeeding with probability at least $1-\frac 1n$.
	
The same holds if the \LOCAL lower bound and algorithm $\cA_{\MPC}$ are both deterministic (but the obtained algorithm $\cA_{cycle}$ remains randomized).
\end{theorem}
}

\begin{proof}[\textbf{Proof of \Cref{thm:lb}}]
This reduction from Lemma \ref{lem:stconn} is taken directly from Section~V in \cite{GKU19}. First, we notice that \cite[Lemma~IV.1]{GKU19} applies to any \MPC algorithm, i.e., is not affected by our changes to the definition of component-stability. Therefore, assuming that any low-space \MPC algorithms distinguishing one $n$-node cycle from two $\frac{n}{2}$-node cycles requires $\Omega(\log n)$ rounds, we obtain that any low-space \MPC algorithm that solves $D$-diameter $s$-$t$ connectivity for $D \le \log^\gamma n$ for a constant $\gamma \in (0, 1)$ requires $\Omega(\log D)$ rounds.
By setting $D = T(n,\Delta)$, we can combine this fact with Lemma \ref{lem:stconn} to conclude the proof of the theorem.
\end{proof}

\junk{
\subsection{Conditional \MPC lower bounds via hardness of 1-vs-2-cycle conjecture}
\label{subsec:thm:lb}

While Lemma \ref{lem:stconn} already provides a strong conditional lower bound for $O(1)$-replicable graph problems in our setting, the lower bound is conditioned on the complexity of solving $D$-diameter $s$-$t$ connectivity, which has not been studied so extensively in the past. In what follows, we apply the approach from \cite{GKU19} to extend the analysis to condition on the more standard \emph{1-vs-2-cycle conjecture}. We begin with the following lemma (proven by Ghaffari et al.\ \cite{GKU19} in the setting of a different notion of component-stable algorithms, as discussed in \Cref{subsec:problems-with-def-comp-stab}).
%

\begin{lemma}
\label{lem:1v2cycle}
Suppose that every \MPC algorithm with local memory $\spac = n^\sparam$ for a constant $\sparam \in (0,1)$ and global memory $\poly(n)$ that can distinguish one $n$-node cycle from two $\frac{n}{2}$-node cycles requires $\Omega(\log n)$ rounds. Then, any \MPC algorithm with local memory $\spac = n^\sparam$ and global memory $\poly(n)$ that solves $D$-diameter $s$-$t$ connectivity for $D\le \log^\gamma n$ for a constant $\gamma \in (0, 1)$ requires $\Omega(\log D)$ rounds.
\end{lemma}

\begin{proof}
We use Lemma~IV.1 from \cite{GKU19}. Since this lemma does not require that the algorithm
is component-stable, our definitional changes have no effect and the proof from \cite{GKU19} applies.
\end{proof}

Now we can combine Lemma \ref{lem:stconn} with Lemma \ref{lem:1v2cycle} to obtain the following theorem.

\begin{theorem}
\label{thm:lb}
Let $\mathcal{P}$ be a $O(1)$-replicable graph problem that has a $T(N,\Delta)$-round lower bound in the randomized \LOCAL model with shared randomness, for constrained function $T$, on graphs with input estimate $N$ and maximum degree $\Delta$, from some normal family $\mathcal{G}$. Suppose that there is an $o(\log T(N,\Delta))$-round\Artur{It was $o(\log D)$-rounds before, and so it was a typo I guess.} low-space component-stable \MPC algorithm $\cA_{\MPC}$ for solving $\mathcal{P}$ on $\mathcal{G}$, succeeding with high probability\Artur{More specific? (define what is the failure bound)}. Then, there exists a low-space randomized \MPC algorithm that can distinguish one $n$-node cycle from two $\frac{n}{2}$-node cycles in $o(\log n)$ rounds, succeeding with high probability.
	
The same holds if the \LOCAL lower bound and algorithm $\cA_{\MPC}$ are both deterministic.
\end{theorem}

\begin{proof}
Setting $D = T(n,\Delta)$, by Lemma \ref{lem:stconn} we can obtain an $o(\log D)$-round low-space \MPC algorithm for solving $D$-diameter $s$-$t$ connectivity, succeeding with high probability. The theorem then follows from Lemma \ref{lem:1v2cycle}.
\end{proof}
}


\subsection{Applications of Theorem \ref{thm:lb}: conditional \MPC lower bounds}
\label{subsec:lb:applications}

Similarly as was been demonstrated by Ghaffari et al.\ \cite{GKU19}, the lifting arguments in \Cref{thm:lb}, when combined with the known lower bounds for \LOCAL algorithms, lead to a host of conditional lower bounds for low-space component-stable \MPC algorithms for many natural graph problems. In particular, we can recover the main applications of Ghaffari et al.\ \cite{GKU19} (see Theorems~I.1, I.2, and I.3 there):

\begin{theorem}
\label{thm:recovering-GKU}
Assuming the connectivity conjecture (that no low-space randomized \MPC algorithm can distinguish one $n$-node cycle from two $\frac{n}{2}$-node cycles in $o(\log n)$ rounds, with failure probability at most $\frac1n$), the following lower bounds hold for randomized low-space component-stable \MPC algorithms:
\begin{itemize}
\item $\Omega(\log\log n)$ rounds for a constant approximation of maximum matching (even on forests\hide{\Artur{Note to Peter: at some moment you mentioned that something is unclear in that claim for forests; could you please double-check it.} \Peter{That was for the deterministic lower bound for maximal matching/MIS - I changed it there}}), a constant approximation of vertex cover, a maximal independent set;
\item $\Omega(\log\log\log n)$ rounds for $(\Delta+1)$-coloring, unless the deterministic complexity of $(\Delta+1)$-coloring in the \LOCAL model is $2^{\log^{o(1)}\log n}$;
\item $\Omega(\log\log\log n)$ rounds for the Lov\'asz Local Lemma problem, or even for the special case of sinkless orientation where we should orient the edges of a $d$-regular graph, where $d \ge 4$, such that each node has at least one outgoing edge.
\end{itemize}
\end{theorem}

\begin{proof}
First, we notice that Ghaffari et al.\ \cite{GKU19} proved the following lower bounds for the randomized \LOCAL model with shared randomness:
\begin{itemize}
\item $\Omega(\sqrt{\log n/\log\log n})$ rounds for a $\polylog(n)$-approximate solution for the minimum vertex cover, maximum matching, or minimum dominating set problem, and for finding a maximal matching (even on trees) or a maximal independent set (Theorem~V.1 and Corollary~V.2 in \cite{GKU19});
\item $\Omega(\sqrt{\log\log n})$ rounds for $(\Delta+1)$-coloring unless the deterministic complexity of $(\Delta+1)$-coloring in the \LOCAL model is $o(\log n)$ (indirectly in Corollary~V.4 in \cite{GKU19});
\item $\Omega(\log\log n)$ rounds for the Lov\'asz Local Lemma problem, or even for the special instance known as sinkless orientation where we should orient the edges of a $d$-regular graph, where $d \ge 4$, such that each node has at least one outgoing edge \cite[Corollary~V.5]{GKU19}.
\end{itemize}

The proof of \Cref{thm:recovering-GKU} now follows immediately by combining these \LOCAL lower bound to \Cref{thm:lb}, by taking each time $\mathcal{G}$ to be the set of all graphs, and by noticing that all the problems considered are $O(1)$-replicable (see \Cref{subsec:which-are-replicable-problem}, and for example, Lemma \ref{lemma:LCL-replicable} for \LCL{}s and Lemma \ref{lem:approx-matching-replicable} for approximate maximum matching). The only issue worth mentioning is that since the class of trees is not normal (see Definition \ref{def:normal-graphs}), the lifting of the \LOCAL lower bound for finding a maximal matching on trees extends only to forests, that is, that any randomized low-space component-stable \MPC algorithm that returns a constant approximation of maximum matching on forests has $\Omega(\log\log n)$ rounds complexity.
\end{proof}

Comparing the results claimed by Ghaffari et al.\ \cite{GKU19} (Theorems~I.1, I.2, and I.3) with those proven in \Cref{thm:recovering-GKU}, there are only two differences: firstly, our lower bound for matching holds only for forests, while \cite{GKU19} claimed it for trees (for discussion of the reason for this difference, see Section \ref{subsec:graph-families}), and secondly, our result for $(\Delta+1)$-coloring is updated to hold on a weaker conjecture of no $2^{\log^{o(1)}\log n}$-round deterministic \LOCAL algorithm for $(\Delta+1)$-coloring (since the original result in \cite{GKU19} conditioned on the $2^{o(\sqrt{\log n})}$-deterministic hardness, which was recently disproved in \cite{RG20}).


\subsubsection{Deterministic conditional \MPC lower bounds}
\label{subsubsec:lb:deterministic-applications}

Finally, observe that \Cref{thm:lb} can be directly combined with existing deterministic \LOCAL lower bounds to obtain deterministic lower bounds for low-space deterministic \MPC algorithms. (Note that unlike in the case of randomized \MPC algorithms above, one can directly apply deterministic \LOCAL lower bounds as black-box in \Cref{thm:lb}, since the issue of shared randomness is of no importance for deterministic algorithms.) For example, we apply it in \Cref{thm:MIS-MM-1} to give $\Omega(\log\Delta +\log\log n)$-round lower bounds for deterministic low-space component-stable \MPC algorithms for maximal matching or maximal independent set (conditioned on the connectivity conjecture). Similarly, we present deterministic conditional lower bounds of $\Omega(\log\log_\Delta n)$ rounds for low-space component-stable \MPC algorithms for sinkless orientation, $(2\Delta-2)$-edge coloring, and $\Delta$-vertex coloring, holding even in forests (Theorems \ref{thm:sinkless-cs}, \ref{thm:ecoloring-1}, and \ref{thm:vcoloring-1}). For some more lower bounds to which the framework is applicable, see, e.g., the recent deterministic \LOCAL lower bounds for ruling sets in \cite{BBO20}. 

%% file: stable-non-deter.tex


We start by presenting a general statement for characterization of the graph family of local problems for which component-unstable \MPC algorithms provably help. Specifically, this includes problems for which there is a provable exponential gap between their \LOCAL deterministic and randomized complexities (e.g., as shown in \cite{BBHORS19,CQW+20,CKP19}). Notice that the obtained \MPC algorithms are \emph{non-unform}.

\begin{theorem}
\label{thm:instability-provably-helps}
Let $\mathcal{P}$ be a $O(1)$-replicable graph problem. Let $T_r(N)$ and $T_d(N)$ be the randomized and deterministic, respectively, \LOCAL round complexity of $\mathcal{P}$ on bounded-degree graphs with $n$ nodes, with an input size estimate $N$ and exact knowledge of $\Delta$. If $T_r(N) < \log^{\gamma} N$ for some constant $\gamma \in (0,1)$, and if there is a provable exponential gap between $T_r(N)$ and $T_d(N)$, then, assuming the connectivity conjecture, there is an exponential gap between component-stable and component-unstable deterministic low-space \MPC complexities for solving $\mathcal{P}$.
\end{theorem}

\begin{proof}
Since $T_r(N) < \log^{\gamma} N$ for some constant $\gamma \in (0,1)$ (see Definition \ref{def:constrained-functions}), by \Cref{thm:lb}, conditioning on the connectivity conjecture, any deterministic low-space \emph{component-stable} \MPC algorithm for $\mathcal{P}$ requires $\Omega(\log(T_d(n)))$ rounds. We will show that there exists a deterministic low-space \emph{component-unstable} \MPC algorithm for $\mathcal{P}$ running in $O(\log(T_r(n)))$ rounds using an $n^2$ factor more machines (the algorithm is possibly non-uniform and non-explicit).

We use Lemma \ref{lem:det-large-prob}, which says that it is sufficient to show a randomized low-space \MPC algorithm for $\mathcal{P}$ that runs in $O(\log(T_r(n)))$ rounds and succeeds with probability at least $1-2^{-n^2}$. The algorithm employs $n^2$ parallel repetitions of procedure $\mathcal{B}$, each uses a disjoint set of machines and a distinct random coins determined by a shared random string of length $\poly(n)$. The final output is determined by picking the \emph{most successful} repetition (i.e., with the maximal number of legal node outputs). Procedure $\mathcal{B}$ simply collects the $2T_r(n)$ ball of each node $u \in V(G)$, which can be done in $O(\log(T_r(n)))$ rounds. The machines then simulate the randomized local algorithm on balls using the shared seed. This completes the description of the simulation.

We first show that this algorithm succeeds with probability $1-2^{-n^2}$. By the properties of the \LOCAL algorithm, a single repetition succeeds with probability at least $1-\frac1n$ (over the randomness of the shared seed). Thus, by employing $n^2$ independent repetitions, we get the probability that all these repetitions fail is less than $2^{-n^2}$.

Notice that since $T_r(n) = o(\log n)$ and the input graph is bounded-degree, the $2T_r(n)$-ball of each node indeed fits the local space of each machine, and thus the \MPC algorithm can be implemented with $n^{o(1)}$ local space, and hence is low-space. Finally, let us observe that the obtained algorithm is \emph{component-unstable}, since it relies on globally agreeing on the outcome of all repetitions.
\end{proof}

Notice that \Cref{thm:instability-provably-helps} implies \Cref{thm:complexity-SDet-Det}, that is, demonstrates that, assuming the connectivity conjecture, there are some graph problems for which there are deterministic low-space component-unstable \MPC algorithms that are provably faster than their component-stable counterparts. However, the weakness of \Cref{thm:instability-provably-helps} is that the obtained deterministic low-space component-unstable \MPC algorithms are \emph{non-unform}. In the rest of this section we will address this issue and show that a similar claim holds also for uniform deterministic \MPC algorithms, and those which use almost-optimal global space.

\subsection{Derandomization tools}
\label{subsec:derand-tools}

In this section we present some basic derandomization tools used in our analysis.

\subsubsection{Strongly $(\epsilon,k)$-wise independent hash functions}

We will use the notion of \emph{strongly $(\epsilon,k)$-wise independent hash functions} which are defined as follows:

\begin{definition}
\label{def:strongly-independent-rvs}
Let $A$ and $B$ be two sets and let $\epsilon \ge 0$. A family $\mathcal{H}$ of functions $A \rightarrow B$ is \textbf{strongly $(\epsilon,k)$-wise independent} if, when $h \in \mathcal{H}$ is chosen uniformly at random, for any $t \le k$ and for any $t$ distinct $x_1, \dots, x_t \in A$ and any (not necessarily distinct) $y_1, \dots, y_t \in B$, we have
\begin{align*}
	\left|\Prob{h(x_i) = y_i, 1 \le i \le t} - \frac{1}{|B|^t}\right|
		&\le
	\epsilon
	\enspace.
\end{align*}
\end{definition}

That is, the probability of any particular $t\le k$ outputs of a random function from $\mathcal H$ differs from what one would expect under uniform full randomness only by $\epsilon$. In our applications, we will choose $\epsilon = n^{-c}$ for sufficiently large constant $c$, and we can then assume that these outputs are indeed fully independent (since the statistical difference is far below the failure probability of our algorithms).

Kurosawa et al. \ \cite{KJS01} presented the following result showing the existence of strongly $(\epsilon,k)$-wise independent families of size (and therefore requiring few bits to specify):

\begin{theorem}[\cite{KJS01}]
\label{thm:hash}
For any sets $A, B$, positive integer $k$, and positive real $\epsilon$, there exists a strongly $(\epsilon,k)$-wise independent family $\mathcal{H}$ of size
\begin{align*}
	2^{O(\log\log |A| + k \log |B| + \log \frac{1}{\eps})}
	\enspace.
\end{align*}
Furthermore, each $h\in \mathcal{H}$ can be specified and evaluated on any input in $\polylog(|A|,|B|,2^k,\frac 1\eps)$ time and space.
\end{theorem}

\begin{proof}
As noted in the proof of Theorem 2.9 of \cite{KJS01}, by applying propositions 1.2 and 2.3 of \cite{KJS01} an $(\epsilon,k)$-wise independent family of the desired size can be constructed from the independent sample spaces of Alon et al. \cite{AGHP92}.
\end{proof}

\junk{TURNS OUT I ACTUALLY DID NEED ALMOST-INDEPENDENCE
	\begin{definition}
		\label{def:k-wise-independent}
		For $N, L, k \in \nat$ such that $k \le N$, a family of functions $\mathcal{H} = \{h : [N] \rightarrow [L]\}$ is \textbf{$k$-wise independent} if for all distinct $x_1, \dots, x_k \in [N]$, the random variables $h(x_1), \dots, h(x_k)$ are independent and uniformly distributed in $[L]$ when $h$ is chosen uniformly at random from~$\mathcal{H}$.
	\end{definition}
	\noindent We will use the following well-known lemma (cf. \cite[Corollary~3.34]{Vadhan12}).
	
	\begin{lemma}
		\label{lem:hash}
		For every $a$, $b$, $k$, there is a family of $k$-wise independent hash functions $\mathcal{H} = \{h : \{0,1\}^a \rightarrow \{0,1\}^b\}$ such that choosing a random function from $\mathcal{H}$ takes $k \cdot \max\{a,b\}$ random bits, and evaluating a function from $\mathcal{H}$ takes time $\poly(a,b,k)$.
\end{lemma}}

\subsubsection{Pseudorandom generators}
\input{PRG.tex}
\subsection{Problems related to the Lov\'asz Local Lemma}
\label{subsec:LLL-related}

We first demonstrate a deterministic complexity separation between component-stable and component-unstable algorithms for a group of problems related to the distributed Lov\'asz Local Lemma: sinkless orientation, $(\Delta+o(\Delta))$-edge coloring, and $o(\Delta)$-coloring triangle-free graphs.

These problems are known to be hard in the \LOCAL model via proofs based on the \emph{round elimination} technique \cite{CQW+20}. We show that by derandomizing an algorithm for the constructive Lov\'asz Local Lemma and plugging this result into known algorithms for the problems, we can surpass the component-stable lower bounds we obtain when lifting the \LOCAL lower bounds to low-space \MPC.

\subsubsection{Algorithmic Lov\'asz Local Lemma}
We first present a deterministic low-space component-unstable \MPC algorithm for the Lov\'asz Local Lemma which uses \emph{heavy local computations} to obtain a small number of \MPC rounds (though for bounded-degree graphs, computation is still polynomial in $n$). Later we will demonstrate that this algorithm can be applied to obtained similar \MPC algorithms for sinkless orientation, edge-coloring, and vertex-coloring.

The \emph{algorithmic Lov\'asz Local Lemma (LLL)} is defined as follows:

\begin{definition}
Consider a set \Vars\ of independent random variables, and a family \ents\ of $n$ (bad) events on these variables. Each event $A \in\ents $ depends on some subset $\Vars(A) \subseteq \Vars$ of variables.
Define the dependency graph $G_\ents = (\ents , \{(A, B) \mid \Vars(A) \cap \Vars(B) \ne \emptyset\})$ that connects any two events which share at least one variable. Let $d$ be the maximum degree in this graph, i.e., each event $A \in \ents$
shares variables with at most $d$ other events $B \in \ents$. Finally, define $p = \max_{A\in\ents}\Prob{A}$.

The Lov\'asz Local Lemma shows that $\Prob{\cap_{A\in \ents} \bar A} > 0$ (i.e., there is some assignment of variables that does not satisfy any of the bad events) under the LLL criterion that $epd \le 1$. Our algorithmic goal is to \emph{find} such an assignment of the variables (possibly under a stronger criterion).

\end{definition}

We give the following deterministic low-space \MPC algorithm. 

\begin{lemma}
\label{lem:LLL}
Any LLL instance with $d=\log^{o(1)} n$ and $p\le \frac 1C d^{-C}$ (for sufficiently high constant $C$), in which each bad event has $\poly(d)$ dependent variables, and these variables are independent fair random bits, can be solved in $\poly(d)+O(\log\log\log n)$ rounds in deterministic low-space \MPC, using $n^{1+o(1)}$ global space and $n^{\poly(d)}$ local computation.
\end{lemma}

\begin{proof}
We derandomize the \LOCAL algorithm of Fischer and Ghaffari \cite{FG17}. The algorithm consists of an $\tilde O(d)+\log^* n$ initial deterministic coloring procedure from \cite{FHK16}, an $O(d^2)$-round randomized \emph{pre-shattering} part, and then a $\polyloglog(n)$-round deterministic \emph{post-shattering} part (which follows from plugging the network decomposition result of \cite{RG20} into the derandomization framework of \cite{GHK18} applied to the LLL algorithm of \cite{CPS17a}; see \cite{RG20} for details). We first collect the $O(d^2+\polyloglog(n))$-radius ball around each event onto a single machine; these balls are of size $d^{O(d^2+\polyloglog(n))} = n^{o(1)}$, and each event's dependence on its variables can be specified in $2^{\poly(d)}=n^{o(1)}$ bits, so the balls fit on single machines. This collection process can be performed in $O(\log d +\log\log\log n)$ rounds by graph exponentiation.

We can now directly and immediately simulate the initial coloring, and we will also be able to directly simulate the deterministic post-shattering part, so it remains to derandomize the pre-shattering part. We do so using hash functions. First, note that the proof of the Shattering Lemma \cite[Lemma~6]{FG17}, giving the crucial shattering property of the randomized part, requires independence only between the dependent variables of groups of $O(\poly(d)\log n)$ events at a time (the $O(1)$-radius neighborhoods of groups of $O(\log n)$ events). So, it requires variables to be sampled from their distributions with only $(\poly(d) \cdot \log n)$-wise independence.
	
By \Cref{thm:hash}, for any constant $c\in \nat$, there exists a strongly $(n^{-c},\poly(d)\log n)$-wise independent family $\mathcal{H}$ of functions $[\poly(n)]\rightarrow \{0,1\}$ of size $2^{\poly(d)\log n}$, i.e.,  requiring $\poly(d)\log n$ bits to specify each function. We use such a function to generate the values taken by all variables. By taking $c$ to be sufficiently high, we can ensure that the statistical error of the hash functions in negligible. Since the algorithm of \cite{FG17} succeeds with high probability in $n$ when variables are sampled independently at random, we have that using a uniformly random function from $\mathcal{H}$ to provide variable values also causes the algorithm to succeed w.h.p.
	
We can therefore perform a distributed implementation of the method of conditional expectations to deterministically fix a function which causes the algorithm to succeed. \cite{CDP20a,CDP20b} show how to implement this method, in low-space \MPC, in such a way that $\Theta(\log n)$ bits specifying the function can be fixed in a single round, provided success at any node, under any function from $\mathcal H$, can be checked locally on a machine. Here, we can check success locally by simply running the algorithm of \cite{FG17} to completion using the given function, since we have already collected sufficiently large local neighborhoods onto single machines.
	
We therefore require $\poly(d)$ \MPC rounds to perform the method of conditional expectations and derandomize the LLL algorithm. So, the total round complexity is $\poly(d)+O(\log\log\log n)$. The global space usage is dominated by the storage of an $O(d^2+\polyloglog(n))$-radius ball around each node, i.e., $n d^{O(d^2+\polyloglog(n))}= n^{1+o(1)}$ global space. The local computation is dominated by the method of conditional expectations, which requires evaluating all functions from $\mathcal{H}$, taking $n^{\poly(d)}$ computation.
\end{proof}

\subsubsection{Sinkless orientation}
\label{subsec:sinkless-orientation}

In this section, we combine the lifting of the \LOCAL lower bounds to deterministic low-space component-stable \MPC algorithms (\Cref{thm:lb}) with derandomization of the constructive Lov\'asz Local Lemma (Lemma \ref{lem:LLL}) to show that conditioned on the connectivity conjecture, deterministic low-space component-unstable \MPC algorithms for sinkless orientation are provably faster than their component-stable counterparts.

We define sinkless orientation to be the problem of orienting the edges of a graph of minimum degree at least $3$, such that each node has at least one outgoing edge (this minimum degree criterion is necessary, since the problem is not possible, e.g., on a path).

A lower bound for sinkless orientation in the \LOCAL model was first proven by \cite{BFH+16}, and extended to a stronger bound for deterministic algorithms by \cite{CKP19}. When combined with \Cref{thm:lb}, we obtain the following theorem.

\begin{theorem}
\label{thm:sinkless-cs}
Assuming the connectivity conjecture,
there is no deterministic component-stable low-space \MPC algorithm that computes a sinkless orientation in $o(\log\log_\Delta n)$ rounds, even in forests.
\end{theorem}

\begin{proof}
By \cite{BFH+16} and \cite{CKP19}, there is an $\Omega(\log_\Delta N)$ lower bound for deterministic sinkless orientiation even in $\Delta$-regular forests, and this holds even if nodes have unique $O(\log N)$-bit IDs and know the exact values of $n$ and $\Delta$. Since sinkless orientation is an edge-labeling problem, we must work on the line graph of the input graph in order to meet our vertex-labeling framework. We therefore let $\mathcal G$ be the (normal) class of all line graphs of forests, and, setting $T(N,\Delta):= \log_{\Delta}^{1/3}N$ (which is a constrained function) we have an $T(N,\Delta)$-round \LOCAL lower bound for the vertex-labeling version on line graphs. By \Cref{thm:lb}, we obtain an $\Omega(\log(T(N,\Delta))) = \Omega(\log\log_{\Delta} n)$-round lower bound for deterministic low-space \MPC algorithms, conditioned on the connectivity conjecture. Converting back from the line graph formulation to the original input graphs, this lower bound is on the family of forests.
\end{proof}

\Cref{thm:sinkless-cs} is complemented by the following result providing a deterministic low-space component-unstable \MPC algorithm for sinkless orientation that surpasses the component-stable lower bound for low-space \MPC.

\begin{theorem}
\label{thm:sinkless-ub}
There is a deterministic low-space \MPC algorithm that computes a sinkless orientation, in any graph of $\Delta =\log^{o(1)}\log n$ maximum degree, in $\poly(\Delta)+O(\log\log \log n) = o(\log\log_\Delta n)$ rounds, using $n^{1+o(1)}$ global space.
The algorithm is component-unstable, and uses  $n^{\poly(\Delta)}$ local computations.
\end{theorem}

\begin{proof}
We concentrate on the algorithm for the case of $d$-regular graphs with $d>500$; Ghaffari and Su \cite{GS17} show how extend to irregular graphs with minimum degree at least $3$ using $O(\log^*n)$ rounds of deterministic pre-processing and $O(1)$ rounds of deterministic post-processing in \LOCAL, which we can simulate directly in low-space \MPC.

In this case, sinkless orientation can be solved by a single application of the distributed LLL (see e.g. \cite{BFH+16}), with $d=\Delta$ and $p=2^{-\Delta}$ (setting $d>500$ is more than sufficient to make this instance satisfy the LLL criterion of Fischer and Ghaffari's algorithm), and the underlying variables are simple uniformly random choices of orientation for each edge (so can be represented with a single fair random bit). Applying Lemma \ref{lem:LLL}, we obtain a $\poly(\Delta) + O(\log\log\log n)$-round deterministic low-space \MPC algorithm using $n^{\poly(\Delta)}$ local computation.
\end{proof}

We also remark that, for $\Delta\le n^{o(1/\log\log n)}$, an $O(\log\log\log n)$-round component-stable \emph{randomized} algorithm exists for the problem, by simply collecting $\Theta(\log\log n)$-radius balls onto machines via graph exponentiation, and then simulating the randomized \LOCAL algorithm of \cite{GS17}. Sinkless orientiation is therefore an example of a problem with a (conditional) separation between randomized and deterministic algorithms, proving \Cref{thm:complexity-SDet-SRan}

\subsubsection{Edge-coloring}
\label{subsec:coloring}

Similarly to the sinkless orientation problem, the framework combining the lifting of the \LOCAL lower bounds to deterministic low-space component-stable \MPC algorithms (\Cref{thm:lb}) with derandomization of the constructive Lov\'asz Local Lemma (Lemma \ref{lem:LLL}) can be used to show that assuming the connectivity conjecture, for the classical edge-coloring there are deterministic low-space component-unstable \MPC algorithms that are provably faster than their component-stable counterparts. We begin with the following application of \Cref{thm:lb}.

\begin{theorem}
\label{thm:ecoloring-1}
Assuming the connectivity conjecture,
there is no deterministic component-stable low-space \MPC algorithm
that computes a $(2\Delta-2)$-edge coloring, even in forests, in $o(\log \log_\Delta n)$ rounds.
\end{theorem}

\begin{proof}
Chang et al. \cite{CQW+20} give a deterministic \LOCAL lower bound of $\Omega(\log_\Delta N)$ for $(2\Delta- 2)$-edge coloring, even in forests. To fit $(2\Delta- 2)$-edge coloring into our framework (in which problem outputs are labels on nodes) we must define it as vertex-coloring on the line graph. In \LOCAL, operations on the line graph can be simulated in the original graph and vice versa with only one round additive overhead. So, we have an $\Omega(\log_\Delta N)$-round deterministic lower bound for `$(2\Delta- 2)$-coloring the line graphs of forests with maximum degree $\Delta$'.
	
The family of line graphs of forests is both hereditary (deleting a node in the line graph corresponds to deleting an edge in the original graph, and forests are closed under edge deletion) and closed under disjoint union (which corresponds to disjoint union in the original graph), so it is normal under Definition \ref{def:normal-graphs}. We set $T(N,\Delta):=\log^{1/3}_\Delta N$, a constrained function. Then, by Theorem \ref{thm:lb}, we obtain a conditional $\Omega(\log\log_\Delta n)$ deterministic lower bound for $(2\Delta- 2)$-edge coloring forests with $n$ nodes and maximum degree $\Delta$.
\end{proof}

\begin{theorem}
\label{thm:ecoloring-2}
There is a deterministic low-space \MPC algorithm that computes a $(\Delta+\sqrt\Delta\log^3 \Delta)$-edge coloring, in any graph of maximum degree $\Delta =\log^{o(1)}\log n$, in $O(\poly(\Delta)+\log\Delta\log\log\log n)= o(\log\log_{\Delta} n)$ rounds, using $n^{1+o(1)}$ global space. The algorithm is component-unstable and uses $n^{\poly(\Delta)}$ local computation.
\end{theorem}

\begin{proof}
We apply the algorithm of \cite{CQW+20}. Setting $\eps =\frac{\log^3 \Delta}{\sqrt\Delta}$	in Theorem 4 of \cite{CQW+20} gives a running time dominated by $O(\log\Delta)$ applications of LLL, with $d = \poly(\Delta)$ and $p=\Delta^{-\omega(1)}$. Furthermore, the variables in the LLL instances are uniformly random choices of colors from each edge's palette (of some current palette size $P$). These can be generated from $\poly(d)$ fair random bits in such a way that the probability of choosing any particular color differs from $1/P$ by at most $2^{-\poly(d)}$; the probability of each bad event therefore also incurs error of $2^{-\poly(d)}$, and so remains $\Delta^{-\omega(1)}$.
	
Applying Lemma \ref{lem:LLL}, we can perform all applications of LLL in $O(\poly(\Delta)+\log\Delta\log\log\log n)$ rounds in deterministic low-space \MPC with $n^{1+o(1)}$ global space. The algorithm of \cite{CQW+20} then applies a final $5\Delta'$-edge coloring algorithm to finish off the remaining graph (where $\Delta'$ is the maximum degree of this graph); we can do this in $O(\log^* n)$ rounds by simulating Linial's (deterministic) vertex-coloring algorithm \cite{Linial92} on the line graph. The total round complexity is $O(\poly(\Delta)+\log\Delta\log\log\log n)$, the global space is $n^{1+o(1)}$, and the local computation is $n^{\poly(\Delta)}$.
\end{proof}

\subsubsection{Vertex-coloring triangle-free graphs}

A combination of \Cref{thm:lb} and Lemma \ref{lem:LLL} can be used to show a similar complexity gap between deterministic component-stable and component-unstable \MPC algorithms for vertex coloring.

\begin{theorem}
\label{thm:vcoloring-1}
Assuming the connectivity conjecture,
there is no deterministic component-stable low-space \MPC algorithm that computes a $\Delta$-vertex coloring, even in forests, in $o(\log \log_\Delta n)$ rounds.
\end{theorem}

\begin{proof}
Chang, Kopelowitz, and Pettie \cite{CKP19} give an $\Omega(\log_{\Delta} N)$-round deterministic lower bound for $\Delta$-coloring trees (and therefore also the normal class of \emph{forests}) of maximum degree $\Delta$ (for $\Delta\ge 3$). Setting $T(N,\Delta):=\log^{1/3}_\Delta N$ (a constrained function), by Theorem \ref{thm:lb}, we obtain a conditional $\Omega(\log\log_\Delta n)$ deterministic lower bound for $\Delta$-coloring forests with $n$ nodes and maximum degree $\Delta$.
\end{proof}

\begin{theorem}
\label{thm:vcoloring-2}
There is a deterministic low-space \MPC algorithm that computes a $O(\frac{\Delta}{\log\Delta})$-vertex coloring, in any triangle-free graph of $\Delta =\log^{o(1)}\log n$ maximum degree, in $\poly(\Delta) + O(\log\Delta \log\log \log n) = o(\log\log_{\Delta} n)$ rounds, using $n^{1+o(1)}$ global space. The algorithm is component-unstable and uses $n^{\poly(\Delta)}$ local computations.
\end{theorem}

\begin{proof}
We plug in our LLL algorithm into the algorithm of \cite{PS15}. When $\Delta = o(\log n)$, this algorithm consists of $O(k+\log^* n)$ applications of LLL in order to $O(\frac{\Delta}{k})$-color the graph, where $k \le (\frac 14 - o(1))\ln \Delta$. We choose some such $k=\Theta (\log\Delta)$.
	
Similarly to the proof of Theorem \ref{thm:ecoloring-2}, the LLL instances required have $d = \poly(\Delta)$ and $p=\Delta^{-\omega(1)}$, and the variables are uniformly random color choices, which we can generate from $\poly(d)$ fair coins while only increasing the probability of any bad event by $2^{-\poly(d)}$. So, we can apple Lemma \ref{lem:LLL}, and obtain a deterministic low-space \MPC algorithm for $O(\frac{\Delta}{\log\Delta})$-coloring triangle-free graphs in $O(\poly(\Delta) +\log\Delta \log\log \log n)$ rounds, using $n^{1+o(1)}$ global space, with $n^{\poly(\Delta)}$ local computation.
\end{proof}

For all the problems above, sinkless orientation, edge-coloring, and vertex-coloring, we have obtained component-unstable algorithms which surpass the conditional lower bounds for component-stable algorithms when $\Delta = \log^{o(1)}\log n$. Furthermore, though in general these algorithms use heavy local computation, for bounded degree ($\Delta= O(1)$) graphs their local computation is $\poly(n)$. Since we still surpass the lower bounds for sufficiently large constant $\Delta$, this demonstrates that \emph{component-instability helps for deterministic algorithms even using polynomial computation}.

\subsection{Extendable algorithms}
\label{subsec:examps-complexity-separation}

We next describe an explicit derandomization recipe for a particular class of local algorithms. This recipe allows one to derandomize $r$-round local algorithms within $O(\log r)$ low-space \MPC rounds, provided that the $r$-radius ball of each node in the graph $G$ fits the local space of the machines. As a consequence, we show component-unstable algorithms for maximal independent set and maximal matching that surpass the lower bounds for component-stable algorithms.

We call the class of \LOCAL algorithms we consider \emph{extendable}. Roughly speaking, these algorithms can extend any partial legal solution (e.g., a collection of legally colored nodes) into a complete solution (similar to the notion of \emph{greedy} algorithms). In addition, the local computation performed at each node in these algorithms must be polynomial. Even though the \LOCAL model does not account for the local computation time, most of the \LOCAL algorithms are in fact efficient in this regard. We next define this notion more formally.

\begin{definition}[\textbf{Extendable algorithms}]
\label{def:local-extend}
Let $\cA$ be a randomized \LOCAL algorithm for an \LCL problem $\mathcal{P}$ with round complexity at most $T(n,\Delta)$ (where $T$ is non-decreasing in $n$ and $\Delta$) on every $n$-node graph, when provided with the exact value of $n$ and maximum degree $\Delta$. Then, $\cA$ is \textbf{extendable} if:
\begin{enumerate}[(i)]
\item Any partially decided subgraph can be extended into a global solution for $\mathcal{P}$. Formally, $\cA$ returns an output labeling on $G$ giving each node a label in $L\cup \{\bot\}$, where $L$ is the output alphabet permitted by the problem $\mathcal{P}$. Then, re-labeling the nodes labeled $\bot$ with any valid output labeling on their induced graph must give a valid output labeling on $G$. This must hold with certainty (even though $\cA$ is a randomized algorithm).
\item In expectation, $\cA$ labels less than $\frac 12$ nodes $\bot$.	
\item In every round $i$ of $\cA$, each node $u$ performs at most $\Delta^{O(T(n,\Delta))}$ internal computation to determine its output for that round (e.g., the messages to be sent in the next round, internal states, and its final output in the last round).
\end{enumerate}
\end{definition}

We next show that using the PRG construction of Lemma \ref{lem:prg-alg} any extendable \LOCAL algorithm running in $t=T(n,\Delta)$ rounds can be simulated deterministically within $O(\log t)$ \MPC-rounds. This has various applications for derandomizing \LOCAL algorithms on low-degree graphs, which consequently yields a separation between stable and unstable deterministic \MPC algorithms. For simplicity, we consider \LOCAL algorithms for \LCL problems, however, this can be extended for the approximation variants of \LCL problems (e.g., approximate max-IS and maximal matching).

\begin{theorem}
\label{thm:derand-extend}
For every constant $\sparam \in (0,1)$, there is some constant $\gamma$ such that any $T(n,\Delta)$-round extendable local algorithm $\cA$ satisfying that $\Delta^{\gamma \cdot T(n,\Delta)} \leq n^{\sparam}$, in which each node uses at most $\Delta^{O(T(n,\Delta))}$ bits of randomness, can be derandomized\footnote{The deterministic \MPC algorithm in the derandomization relies on the PRG construction in Lemma \ref{lem:prg-alg}, and therefore it may perform \emph{heavy local computations}, or alternatively, this result can be stated as a \emph{non-uniform derandomization} where the PRG function is hard-coded in the \MPC machines, in which case local computation is polynomial.} within $O(\log(T(n,\Delta))+\log^*n)$ \MPC rounds with local space $O(n^\sparam)$ and global space $O(n^{1+\sparam})$.
\end{theorem}

\begin{proof}
Consider a fixed $n$-node graph $G=(V,E)$, with maximum degree $\Delta$, and let $t:=T(n,\Delta)$. First, the \MPC algorithm allocates a separate machine $M_u$ to each node $u$ that stores its $2t$-radius ball in $G$. This can be done in $O(\log t)$ rounds, by the standard graph exponentiation technique. Next, the algorithm computes a $\Delta^{4t}$-coloring in the graph $G^{2t}$. The purpose of this step is to reduce the name space of the nodes from $O(\log N)$ bits into $O(t \log \Delta)$ bits, such that in each $2t$-radius ball, the new IDs of the nodes (i.e., their colors) are unique. This coloring can be implemented within $O(\log^* N)$ deterministic rounds \cite{Kuhn09}.
	
We will simulate $O(1)$ iterations of algorithm $\cA$, each time running on the subset of nodes labeled with $\bot$ in the last iteration. That is, we start on our input graph $G$, and provide $\cA$ with the values $n$ and $\Delta$, and in subsequent iterations provide the current values $n^*\le n$ and $\Delta^* \le \Delta$. By extendability of $\cA$, we know that the resulting labeling can be extended to a valid full solution by relabeling $\bot$ with any valid labeling on their induced graph. It remains only to show that after $O(1)$ iterations, we can ensure that \emph{no} nodes remain labeled $\bot$.
	
In each iteration of $\cA$ we let machine $M_u$ determine the output for node $u$, and we use the PRG construction of Lemma \ref{lem:prg-alg} in order to deterministically fix a good random seed for the algorithm. We need to provide each node with $\Delta^{O(t)}$ random bits, and we do so based on the new IDs (i.e., nodes with the same new IDs receive the same random bits, but since these nodes are of distance at least $2t$ apart, this does not cause any dependency issues). So, our PRG will need to generate $\Delta^{O(t)}$ total pseudorandom bits.
	
By Lemma \ref{lem:prg-alg}, a $(\Delta^{O(t)}, n^{-\eps})$ PRG can be constructed using $e^{O(t\log\Delta+\eps\log n)}$ space. Setting $\eps$ to be a sufficiently small constant, this is $O(n^\sparam)$. Note that we always use this PRG (i.e., with parameters in terms of original number of nodes $n$), and do not update to the current value $n^*$.
	
When we run $\cA$, with a uniformly random seed from the PRG, on a graph of size $n_*\le n$ (since we run only on the induced graphs of subsets of nodes labeled $\bot$), the output at each node is $(\Delta^{O(t)}, n^{-\eps})$ indistinguishable from its output under full randomness. So, the expected number of nodes which output $\bot$ is at most $\frac 12+n_*\cdot n^{-\eps}$. After $O(1/\eps)$ iterations we therefore reduce the expected number of nodes labeled $\bot$ below $1$.
	
Now it remains only to \emph{deterministically} choose a seed from the PRG that achieves this expected value of undecided nodes (and, since the number of undecided nodes is an integer, it must then be $0$). To do so, we use a distributed implementation of the classical method of conditional expectations, a means of having all machines agree on some seed which achieves at most the expected value of some cost measure (in this case, the number of nodes labeled $\bot$). \cite{CDP20a,CDP20b} show how to implement this method, in low-space \MPC, in such a way that $\Theta(\log n)$ bits specifying the seed can be fixed in a single round, provided that the global cost measure is the sum of functions computable locally on single machines. Here, each machine $M_u$ can locally compute the indicator variable for the event $\{\text{$u$ is labeled $\bot$}\}$ under $\cA$ using any particular seed, and the global cost function is the sum (over all nodes $u\in V$) of these variables.
	
Since the seeds for the PRG are $O(t\log\Delta+\eps\log n) = O(\log n)$ bits long, we can perform the method of conditional expectations in $O(1)$ rounds in each iteration, fixing a globally agreed seed from our PRG which ensures that the number of nodes labeled $\bot$ is at most its expectation. Then, after $O(\frac{1}{\eps}) = O(1)$ iterations, we have no remaining nodes labeled $\bot$, and so have a complete valid solution for the problem. The total running time of the deterministic \MPC algorithm is therefore $O(\log t)$ (for initially collecting balls) plus $O(\log^* n)$ (for coloring).
\end{proof}


\subsubsection{Application to maximal independent set and maximal matching}
\label{subsubsec:MIS-MM}

To demonstrate the applicability of this derandomization recipe, we show how it can be used to improve the deterministic running times of two cornerstone problems in low-space \MPC: maximal independent set and maximal matching. The best prior round complexities for both problems is $O(\log \Delta+\log\log n)$ \cite{CDP20a}; here we improve the dependency on $n$ to triple-logarithmic, and surpass the component-stable lower bound for $\Delta=2^{\log^{o(1)} n}$.

\begin{theorem}
\label{thm:MIS}
For any constant $\sparam > 0$, a maximal independent set and maximal matching can be found deterministically in low-space \MPC in $O(\log\log\Delta + \log\log\log n)$ rounds when $\Delta=2^{\log^{o(1)} n}$, with local space $O(n^\sparam)$ and global space $O(n^{1+\sparam})$.
\end{theorem}

\begin{proof}
We focus on maximal independent set, since the results for maximal matching will then follow by reduction.
	
We will use the \LOCAL algorithm of Ghaffari \cite{Ghaffari16}, which, when combined with the polylogarithmic network decomposition of \cite{RG20}, runs in $t = O(\log\Delta + \polyloglog(n))$ rounds, on graphs with $n$ nodes, when provided with the values $n$ and $\Delta$, succeeding (globally) with probability at least $1-\frac {1}{n^2}$. We use the following final labeling: nodes which have been placed in the output independent set are labeled $\mathbb{IN}$, adjacent nodes are labeled $\mathbb{OUT}$, and all other nodes are labeled~$\bot$.
	
Ghaffari's algorithm has the following important property: even when the algorithm probabilistically fails, it never places two adjacent nodes in the output independent set; instead it merely fails to decide the status of some nodes. Therefore, any output labeling produced by Ghaffari's algorithm in this way is extendable to a full solution by any valid output (i.e., a valid MIS) on the induced graph of undecided nodes (those labeled $\bot$).
	
Another property we need is a bound on the number of random bits used by each node. Ghaffari's algorithm uses $O(t \log \Delta)$ random bits per node ($O(\log \Delta)$ per round $i$), since a node $v$'s only random choice each round is to `mark' itself with some probability $p_t(v)$, which is equal to $2^{-k_i}$ for some $k_i\in [\lceil \log \Delta \rceil]$. To perform this choice, $v$ can take $k_i$ random bits and mark itself if they are all $0$.
	
Furthermore, since with probability at least $1-\frac {1}{n^2}$ the algorithm succeeds globally (i.e., labels no nodes $\bot$), the \emph{expected} number of nodes labeled $\bot$ is at most $ \frac{n}{n^2}<\frac 12$.
	
Therefore the algorithm is extendable. We have $\Delta^{O(t)} = 2^{\log^{o(1)} N} \le N^{\sparam}$, and so we can apply Theorem \ref{thm:derand-extend} to solve MIS in $O(\log t) = O(\log\log\Delta +\log\log\log n)$ rounds deterministically in \MPC with local space $O(n^\sparam)$ and global space $O(n^{1+\sparam})$.
	
To perform Maximal Matching, we use standard reduction of finding an MIS on the line graph of the input graph. It is well-known that this corresponds to a maximal matching in the original graph, and in \LOCAL only increases the round complexity by $1$, which means that we still adhere to our \MPC space bounds.	
\end{proof}

\begin{corollary}
\label{cor:thm:MIS}
For any constant $\sparam > 0$, a maximal independent set and maximal matching can be found deterministically in low-space \MPC in $O(\log \Delta + \log\log\log n)$ rounds, with local space $n^\sparam$ and global space $n^{1+\sparam}$.
\end{corollary}

\begin{proof}
We use Theorem \ref{thm:MIS} when $\Delta=2^{\log^{o(1)} n}$; otherwise, we can use the low-space \MPC algorithm of \cite{CDP20a}, with running time $O(\log \Delta+\log\log n) = O(\log \Delta)$.
\end{proof}

\begin{theorem}
\label{thm:MIS-MM-1}
Assuming the connectivity conjecture,
there is no deterministic low-space component-stable \MPC algorithm that computes a maximal matching or maximal independent set in $o(\log\Delta +\log\log n)$ rounds.
\end{theorem}

\begin{proof}
Balliu et al. \cite{BBHORS19} show an $\Omega(\min\{\Delta,\frac{\log N}{\log\log N}\})$-round\footnote{That is, a deterministic \LOCAL algorithm cannot simultaneously have round complexity $o(\Delta)$ and $o(\frac{\log N}{\log\log N})$; algorithms do exist which are faster in one parameter at the expense of the other.} deterministic \LOCAL lower bound for both problems, even when exact values of $n$ and $\Delta$ are known. Using constrained function $T(N,\Delta) = \sqrt{\min\{\Delta,\log N\}}$, we apply Theorem \ref{thm:lb} to obtain an $\Omega(\min\{\log \Delta,\log \log N\})$-round conditional lower bound for component-stable low-space deterministic \MPC algorithms, (directly for maximal independent set, and using the standard conversion to the line graph for maximal matching).
\end{proof}

Hence, Theorem \ref{thm:MIS} surpasses this component-stable conditional lower bound when $\Delta=2^{\log^{o(1)} n}$.

\hide{
	\subsubsection{Maximal matching and maximal independent set}
	\label{subsec:MIS-MM}
	
	In this section we show that existing efficient \LOCAL algorithms for the cornerstone problems of maximal matching and maximal independent set are \emph{extendable}, and we can therefore apply Theorem \ref{thm:derand-extend} to obtain efficient algorithms in deterministic low-space \MPC, which surpass the lower bounds for component-stable algorithms.

	\begin{theorem}
		\label{thm:MIS-MM-2}
		For any $\alpha>0$ there is a deterministic $(\alpha,1+\alpha)$-\MPC algorithm that computes a maximal matching or maximal independent set, in any graph of maximum degree $\Delta=  2^{O(\alpha\sqrt{\log n})}$, in $O(\log\log\Delta + \log\log\log n)$ rounds.
	\end{theorem}
	
	\begin{proof}
		To do - We begin by collecting balls of radius $O(\log\Delta+\log^6\log n)$, via graph exponentiation. We then perform Ghaffari's MIS algorithm, using the method of conditional expectations on our PRG. The PRG guarantees that each node is removed (either by joining the MIS or by one of its neighbors doing so) with probability at least $n^{-\delta}$. By the method of conditional expectations, we can therefore choose a seed which reduces the number of nodes in the graph by an $n^{-\delta}$ factor. Repeating this process $\frac 1\delta = O(1)$ times completes MIS.
		
		We can perform maximal matching in the same asymptotic round complexity by simulating this algorithm on the line graph of the input graph.
	\end{proof}
} 

%% file: PRG.tex
A \emph{Pseudorandom Generator (PRG)} is a function that gets a short random seed and expands it to a long one which look random, in the sense that it is indistinguishable from a random seed of the same length for such a formula.

\begin{definition}[Computational indistinguishability, Definition 7.1 in \cite{Vadhan12}]
\label{def:computational-indistinguishability}
Random variables $X$ and $Y$ taking values in $\{0,1\}^m$ are $(t,\eps)$ \textbf{indistinguishable} if for every non-uniform algorithm $T$ running in time at most $t$, we have $|\Pr[T(X)=1]-\Pr[T(Y)=1]| \le \eps$.
\end{definition}

\begin{definition}[PRG, Definition 7.3 in \cite{Vadhan12}]
\label{def:PRG}
A deterministic function $\mathcal{G}:\{0,1\}^d \to \{0,1\}^m$ is an \textbf{$(t,\eps)$ pseudorandom generator (PRG)} if (1) $d < m$, and (2) $\mathcal{G}(U_d)$ and $U_m$ are $(t,\eps)$ indistinguishable.
\end{definition}

The following proposition uses the probabilistic method to show the existence of good PGRs.

\begin{proposition}[Proposition 7.8 in \cite{Vadhan12}]
\label{prop:perfect-PRG}
For all $m \in \mathbb{N}$ and $\eps>0$, there exists a (non-explicit) $(m,\eps)$ pseudorandom generator $\mathcal{G}: \{0,1\}^d \to \{0,1\}^m$ with seed length $d = \Theta(\log m+\log(1/\eps))$.
\end{proposition}

We next notice that $(m,\eps)$ PRGs of Proposition \ref{prop:perfect-PRG} can be computed by exhaustive search.


\begin{lemma}[Time and space complexity of perfect PRGs]
\label{lem:prg-alg}
%
For all $m \in \mathbb{N}$ and $\eps>0$, there exists an algorithm for computing the $(m,\eps)$ PRG of Proposition \ref{prop:perfect-PRG} in time $\exp(\poly(m/\eps))$ and space $\poly(m/\eps)$.
\end{lemma}

\begin{proof}
Our goal is to find an $(m,\eps)$ PRG from $\{0,1\}^d$ to $\{0,1\}^m$ with seed length $d = \Theta(\log m+\log 1/\eps)$, and the existence of one such PRG follows from Proposition \ref{prop:perfect-PRG}.

The procedure for computing the PRG iterates over all different functions $\mathcal{G}: \{0,1\}^d \to \{0,1\}^m$ according to some fixed order. For each fixed function $\mathcal{G}: \{0,1\}^d \to \{0,1\}^m$, it iterates over all non-uniform algorithms $\cA$ running in time $m$ according to some fixed order, as well. To examine if $\mathcal{G}$ $\eps$-fools $\cA$ (in the sense of Definition \ref{def:PRG}), it is required to compare the output distribution of $\cA$ run on a random $m$-bit string to the $m$-bit pseudo-random string obtained by evaluating the function $\mathcal{G}$ on a random $d$-bit vector. This is done by computing the output of algorithm $\cA$ under each $\mathcal{G}(x)$ for every $x \in \{0,1\}^d$, as well as evaluating the algorithm $\cA$ under each $y \in \{0,1\}^m$. The procedure picks the PRG function $\mathcal{G}^*:\{0,1\}^d \to \{0,1\}^m$ that $\eps$-fools all $m$-time algorithms; the existence of one such $\mathcal{G}^*$ follows from Proposition \ref{prop:perfect-PRG}.

We next bound the time and space complexity of this procedure. There are $(2^m)^{2^d} = 2^{O(m2^d)}$ different functions $\mathcal{G}$ mapping $d$ bits into $m$ bits, and there are $2^{O(m \log m)}$ non-uniform algorithms running in time $m$ (i.e., Boolean circuits of size $m$). The time complexity is bounded by $2^{O(m2^d)} \cdot 2^{O(m \log m)} \cdot 2^d \cdot 2^d = 2^{\poly(m/\eps)}$.
%

We next bound the space complexity. To iterate over all the possible PRG candidates $\mathcal{G}: \{0,1\}^d \to \{0,1\}^m$, one needs to store the index of the current candidate function and its representation, which uses $O(m \cdot 2^d)$ bits. To iterate over all the $m$-time algorithms according to some fixed order, one needs to store the index of the current $\cA$, and its representation which can be done with $O(m \log m)$ bits. Fixing the candidate PRG function to $\mathcal{G}$ and the given $m$-algorithm to $\cA$, the evaluation of the algorithm $\cA$ under each $\mathcal{G}(x)$ for every $x \in \{0,1\}^d$ can be done using $O(m \cdot 2^d)$ space. Finally, in order to evaluate $\cA$ under each $y \in \{0,1\}^m$, it is sufficient to store the current index of the vector $y$ considered. Altogether, the space requirements is $O(m \cdot 2^d)$;
since $d = \Theta(\log m+\log(1/\eps))$, this is $\poly(m/\eps)$.
\end{proof}

%% file: sepRandMPC.tex


\hide{The framework developed in \Cref{sec:prelim} and culminating in \Cref{thm:lb} (see also \cite{GKU19}) is a host of conditional lower bounds for component-stable \MPC algorithms for a number of natural graph problems. The restriction of an \MPC algorithm being component-stable seems like an unimportant artifact of this framework, and in fact, as mentioned in Introduction, Ghaffari et al.\ \cite{GKU19} wrote \emph{``To the best of our knowledge, all known algorithms in the literature are component-stable or can easily be made component-stable with no asymptotic increase in the round complexity.''}}

In this section, we demonstrate the existence of a natural problem for which there is a (conditional) gap between \emph{randomized} component-stable and component-unstable algorithms. 
This will give us a proof of \Cref{thm:approx-IS}.

We consider the task of computing large independent sets in $n$-node graphs with maximum degree $\Delta$. Recently, \cite{KKSS19} has shown that for any $n$, there exists $n$-node graphs with maximum degree $\Delta = \Omega(n/\log n)$, for which any randomized \LOCAL algorithm for computing an independent set of size $\Omega(n/\Delta)$ with success probability $1-\frac 1n$  (in fact, even reaching a weaker success guarantee of $1-\frac {1}{\log n}$), requires $\Omega(\log^*n)$ rounds. (Here, and throughout this section, we assume that $\Delta\ge 1$ so that the problem is well-defined.) Since the cardinality of the maximum independent set in their lower bound graph can be bounded by $O(n/\Delta)$, their result can also be stated as a lower bound for computing a constant approximation for the maximum independent set. We start by providing a mild adaptation to the lower bound proof of \cite{KKSS19} so that it would fit the framework from \Cref{sec:prelim} and from \cite{GKU19}. In particular, to be able to lift this \LOCAL lower bound into the \MPC setting from \Cref{thm:lb} in \Cref{sec:prelim} and from \cite{GKU19}, it is required for the lower bound to hold even if nodes are equipped with shared randomness, and with an estimate $N$ on the number of nodes $n$ which is at least more than a $\log N$-factor loose, i.e., nodes cannot distinguish between $n=\Theta(N/\log N)$ and $n=\Theta(N)$. We begin with the following fact implicit in Corollary~V.4 of \cite{GKU19}.

\begin{fact}[\cite{GKU19}]
\label{fc:MIS-ring-LB}
Any randomized \LOCAL algorithm to compute a maximum independent set in $n$-node cycle graph requires $\Omega(\log^*n)$ rounds; this holds even if the nodes know $n$ (and, naturally, know $\Delta = 2$) and even if they have access to an unlimited amount of shared randomness.
\end{fact}

We can then plug this lower bound, strengthened to hold against shared randomness, into the result of \cite{KKSS19}, yielding \Cref{thm:approxIS-LB}. The proof is identical to Theorem 4 of \cite{KKSS19}; it is unaffected by the change to shared randomness.

\begin{theorem}[\cite{KKSS19}]
\label{thm:approxIS-LB}
Any randomized \LOCAL algorithm to compute an independent set of size $\Omega(n/\Delta)$ in all $n$-node graphs (i.e., over the full range of $\Delta\in [1,n]$) requires $\Omega(\log^* n)$ rounds. This holds even if the nodes have the exact maximum degree $\Delta$, access to an unlimited amount of shared randomness, and an input size estimate $N$ of $n$ which is more than a $\log N$-factor loose, i.e., they cannot distinguish between $n=\Theta(N/\log N)$ and $n=\Theta(N)$.
\end{theorem}

\hide{
\begin{proof}[Proof Sketch]
The lower bound of \cite{KKSS19} is based on the $\Omega(\log^* n)$ lower bound for maximum independent set for the ring by Naor \cite{Naor91}. Since it is unclear if the lower bound in \cite{Naor91} also holds with shared randomness, instead we will be using the lower bound from Fact \ref{fc:MIS-ring-LB}. For the sake of completeness, we sketch that argument of \cite{KKSS19}.

Let $\mathcal{B}$ be a randomized algorithm for computing an independent set of size at least $n/(c\Delta)$ with probability $1-\frac 1N$ and using $T(N)$ rounds, when given input size estimate $N$ with $n\le N \le poly(n)$, and when all nodes have random IDs in $[poly(N)]$. We will show that this implies that there is a $O( T(n_0))$-round algorithm $\cA$, that with success probability at least $1 -\frac {1}{n_0}$ outputs a maximum independent set on $n_0$-node cycle graph, when provided with the exact value of $n_0$ (and maximum degree $\Delta = 2$).

Let $C$ be a cycle graph of $n_0$ nodes with random IDs, and let $C_1$ be a cycle of $n_0$ many $n_1$-cliques, for some $n_1$ with $\sqrt{n_0}\le n_1 \le n_0$, where neighboring cliques are connected as bicliques. Algorithm $\cA$ computes a maximum independent set on $C$ in the following manner. First, each node in $C$ simulates an $n_1$-clique (i.e., the overall simulated graph is $C_1$). Each node $v$ with ID $i_v$ uses a distinct set of $n_1 \cdot O(\log n)$ bits in the shared seed in order to assign the $n_1$ nodes in its simulated clique random IDs in $[poly(n_0,n_1)]$. At this point, the nodes in $C_1$ have random IDs, which are unique with high probability. Next, algorithm $\cA$ applies algorithm $\mathcal{B}$ to compute a large independent set $I'$ in the simulated graph $C_1$, providing it with input size estimate $N= c\cdot n_0\cdot n_1$. Since there is a one-to-one mapping between independent sets in $C_1$ and in $C$, this yields an independent set $I''$ in $C$. The maximum independent set computation is then completed by filling in the gaps between incident $I''$ nodes.

The key challenge is in showing that the round complexity of the algorithm is $O(T(N))$. This claim is shown by exploiting the fact that $\mathcal{B}$ is an $T(N)$-round distributed algorithm, and its output at a node $v$ must be the same whether the overall graph is $C_1$ or a $O(T(N))$-length cycle of $n_1$-cliques, since these cases are indistinguishable to $\mathcal B$ at $v$. The input size estimate $N$ satisfies both $n_0\cdot n_1 \le N \le poly(n_0\cdot n_1)$ and $T(N)\cdot n_1 \le N \le poly(T(N)\cdot n_1)$, so would be valid in either case. (The indistinguishability argument would fail, however, if $\mathcal{B}$ knew the exact number of nodes in its input graph).

By the approximation guarantee of $\cA$ it then holds that the $O(T(N))$-neighborhood of each node must contain at least one marked node (see \cite{Naor91}), and thus the gap between marked independent set nodes (output by Alg. $\mathcal{B}$) is indeed bounded by $O(T(N))$, as required. $\cA$ can then fill in these gaps and reach maximum independent set as required, in $O(T(N))$ rounds, with success probability at least $1-\frac{1}{T} \ge 1-\frac{1}{n_0}$.
\end{proof}}

\hide{This seems incomplete - it's not clear to me where the the $\Delta = \Theta(n/\log\log n)$ from the statement comes from}

\hide{
\begin{proof}[Proof Sketch]
	The lower bound of \cite{KKSS19} is based on the $O(\log^* n)$ lower bound for maximum independent set for the ring by Naor \cite{Naor91}. Since it is unclear if the lower bound in \cite{Naor91} also holds with shared randomness, instead we will be using the lower bound from Fact \ref{fc:MIS-ring-LB}. For the sake of completeness, we sketch that argument of \cite{KKSS19}.
	
	Let $\mathcal{B}$ be a randomized algorithm for computing an independent set of size at least $n/(c\Delta)$ with probability $1-\frac 1N$ and using $T(N)$ rounds, when all nodes have random IDs. We will show that this implies that for every integer $n_1$, there is a $O(c \cdot T(n_0 \cdot n_1))$-round algorithm $\cA$ that with probability of $1 - n_0 \cdot p(n_1)$ outputs a maximum independent set on $n_0$-node cycle graph.
	
	Let $C$ be a cycle graph of $n_0$ nodes with random IDs, and let $C_1$ be a cycle of $n_0$ many $n_1$-cliques where neighboring cliques are connected as bicliques. Algorithm $\cA$ computes a maximum independent set on $C$ in the following manner. First, each node $v$ with ID $i_v$ uses a distinct set of $n_1 \cdot O(\log n)$ bits in the shared seed in order to assign the $n_1$ nodes in its clique random IDs. At this point, the nodes in $C_1$ have random IDs. Next, algorithm $\cA$ applies algorithm $\mathcal{B}$ to compute a large independent set $I'$ in $C_1$. Since there is a one-to-one mapping into independent sets in $C$, this yields an independent set $I''$ in $C$. The maximum independent set computation is then completed by filling in the gaps between incident $I''$ nodes. The key challenge is in showing that the round complexity of the algorithm is $O(c \cdot T(n_0 \cdot n_1))$. This claim is shown by exploiting the fact that $\mathcal{B}$ is an $T(n_0 \cdot n_1)$-round distributed algorithm, and thus a node $v$ running $\cA$ cannot distinguish between $C_1$ and a $O(c \cdot T(n_0 \cdot n1))$-length cycle of cliques containing $v$, so long as the input size estimate $N$ we provide satisfies both $n\le N\le n^k$ and $n\le N\le n^k$ (which we can ensure, if $k$ is a sufficiently large constant). It is crucial to observe that this indistinguishability argument still holds even if the nodes of $C$ are given a polynomial estimate on $n_0$ (but it would fail upon knowing exactly $n$), and with shared random seed. Specifically, since $n_1$ is fixed, the nodes that simulate $\cA$ have an estimate of $\poly(n_0) \cdot n_1$ on the size of $C_1$ and therefore the algorithm indeed cannot distinguish between these cases. By the approximation guarantee of $\cA$ it then holds that the $O(c \cdot T(n_0 \cdot n1))$-neighborhood of each node must contain at least one marked node, and thus the gap between marked independent set nodes (output by Alg. $\mathcal{B}$) is indeed bounded by $O(c \cdot T(n_0 \cdot n1))$, as required.
\end{proof}}

We can now combine \Cref{thm:approxIS-LB} and Lemma \ref{lem:IS-replicable} with the lifting argument in \Cref{thm:lb}, using the constrained function $T(N,\Delta) := \log^* N$,  to obtain the following lower bound for component-stable \MPC algorithms.

\begin{lemma}[Super-constant lower bound for component-stable IS]
\label{lemma:stableMPC-LB}
Assuming that the connectivity conjecture holds, there is no $o(\log\log^* n)$-round low-space component-stable \MPC algorithm that computes an independent set of size $\Omega(n/\Delta)$, in all graphs with $n$ nodes and $\Delta\in [1, n)$, with success probability at least $1-\frac 1n$.
\end{lemma}

Finally, we show that there is a very simple (component-unstable) \MPC algorithm for computing large independent sets in a constant number of rounds.

\hide{
	
	 This algorithm in fact computes an independent set of cardinality $\Omega(\mathbb T(G))$, where $\mathbb T(G)$ is the Turan bound given by $\mathbb T(G) = \sum_v 1/(d_v+1)$, and $d_v$ is the degree of node $v$ in $G$. Notice that $\mathbb T(G) \ge \frac{n}{\Delta+1}$.}

\paragraph{$O(1)$-round randomized algorithm.}

\hide{
We first give a $O(1)$-round randomized algorithm that computes an independent set of $\Theta(\mathbb T(G))$ nodes, in expectation. By running $c \log n$ independent repetitions of this step (simultaneously in parallel) for some sufficiently large constant $c$, we get an independent set of cardinality $\Theta(\mathbb T(G))$ with high probability. First, each node computes its degree, which can be done in $O(1)$ rounds. The algorithm then consists of a single Luby's step (cf. \cite{Luby86}), where each node $v$ picks a number $\chi_v \in [0,1]$ uniformly at random. A node $v$ joins the independent set if $\chi_v < \chi_u$ for every neighbor $u$ of $v$. This can be verified in $O(1)$ rounds in the low-space \MPC setting. We next claim that in expectation at least $\mathbb T(G)/3$ nodes join the independent set. To see this observe that the probability that a node $v$ has the minimum $\chi_v$ value in its neighborhood is at least $1/d_v \cdot (1-1/d_v)^{d_v} \ge 1/(3d_v)$. Thus, in expectation, the number of nodes that joins the independent set is at least $1/3 \cdot \sum_v 1/d_v \ge \mathbb T(G)/4$. By Markov inequality, it holds that w.h.p at least one of the $c \log n$ independent repetitions computes an independent set with at least $\mathbb T(G)/8$ nodes.}

We first give a $O(1)$-round randomized algorithm that computes an independent set of $\Theta(n/\Delta)$ nodes, in expectation. By running $c \log n$ independent repetitions of this step (simultaneously in parallel) for some sufficiently large constant $c$, we get an independent set of cardinality $\Theta(n/\Delta))$ with high probability. First, each node computes its degree, which can be done in $O(1)$ rounds. The algorithm then consists of a single step of Luby's algorithm (see \cite{Luby86}), where each node $v$ picks a number $\chi_v \in [0,1]$ uniformly at random. A node $v$ joins the independent set if $\chi_v < \chi_u$ for every neighbor $u$ of $v$. This can be verified in $O(1)$ rounds in the low-space \MPC setting. We next observe that the probability that a node $v$ has the minimum $\chi_v$ value in its neighborhood is at least $1/(d_v+1) \ge 1/(\Delta+1)$. Thus, in expectation, the number of nodes that joins the independent set is at least $n/\Delta$. By Markov inequality, it holds that w.h.p. at least one of the $c \log n$ independent repetitions computes an independent set with at least $ n/2(\Delta+1)$ nodes.


\paragraph{$O(1)$-round deterministic algorithm.}

We slightly change the algorithm described above so that it would work with pairwise independence, and hence to become deterministic.

\begin{claim}
Consider a simulation of a single step of Luby's algorithm with pairwise independence. Then, the expected number of independent set nodes is $n/(4\Delta+1)$.
\end{claim}

\begin{proof}
For a node $v$ we consider the event that $\chi_v < 1/(2\Delta)$ and that $\chi_u \ge 1/(2\Delta)$ for every neighbor $u$ of $v$. Note that if this event occurs then $v$ joins the independent set. To bound the probability of this event for node $v$, we first bound the probability that $\chi_u \ge 1/(2\Delta)$ for every neighbor $u$ \emph{conditioned} on the event that $\chi_v < 1/(2\Delta)$. Due to pairwise independence, for every neighbor $u$, it holds that $\chi_u \le 1/(2\Delta)$ with probability $1/2\Delta$ even when conditioning on $\chi_v < 1/(2\Delta)$. Thus, conditioned on $\chi_v < 1/(2\Delta)$, by the union bound, the probability that $v$ has some neighbor $u'$ with $\chi_{u'} \le 1/2\Delta$ is at most $1/2$. Overall, the probability for the event to hold is $1/4\Delta$.
\end{proof}

This step can be derandomized by applying $O(1)$ steps of graph sparsification, exactly in the same manner as done for the derandomization of Luby's step in the maximal independent set algorithm of \cite{CDP20a}.
We therefore have the following:

\begin{theorem}
\label{thm:stableMPC-LB}
There is a deterministic low-space \MPC algorithm that in $O(1)$ rounds computes an independent set of size $\Omega(n/\Delta)$, in all graphs on $n$ nodes with $\Delta = [1,n]$.
%
\end{theorem}

\begin{proofs}
Let $\delta>0$ be a constant sufficiently smaller than $\sparam$. If $\Delta > n^\delta$, then by the deterministic sparsification framework of \cite{CDP20a}, we can derandomize the subsampling of each node with probability $\frac{n^\delta}{\Delta}$, using $O(1)$-wise independent hash functions coupled with the method of conditional expectations, in such a way that:

\begin{itemize}
\item The number of nodes in the resulting sparsified graph is $\Theta(\frac{n^{1+\delta}}{\Delta})$;
\item The maximum (induced) degree in the sparsified graph is $O(n^{\delta})$.
\end{itemize}	

Since degrees are now low enough to fit $2$-hop neighborhoods onto single machines (as $\delta$ is sufficiently lower than $\sparam$), we can now derandomize a step of Luby's algorithm on this subsampled graph (or on the original graph if $\Delta\le n^\delta$) exactly as in \cite{CDP20a}. In doing so, we find an independent set of size $\Omega(\frac{n^{1+\delta}}{\Delta}/{n^{\delta}})= \Omega(\frac{n}{\Delta+1})$.
\end{proofs}
\hide{


PETER - I had to take out this section, because	I don't think it works under our definition of \LOCAL algorithms. In particular, the proof only works if success probability depends on the number of nodes $n$ in the input graph, but we've defined \LOCAL algorithms to have the success probability depend on the input size estimate $N$ (which I think is the right definition, since their \emph{behavior} only depends on $N$ and not $n$). In this case, as far as I can see, the lemma isn't true.
	

\subsection{A barrier for exploiting instability for \LCL problems}

We then turn to ask whether one can exploit the same boosting effect to break the conditional lower bounds for \LCL problems as well. Interestingly, we show (see \Cref{sec:stable-nonstable-rand}) that unlike the problem of approximating maximum independent set, there are no improved \LOCAL randomized algorithms for any \LCL problem when relaxing the global success guarantee from $1 - \frac1N$ to $\frac1N$ for some $N = \poly(n)$.

\begin{lemma}
	\label{lem:rand-barrier}
	Let $\mathcal{P}$ be an \LCL problem that for some family of $n$-node graphs $\mathcal{G}$ requires $\Omega(T(N))$ rounds for any randomized \LOCAL algorithm that succeeds with probability at least $1-\frac1N$, where $N = \poly(n) \ge n$ and $T(N) = o(\log N)$. Then for any constant $c$, any randomized \LOCAL algorithm that succeeds with probability of $\frac{1}{N^c}$ on any $n'$-node graphs with $N^2 \ge n'$ requires $\Omega(T(N))$ rounds.
	\Peter{This lemma, and its proof, needs slightly updating to meet our definition/notation of \LOCAL algorithms (Section \ref{subsubsec:LOCAL-algs})}
	\Artur{Peter: could you please try to fix it?}
	%
\end{lemma}

\Cref{thm:approxIS-LB,thm:stableMPC-LB} show that there are $O(1)$-replicable problems where low-space component-unstable \MPC algorithms can provably help and beat known conditional lower bounds for low-space component-stable \MPC algorithms. Next, we show that such separation is more difficult to achieve for \LCL problems. We consider Lemma~\ref{lem:rand-barrier} and provide an indication that component-instability might not help randomized \LCL algorithms; at least not by means of boosting the success guarantee as provided for the approximate maximum independent set problem. In particular, we show that relaxing the success guarantee of the randomized local algorithm on $n$-node graphs to be merely $1/n^c$, instead of the high probability bound of $1-1/n^c$, cannot improve the \LOCAL complexity of the problem. This rules out the possibility of obtaining improved component-unstable \MPC algorithms by means of simulating such a relaxed \LOCAL algorithm (and consequently boosting the success guarantee into $1-1/n^c$ by applying parallel repetitions of this simulation).

\begin{proof}[Proof of Lemma \ref{lem:rand-barrier}]
The proof is by contradiction. Assume that there is an algorithm $\cA$ that on every $n'$-node graph $G$ succeeds in $t=o(T(N'))$ rounds with probability at least $\frac{1}{(N')^{c/2}}$, with the estimate $N' = N^2$ and $n' \in [n, n^2 \cdot N \cdot \log N]$. We will show that this implies that $\cA$ solves the problem $\mathcal{P}$ with probability at least $1-1/N$ on every $n$-node graph $G \in \mathcal{G}$ as well, thus leading to a contradiction.

For every graph $G \in \mathcal{G}$ and for every $v \in V(G)$, define the graph $G^v$ to be made of $\ell = c \cdot n \cdot N \cdot \log N$ disconnected copies of the $2t$-radius ball of $v$ in $G$, namely, $B_{2t}(v,G)$. Let $p_v$ be the probability that $v$'s output is correct when running $\cA$ on $G$. Note that this probability depends only on the $2t$-radius ball of $v$ in $G$. It then holds that in each of the $\ell$ copies of $G_v$, the centered node (copy of $v$) is correct in $G^v$ with probability of $p_v$ as well.

Note that $|V(G_v)|\in [n,n \cdot \ell]$ and thus $N^2$ is indeed a polynomial estimate for $|V(G_v)|$. Since $\cA$ succeeds on $G_v$ with probability of $1/N^c$, and since the outputs of each $v$'s copy are independent, we get that $p_v^{\ell} \ge 1/N^c$, concluding that $p_v \ge 1-1/(nN)$. We therefore get that each node $v$ succeeds in the original $G$ with probability of $1-1/nN$, and thus by the union bound, $\cA$ succeeds on $G$ with probability of at least $1-1/N$. We conclude that $t=\Omega(T(N))$ as desired.
\end{proof}

\begin{lemma}
	\label{lem:rand-barrier}
	Let $\mathcal{P}$ be an \LCL problem that for some family of $n$-node graphs $\mathcal{G}$ requires $\Omega(T(N))$ rounds for any randomized \LOCAL algorithm that succeeds with probability at least $1-\frac1N$, where $N = \poly(n) \ge n$ and $T(N) = o(\log N)$.
	
Assume that there is a \LOCAL algorithm $\cA$ that that solves an $n$-radius checkable problem $\mathcal{P}$, on $n$-node graphs, given exact knowledge of $n$ (this is a weaker algorithm, and therefore a weaker assumption, than one that works with a polynomial estimate of $n$), in $T(n)$ rounds with success probability $\frac{1}{n^c}$. Then, there is a \LOCAL algorithm $\mathcal B$ also solves $\mathcal{P}$, on $n$-node graphs, given knowledge of input size estimate $N$ with $n\le N \le n^C$ (for some constant $C$), with success probability $\frac{1}{N}$.
	
	Then for any constant $c$, any randomized \LOCAL algorithm that succeeds with probability of $\frac{1}{N^c}$ on any $n'$-node graphs with $N^2 \ge n'$ requires $\Omega(T(N))$ rounds.
	\Peter{This lemma, and its proof, needs slightly updating to meet our definition/notation of \LOCAL algorithms (Section \ref{subsubsec:LOCAL-algs})}
	\Artur{Peter: could you please try to fix it?}
\end{lemma}

\begin{proof}[Proof of Lemma \ref{lem:rand-barrier}]
	
To determine its output on a node $v$, $\mathcal B$ collects the $T(VVVV)$-radius ball around $v$. It then simulates $\cA$ on the graph given by this ball with $VVV$ added disconnected nodes. $\cA$ returns a

Consider the outputs of $\cA$ at a node $v$ in $n$-node graph $G$, and in a larger graph $G'$ consisting of $G$ with $n^c-n$ extra disconnected nodes, when in both cases it is given input size estimate $N \in [n^{C-c},n^C]$ (which VVV). These two cases are indistinguishable from nodes in $G$, under any \LOCAL algorithm, the the output on (the copy of) $G$ must be the same in both cases. Since the \emph{validity} of this (portion of the) output depends only the topology and IDs of $G$, and $\cA$ succeeds on $G'$ with probability

\end{proof}}

%% file: general-derand-arg.tex


The conditional lifting arguments for \emph{component-stable} \MPC algorithms in \Cref{thm:lb} imply that for some ($O(1)$-replicable, and hence, e.g., \LCL{}s) graph problems there is a (conditional) exponential gap between randomized and deterministic algorithms, e.g., for problems that admit such an exponential gap in the \LOCAL\ model \cite{CQW+20,CKP19}. In this section we provide a proof of \Cref{thm:complexity-Det-Ran} that in the \MPC model with polynomially many machines, no such gap exists from a complexity perspective. That is, we show that any randomized low-space \MPC algorithm with round complexity $T(n,\Delta)$ and which uses a polynomial number of machines, can be transformed into a deterministic \MPC algorithm that runs in $O(T(n,\Delta))$ rounds and also uses a polynomial number of machines. This deterministic algorithm is component-unstable, and it is also non-uniform and non-explicit.

We begin with an auxiliary lemma that adapts the approach for \LOCAL algorithms from \cite{CKP19} (see also \cite{GK19}) to \MPC algorithms. (Whereas in \cite{CKP19,GK19} this claim holds only for \LCL problems, in our case we will use it in Lemma \ref{lem:det-alg-many-machines} for any problem whose solution can be verified efficiently in the \MPC setting.)

\begin{lemma}[Implicit in \cite{CKP19}]
\label{lem:det-large-prob}
Let $\cA$ be (a possibly non-uniform) randomized \MPC algorithm that solves a graph problem $\cP$ on $n$-node graphs with maximum degree $\Delta$ in $T(n,\Delta)$ rounds, with probability at least $1-2^{-n^2}$. Then, there is a \emph{non-uniform}, \emph{non-explicit} \textsf{deterministic} \MPC algorithm $\cA'$ that solves $\cP$ on $n$-node graphs with maximum degree $\Delta$ in $O(T(n,\Delta))$ rounds.
%

The local and global space of algorithm $\cA'$ is asymptotically the same as that of~$\cA$.
\end{lemma}

\begin{proof}[Proof Sketch]
As we defined it in \Cref{subsubsec:MPC-algs}, the randomized \MPC algorithm $\cA$ has access to shared randomness $\mathcal{S}$ of polynomial length. Once the string $\mathcal{S}$ is fixed, $\cA$ is deterministic.
Let $\mathcal{G}_{n,\Delta}$ be the family of graphs with at most $n$ nodes and maximum degree $\Delta$. Notice than $|\mathcal{G}_{n,\Delta}| \le 2^{n^2}$. Therefore, if we run $\cA$ on each seed and each possible seed fails on at least one graph in $\mathcal{G}_{n,\Delta}$, the success guarantee of the algorithm $\cA$ cannot be better than $1-1/2^{n^2}$. Therefore, since $\cA$ succeeds with probability at least $1-1/2^{n^2}$, there must be at least one seed $\mathcal{S}^*$ that when provided to the \MPC algorithm, the algorithm is \emph{correct for every graph} in $\mathcal{G}_{n,\Delta}$.
%

This gives us a deterministic algorithm $\cA'$, which uses the seed $\mathcal{S}^*$ to solve $\mathcal{P}$ on every graph in $\mathcal{G}_{n,\Delta}$. Note that the resulting deterministic algorithm is non-explicit in the sense that the desired seed should be hard-coded into the machines. Furthermore, it is non-uniform in the sense that a different seed is hard-coded for each $n$, and so we use a different algorithm for each $n$ hard-code seeds for all possible $n$ into the same algorithm (in order to fit within our $poly(n)$ global space bound).
\end{proof}

\junk{
\begin{remark}
As in \cite{CKP19}, 
we notice that Lemma \ref{lem:det-large-prob} works equally when the complexity $T$ is possibly a function of other quantitative global graph parameters, and so it may depends on measures of local sparsity, arboricity/degeneracy, or neighborhood growth.
\end{remark}
}

\begin{lemma}
\label{lem:det-alg-many-machines}
Let $\cA$ be (a possibly non-uniform) randomized \MPC algorithm that solves a graph problem $\cP$ on $n$-node graphs with maximum degree $\Delta$ in $T(n,\Delta)$ rounds, with probability at most $1-\frac1n$ using $n^\alpha$ local space and $n^{\beta}$ global space. In addition, assume that the correctness of the algorithm can be checked in $T(n,\Delta)$ \MPC rounds, using $n^\alpha$ local space and $n^{\beta}$ global space. Then, there exists a component-unstable, non-uniform, non-explicit deterministic \MPC algorithm $\cA'$ that solves $\cP$ on $n$-node graphs with maximum degree $\Delta$ in $O(T(n,\Delta))$ rounds. The local space of $\cA'$ is $n^{\alpha}$ and the global space is $\tilde O(n^{2+\beta})$.
%
\end{lemma}

\begin{proof}
We first boost the success guarantee of $\cA$ by running multiple parallel simulations: we will run $\ell = O(n^2)$ simulations of algorithm $\cA$ in parallel, using a distinct set of machines for each simulation. We can determine the correctness of all simulations in $T(n,\Delta)$ \MPC rounds, and the final output is determined by choosing an arbitrary correct simulation, if one exists. The probability that the algorithm fails in all these simulations is at most $(\frac1n)^{\ell} = 2^{-\omega(n^2)}$. Thus, if we incorporate this algorithm in Lemma \ref{lem:det-large-prob}, there exists a non-uniform, non-explicit deterministic \MPC algorithm $\cA'$ that solves the problem in $T(n,\Delta)$ rounds using the same local space as $\cA$ and $\tilde O(n^{2+\beta})$ global space.

Let us finally mention that the obtained algorithm $\cA'$ is component-unstable, since it relies on globally agreeing on the outcome of all simulations.
\end{proof}

Now Lemma \ref{lem:det-alg-many-machines} yields the proof of \Cref{thm:complexity-Det-Ran}; informally, $\MPCDetNonStable = \MPCRandNonStable$ holds in the regime of polynomially many machines and for non-uniform, non-explicit low-space \MPC algorithms. 

%% file: conclusions.tex


In this paper, we investigate the power of component-instability for solving local graph problems in the low-space \MPC model. Our main conclusion is that component instability is useful mainly in two (possibly related) aspects: amplification of the success guarantee and derandomization. In the context of randomized algorithms, it allows one to boost the success guarantee of the algorithm. This appears to be useful especially for approximation problems (e.g., maximizing or minimizing a subset of edges or vertices with a given property). As we note, it is unlikely to help in reducing the complexity for \LCL problems. In the context of derandomization, it allows one to efficiently simulate the randomized local algorithm while globally foraging for a short seed. Amplification and derandomization are both obtained by a global computation regardless of the connected components of the graph. A major open problem left by this work is to provide conditional hardness results for graph problems (say, first, for \LCL problems) that hold also for unstable algorithms.
